\documentclass{article}
\usepackage{graphicx, xcolor, amsmath, amssymb, bm, appendix, amsthm, proof, enumerate}
\usepackage{hyperref}
\usepackage{algorithmic}
\usepackage{algorithm}
\usepackage{color}

\theoremstyle{plain}
\newtheorem{thm}{Theorem}
\newtheorem{cor}{Corollary}
\newtheorem{lem}{Lemma}
\newtheorem{dfn}{Definition}

\title{Spatio-temporal coarse-grained Hawkes processes}
\author{Tomoya Uda\thanks{The Graduate University for Advanced Studies, SOKENDAI}, Shinsuke Koyama\thanks{The Institute of Statistical Mathematics}}
\date{\today}

\begin{document}

\maketitle
\begin{abstract}
Spatio-temporal Hawkes processes are widely used to model self-exciting event data, but conventional inference methods require access to the occurrence time and location of every event. In many applications, however, observations are available only as counts aggregated over temporal intervals and spatial regions. We propose a spatio-temporal coarse-grained Hawkes process, a discrete-time count model that approximates aggregated observations generated by an underlying spatio-temporal Hawkes process. The proposed model constructs effective excitation kernels by averaging the temporal and spatial triggering effects over spatio-temporal bins while explicitly incorporating excitation occurring within the same temporal bin. This formulation enables direct modeling of aggregated count data without introducing latent event times or locations. 
We characterize the first- and second-order properties of the proposed process and derive asymptotic approximation errors for the stationary mean and autocovariance relative to those of the corresponding aggregated Hawkes process in the joint limit of fine temporal and spatial discretization.
The analysis shows that the proposed model provides a higher-order approximation than the conventional binned Poisson approximation by accounting for intra-bin excitation. Numerical experiments confirm the theoretical results and demonstrate accurate approximation over a broad range of aggregation scales. We further develop a moment-based estimation procedure and apply the proposed framework to earthquake occurrence data through an aggregated Epidemic-Type Aftershock Sequence (ETAS) model.
The results indicate that the proposed approach provides a computationally efficient alternative for inference and prediction from spatio-temporally aggregated event data.
\end{abstract}

\section{Introduction}
Many phenomena in nature and society exhibit dependence across both space and time. In a wide range of applications, the occurrence of an event increases the likelihood of subsequent events in its temporal and spatial neighborhood. The Hawkes process \cite{hawkes1971spectra,hawkes1971point} provides a natural point-process framework for describing such self-exciting dynamics. Spatio-temporal Hawkes processes have been extensively used to characterize event patterns in seismology \cite{adamopoulos1976cluster,ogata1988statistical,ogata1998space}, criminology \cite{mohler2011self,zhuang2019semiparametric}, and traffic incident analysis \cite{kalair2021non}, while temporal Hawkes processes have found applications in neurophysiology \cite{chornoboy1988maximum,pernice2011structure}, finance \cite{bacry2015hawkes,hawkes2018hawkes,hewlett2006clustering}, and social media analytics \cite{fox2016modeling,kobayashi2016tideh,koyama2020statistical}.

A fundamental requirement of the Hawkes process is access to the occurrence time and location of each event. In many applications, however, such fine-grained information is unavailable or impractical to process. Instead, observations are often provided as counts aggregated over temporal intervals and spatial regions: for example, the number of earthquakes recorded in a geographic area during a given period. Such aggregation induces a loss of information about the precise configurations of events, including the temporal and spatial relationships that generate self-excitation. Developing statistical models that preserve the essential dependence structure of the underlying Hawkes process while operating directly on aggregated counts is therefore an important methodological problem.

For temporal aggregation, the process obtained by counting events from a continuous-time Hawkes process within fixed time intervals is referred to as a binned Hawkes process. Since the binned process is generated by first realizing the continuous-time events and then aggregating them, it is not itself a standard discrete-time Hawkes process, and direct likelihood-based inference is generally unavailable. Existing approaches can broadly be divided into two classes. The first treats the unobserved event configurations within each bin as latent variables. Examples include Monte Carlo expectation-maximization (MC-EM) methods \cite{shlomovich2022parameter,shlomovich2022parameter2}, Pseudo-Marginal Metropolis-Hastings (PMMH) \cite{chen2025}, and Bayesian inference methods \cite{zhou2025bayesian}. These approaches can retain detailed information about the underlying Hawkes dynamics, but repeated sampling or imputation of latent events can be computationally expensive, particularly for large data sets.

The second class constructs a discrete-time model directly for the observed counts. Examples include the binned Poisson approximation and integer-valued autoregressive (INAR) models \cite{shlomovich2022parameter,kirchner2016hawkes}. Such approaches offer substantial computational advantages because their likelihoods or objective functions have relatively simple forms. However, standard discretization approximations can fail to capture intra-bin excitation, whereby an event occurring within a bin triggers additional events during the same bin. This effect can be substantial when the aggregation interval is not sufficiently short. The temporal coarse-grained Hawkes process proposed in \cite{koyama2025coarse} addresses this limitation by constructing an effective discrete-time excitation kernel that incorporates intra-bin excitation. In particular, it was shown to provide a higher-order approximation to the binned Hawkes process in terms of its spectral density than the binned Poisson approximation.

The temporal aggregation problem also has a direct computational interpretation. For a continuous-time Hawkes process with a general excitation kernel, naive likelihood evaluation can require $O(N^2)$ operations, where $N$ is the number of observed events. Aggregating events into bins can substantially reduce this computational burden, and discretized inference methods have therefore been developed even in settings where the exact event times are available \cite{staerman2023fadin,siviero2024flexible}. Similar considerations become more pronounced for spatio-temporal processes, for which the number of events can be large. Spatial aggregation can consequently provide an additional reduction in data complexity while yielding observations that are naturally represented as counts over spatial regions.

Despite these advantages, extending coarse-grained modeling from temporal to spatio-temporal aggregation is nontrivial. Temporal and spatial aggregation jointly obscure the event configurations that determine excitation, and the resulting approximation must account for excitation both across and within the spatio-temporal bins. In particular, a model obtained by independently discretizing the temporal and spatial coordinates need not preserve the statistical structure of the aggregated process. A useful coarse-grained model should instead provide an effective excitation mechanism that represents the average influence of unobserved events over both dimensions while retaining the essential second-order dependence induced by self-excitation. 

In this paper, we develop a spatio-temporal coarse-grained Hawkes process as an extension of the temporal coarse-grained Hawkes process of \cite{koyama2025coarse}. The proposed model describes event counts aggregated over both temporal intervals and spatial regions and constructs an effective excitation kernel by averaging the underlying spatio-temporal excitation effects over the corresponding bins. A key feature is that the construction retains excitation between events occurring within the same temporal bin, thereby incorporating intra-bin effects that are absent from conventional discretization approaches. Because the exact locations of individual events may be unavailable, the spatial averaging can be performed with respect to an assumed distribution of event locations within each spatial region. This allows the model to exploit spatial excitation without requiring reconstruction of the latent event locations.

We establish the statistical properties of the proposed model and its relationship to the corresponding aggregated process generated by an underlying continuous-time spatio-temporal Hawkes process. In particular, we derive the stationary mean and autocovariance and characterize the orders of their approximation errors. These results provide a theoretical basis for evaluating the accuracy of the proposed coarse-graining procedure as a function of the temporal and spatial aggregation scales. We further compare the resulting approximation with the binned Poisson approximation, thereby quantifying the benefit of explicitly incorporating within-bin excitation. Numerical experiments examine the dependence of the approximation error on the temporal and spatial bin sizes. Finally, we apply the proposed framework to earthquake data and demonstrate parameter estimation from aggregated spatio-temporal event counts.

The contribution of this work is therefore twofold. First, we extend the coarse-grained Hawkes framework from temporal aggregation to the substantially more general setting of joint spatio-temporal aggregation. Second, we provide a theoretical characterization of the approximation induced by this coarse-graining and demonstrate its practical utility for inference from spatially and temporally aggregated event data. The proposed framework offers a computationally tractable alternative to latent-event inference while retaining a key feature of Hawkes dynamics: self-excitation within the aggregation scale.

The remainder of this paper is organized as follows. Section~\ref{sec.hawkes} reviews spatio-temporal Hawkes processes. Section~\ref{subsec.STCGHP} introduces the spatio-temporal coarse-grained Hawkes process and derives its statistical properties. The theoretical approximation results and numerical evaluations are presented in Sections~\ref{subsec.asymptotics} and \ref{subsec.numerical}, respectively. Section~\ref{sec.estimation} describes the parameter estimation procedure and its application to earthquake data.

\section{Spatio-temporal Hawkes processes}\label{sec.hawkes}
In this section, we briefly review spatio-temporal Hawkes processes and their statistical properties, and specify the regularity and stationarity conditions required for the analysis below.
Let $\{ (t_i,\bm{x}_i)\}_{i\in\mathbb{N}}$
be a sequence of random variables taking values in $\mathbb{R}_{+}\times A$, where $A\subset\mathbb{R}^q$ is a bounded spatial domain with finite volume $|A|<\infty$. Each pair $(t_i,\bm{x}_i)$ represents the occurrence time and location of an event. For $S\in\mathcal{B}(A)$, where $\mathcal{B}(A)$ denotes the Borel $\sigma$-algebra on $A$, define the counting process 
$N(t,S)=\sum_{i\in\mathbb{N}} \bm{1}_{t_i\leq t,\bm{x}_i\in S}$
and let
$\mathcal{H}_t=\{ (t_i,\bm{x}_i) \mid t_i<t\}$
denote the history of the process up to time $t$.

The conditional intensity function
$\lambda(t,\bm{x}|\mathcal{H}_t)$ specifies the stochastic dynamics of the process. Specifically, for any $S\in\mathcal{B}(A)$,
\begin{align}
    \mathrm{Pr}(N(t+\delta t,S)-N(t,S)=1|\mathcal{H}_t) &= \int_S\lambda(t,\bm{x}|\mathcal{H}_t)d\bm{x}\delta t+o(\delta t), \label{eq.rule1}\\
    \mathrm{Pr}(N(t+\delta t,S)-N(t,S)>1|\mathcal{H}_t)&= o(\delta t), \label{eq.rule2}
\end{align}
as $\delta t\to0$.
A spatio-temporal Hawkes process is defined by the conditional intensity
\begin{equation}
    \lambda(t,\bm{x}|\mathcal{H}_t) = \mu(\bm{x})+\sum_{t_i<t}\phi(t-t_i)f(\bm{x},\bm{x}_i),
     \label{eq.intensity}
\end{equation}
where $\mu(\bm{x})\geq0$ is the background intensity, while $\phi$ and $f$ are nonnegative temporal and spatial excitation kernels, respectively. 
The temporal kernel $\phi$ describes the dependence of excitation on the elapsed time since a past event, whereas $f(\bm{x},\bm{y})$ describes the spatial distribution of offspring events at location $\bm{x}$ generated by a parent event at location $\bm{y}$.

We assume that the background intensity $\mu(\bm{x})$ is bounded. 
We impose the normalization condition
\begin{equation}
\int_A f(\bm{x},\bm{y}) d\bm{x} = 1\quad \text{for all}\ \bm{y}\in A, \label{eq.normalization-1}
\end{equation}
Thus, for a parent event at $\bm{y}$, $f(\bm{x},\bm{y})$ represents the spatial density of its offspring. We additionally assume
\begin{equation}
 \int_A f(\bm{x},\bm{y}) d\bm{y} = 1\quad \text{for all}\ \bm{x}\in A. \label{eq.normalization-2}
\end{equation}
which is a mass-preservation condition used in the spatial approximation results below. The second condition follows from the first when $f$ is symmetric or when $f$ is shift-invariant under periodic boundary conditions.
For the spatial approximation analysis, we further assume that $f$ is $L_y$-Lipschitz continuous with respect to its second argument, namely,
\begin{equation}
 |f(\bm{x},\bm{y}_1)-f(\bm{x},\bm{y}_2)|\leq L_y\| \bm{y}_1-\bm{y}_2\|_2,
\label{eq.lipschitz-f}
\end{equation}
for all $\bm{x},\bm{y}_1,\bm{y}_2\in A$ and some constant $L_y>0$. 

For the temporal excitation kernel, we assume that $\phi(t)$ is nonnegative, vanishes for $t<0$, is bounded, and belongs to $L^1(\mathbb{R}_{\geq 0})\cap W^{1,2}(\mathbb{R}_{\geq 0})$. In particular,
\begin{equation}
\| \phi\|_{1}=\int^{\infty}_{-\infty}\phi(t)dt=\int^{\infty}_0\phi(t)dt<\infty,\quad \| \phi\|_{\infty} = \sup_{t\in\mathbb{R}_{\geq 0}} \phi(t)<\infty,   
\end{equation}
\begin{equation}
    \|\phi'\|_{L^2(\mathbb{R}_{\geq 0})}=\left( \int_{\mathbb{R}_{\geq 0}}|\phi'(t)|^2dt\right)^{1/2}<\infty.
\end{equation}

Under the normalization condition in Eq.~\eqref{eq.normalization-1}, the branching ratio is given by
\begin{equation}
 \alpha = \int^{\infty}_0 \int_A \phi(t)f(\bm{x},\bm{y})d\bm{x}dt=\int^{\infty}_0 \phi(t)dt=\| \phi\|_{1}
\label{eq.branching-ratio}
\end{equation}
It represents the expected number of direct offspring events generated by a single event. We assume
\begin{equation}
\alpha<1,
\label{eq.stationarity-condition}
\end{equation}
which is the subcriticality condition ensuring the existence of a stationary Hawkes process under the present assumptions. Throughout Sections~\ref{subsec.first-order} and~\ref{subsec.second-order}, we impose this condition.

Under stationarity, the mean intensity
$\bar{\lambda}(\bm{x})=\mathbb{E}[\lambda(t,\bm{x})]$
is given by
\begin{equation}
\  \bar{\lambda}(\bm{x}) = \sum^{\infty}_{n=0}\alpha^n(f^{(*n)}*\mu)(\bm{x}),
\label{eq.stationary-intensity}
\end{equation}
where $f^{(*n)}$ denotes the $n$-fold spatial convolution of $f$. Specifically,
\begin{equation}
f^{(*0)}(\bm{x},\bm{y})=\delta(\bm{x}-\bm{y}),
\end{equation}
and, for $n\geq1$,
\begin{equation}
 f^{(*n)}(\bm{x},\bm{y}) = \int_A f^{*(n-1)}(\bm{x},\bm{z})f(\bm{z},\bm{y})d\bm{z}.
\label{eq.spatial-convolution}
\end{equation}

Finally, the conditional intensity yields a martingale characterization of the counting process. Define
\begin{equation}
M_t(S) = N(t,S)-\int_0^t\int_S \lambda(s,\bm{x})d\bm{x}ds
\label{eq.martingale}
\end{equation}
Then, for every $S\in\mathcal{B}(A)$, $M_t(S)$ is a martingale with respect to the filtration generated by $\mathcal{H}_t$. In particular, for $s<t$,
\begin{equation}
\mathbb{E}[M_t(S)-M_s(S)|\mathcal{H}_s] = 0,
\label{eq.martingale-1st}
\end{equation}
and
\begin{equation}
\mathbb{E}[(M_t(S)-M_s(S))(M_t(S')-M_s(S'))|\mathcal{H}_s]=\mathbb{E}\left[ \left.\int_s^t\int_{S\cap S'} \lambda(u,\bm{x})d\bm{x}du\right|\mathcal{H}_s\right] \label{eq.martingale-2nd}
\end{equation}
for $S,S'\in\mathcal{B}(A)$; see, for example, \cite{daley2008introduction2}. When $S=S'$, Eq.~\eqref{eq.martingale-2nd} gives the conditional variance of the martingale increment. When $S\cap S'=\emptyset$, it shows that the martingale increments associated with $S$ and $S'$ are conditionally orthogonal.

\section{Spatio-temporal coarse-grained Hawkes processes}\label{sec.STCGHP}

\subsection{Definition}\label{subsec.STCGHP}
We consider a setting in which only event counts aggregated over temporal intervals and spatial regions are observed. Let $\bm{X}_1,...,\bm{X}_n$ denote the observed bin-count vectors, where
\[
    X_{i,k} = N(k\Delta t,S_i)-N((k-1)\Delta t,S_i),
\]
and $X_{i,k}$ is the number of events observed in the spatial region $S_i$ during the $k$-th temporal bin. Here, $\{S_i\}_{i=1}^d$ is a partition of the spatial domain $A$ satisfying $S_i\in \mathcal{B}(A),\cup_{i=1}^d S_i=A$ and $S_i\cap S_j=\emptyset$ for $i\neq j$. 
Our objective is to construct a statistical model for these aggregated counts that retains the essential statistical properties of the underlying binned spatio-temporal Hawkes process. Conventional approaches are often based on direct temporal and spatial discretization together with Poisson-distributed counts; see Appendix~\ref{subset.binnedpoisson}. Here, instead, we extend the coarse-graining procedure introduced in \cite{koyama2025coarse} to the spatio-temporal setting, thereby incorporating excitation effects occurring within the same temporal bin.

Consider the integrated conditional intensity over the $n$-th temporal bin and the spatial region $S_i$. Its conditional expectation, given the observed bin counts, can be written as
\begin{align}
\lambda_{i,n} &= \mathbb{E}\left[ \left.\int^{n\Delta t}_{(n-1)\Delta t}\int_{S_i} \left( \mu(\bm{x})+\sum_{t_i<t}\phi(t-t_i)f(\bm{x},\bm{x}_i) \right) d\bm{x}dt \;\right|\;\bm{X}_1,...,\bm{X}_n \right] \nonumber\\
&= \Delta t\int_{S_i}\mu(\bm{x})d\bm{x} \nonumber \\
&\quad + \sum^{n-1}_{k=1}\sum^d_{l=1} \mathbb{E}\left[ \left.\int^{n\Delta t}_{(n-1)\Delta t}\int_{S_i} \sum_{\substack{t_i\in[(k-1)\Delta t,k\Delta t)\\\bm{x}_i\in S_l}}\phi(t-t_i)f(\bm{x},\bm{x}_i) d\bm{x}dt \;\right|\;X_{l,k} \right] \nonumber \\
&\quad +\sum^d_{l=1} \mathbb{E}\left[ \left.\int^{n\Delta t}_{(n-1)\Delta t}\int_{S_i} \sum_{\substack{(n-1)\Delta t\leq t_i<t\\\bm{x}_i\in S_l}}\phi(t-t_i)f(\bm{x},\bm{x}_i) d\bm{x}dt \;\right|\;X_{l,n} \right]. \label{eq.lambdain}
\end{align}
The conditional expectations in Eq.~\eqref{eq.lambdain} generally do not admit closed-form expressions. We therefore approximate them by assuming that, conditional on the number of events in each bin, event times are independently and uniformly distributed within the corresponding temporal bin, while event locations are independently distributed according to a prescribed density $p_l(\bm{x})$ supported on $S_l$. Specifically,
\begin{equation}
p_l(\bm{x})\geq0,
\qquad
\int_{S_l}p_l(\bm{x})d\bm{x}=1.
\label{eq.normalization-spatial-density}
\end{equation}
Under this assumption, the contribution from events occurring in the $k$-th temporal bin, $k<n$, is approximated by
\begin{align*}
    &\mathbb{E}\left[ \left.\int^{n\Delta t}_{(n-1)\Delta t}\int_{S_i} \sum_{\substack{t_i\in[(k-1)\Delta t,k\Delta t)\\\bm{x}_i\in S_l}}\phi(t-t_i)f(\bm{x},\bm{x}_i) d\bm{x}dt \;\right|\;X_{l,k} \right]\\
    &\quad \approx \dfrac{X_{l,k}}{\Delta t} \int^{n\Delta t}_{(n-1)\Delta t} \int_{S_i}\int^{k\Delta t}_{(k-1)\Delta t} \int_{S_l} \phi(t-s)f(\bm{x},\bm{y})p_l(\bm{y})d\bm{y}dsd\bm{x}dt\\
    &\quad =\phi_{n-k}f_{il} X_{l,k},
\end{align*}
where $\phi_{n-k}$ is the temporal coarse-grained excitation kernel defined by
\begin{equation}
    \phi_{n-k} = \dfrac{1}{\Delta t}\int^{n\Delta t}_{(n-1)\Delta t}\int^{k\Delta t}_{(k-1)\Delta t}\phi(t-s)dsdt,\quad k=1,2,...,n-1,
    \label{eq.temporal-cgkernel-1}
\end{equation}
and $f_{il}$ is the spatial coarse-grained excitation kernel defined by
\begin{equation}
    f_{il} = \int_{S_i}\int_{S_l} f(\bm{x},\bm{y})p_l(\bm{y})d\bm{y}d\bm{x}.
    \label{eq.spatial-cgkernel}
\end{equation}
The contribution from events occurring within the current temporal bin is treated similarly. We obtain
\begin{align*}
    &\mathbb{E}\left[ \left.\int^{n\Delta t}_{(n-1)\Delta t}\int_{S_i} \sum_{\substack{(n-1)\Delta t\leq t_i<t\\\bm{x}_i\in S_l}}\phi(t-t_i)f(\bm{x},\bm{x}_i) d\bm{x}dt \;\right|\;X_{l,n} \right]\\
    &\quad \approx \dfrac{X_{l,n}}{\Delta t} \int^{n\Delta t}_{(n-1)\Delta t} \int_{S_i}\int^{t}_{(n-1)\Delta t} \int_{S_l} \phi(t-s)f(\bm{x},\bm{y})p_l(\bm{y})d\bm{y}dsd\bm{x}dt\\
    &\quad =\phi_{0}f_{il} X_{l,n},
\end{align*}
where $\phi_0$ represents the temporal intra-bin excitation effect and is defined by
\begin{equation}
    \phi_0 = \dfrac{1}{\Delta t}\int^{n\Delta t}_{(n-1)\Delta t}\int^t_{(n-1)\Delta t} \phi(t-s)dsdt.
    \label{eq.temporal-cgkernel-0}
\end{equation}
Combining these expressions gives the following approximation to the integrated conditional intensity:
\begin{equation}
    \bm{\lambda}_n \approx \bm{\mu}\Delta t+\sum_{k\leq n}\phi_{n-k}F\bm{X}_k,
    \label{eq.statiotemporal-lambda-0}
\end{equation}
where $\bm{\mu}=(\mu_1,\ldots,\mu_d)^T$ is the vector of background intensities integrated over the spatial regions,
\[
\mu_i = \int_{S_i}\mu(\bm{x})d\bm{x}, 
\]
and $F$ is the matrix of spatial excitation effects whose $(i,j)$-th element is $f_{ij}$.

A key feature of Eq.~\eqref{eq.statiotemporal-lambda-0} is the presence of the term
$\phi_0F\bm{X}_n$, which represents excitation between events occurring within the same temporal bin. This intra-bin excitation is absent from conventional discretization approaches such as the binned Poisson approximation; see Appendix~\ref{subset.binnedpoisson}.

Because $\bm{\lambda}_n$ contains the current count vector $\bm{X}_n$ through the intra-bin excitation term, it is not measurable with respect to the history up to the $(n-1)$-th bin. We therefore characterize the statistical properties of the bin-count process through the conditional moment properties of its residual process, in analogy with the martingale characterization of the continuous-time spatio-temporal Hawkes process in Eqs.~\eqref{eq.martingale-1st}--\eqref{eq.martingale-2nd}.

\begin{dfn} \label{dfn.STCGHP}
A \emph{spatio-temporal coarse-grained Hawkes process} is a non-negative integer-valued stochastic process $\{\bm{X}_n\}_n$ characterized by the residual process
\begin{equation}
    \Delta \bm{M}_n=\bm{X}_n-\bm{\lambda}_n,
 \end{equation}
and the following conditional moment properties:
\begin{align}
        \mathbb{E}[\Delta \bm{M}_n|H_{n-1}] &= \bm{0} \label{eq.spatiotemporal-first},\\
        \mathbb{E}[\Delta \bm{M}_n\Delta \bm{M}_n^T|H_{n-1}] &= \mathrm{diag}(\mathbb{E}[\bm{\lambda}_n|H_{n-1}]), \label{eq.spatiotemporal-second}
\end{align}
where $H_{n-1}=\sigma( \bm{X}_k:k\leq n-1)$ denotes the history up to the $(n-1)$-th bin, and
\begin{equation}
        \bm{\lambda}_n=\bm{\mu}\Delta t+\sum_{k\leq n}\phi_{n-k}F\bm{X}_k. \label{eq.spatiotemporal-lambda}
 \end{equation}
\end{dfn}

Definition~\ref{dfn.STCGHP} specifies only the first- and second-order conditional moment properties of the process and does not determine the complete probability distribution of ${\bm{X}_n}$. Thus, it should be interpreted as a second-order specification of a class of stochastic processes. Parameter estimation can be performed using the moment-based method constructed from this definition, as described in Section~\ref{subsec.lossfunc}, without specifying a complete conditional distribution. In contrast, likelihood-based inference and direct simulation of the process require an additional specification of its conditional distribution.

As shown in \cite{koyama2025coarse}, the sum of the temporal coarse-grained kernel defined by Eqs.~\eqref{eq.temporal-cgkernel-1} and \eqref{eq.temporal-cgkernel-0} coincides with the integral of the original temporal excitation kernel:
\begin{equation}
\sum_{k=0}^{\infty}\phi_k = \int_0^{\infty}\phi(t)dt = \alpha,
\label{eq.cg-branching-ratio}
\end{equation}
where $\alpha$ is the branching ratio of the underlying spatio-temporal Hawkes process. 
Similarly, it follows directly from the normalization condition in Eq.~\eqref{eq.normalization-1} and the definition in Eq.~\eqref{eq.spatial-cgkernel} that the spatial coarse-grained kernel satisfies
\begin{equation}
\sum_{i=1}^{d}f_{il} = 1,
\label{eq.spatial-cg-normalization}
\end{equation}
for each $l=1,\ldots,d$.

The spatio-temporal coarse-grained Hawkes process can be viewed as a natural extension of the temporal coarse-grained Hawkes process introduced in \cite{koyama2025coarse}. In particular, many of its statistical properties follow by replacing the temporal excitation coefficient $\phi_{n-k}$ in the temporal model with the matrix-valued excitation kernel $\phi_{n-k}F$.
For example, the conditional expectation and covariance matrix of $\bm{X}_n$, given the history $H_{n-1}$, are respectively
\begin{equation}
\mathbb{E}[\bm{X}_n | H_{n-1}]  
= (I-\phi_0F)^{-1} \left( \bm{\mu}\Delta t + \sum_{k \le n-1}\phi_{n-k}F\bm{X}_k \right)
=: \bm{\lambda}^\ast_n,
\label{eq.conditional-mean}
\end{equation}
and
\begin{equation}
\operatorname{Var}(\bm{X}_n | H_{n-1})
= (I-\phi_0F)^{-1} \mathrm{diag}(\bm{\lambda}^\ast_n ){(I-\phi_0F^T)^{-1}}.
\label{eq.conditional-covariance}
\end{equation}
Here, $I$ denotes the $d\times d$ identity matrix.
Under the stationarity condition $\alpha<1$, the matrix $(I-\phi_0F)^{-1}$ exists. The conditional covariance matrix is generally non-diagonal, reflecting dependence among the spatially aggregated counts induced by intra-bin excitation. Moreover, its diagonal elements are generally larger than the corresponding conditional means, indicating overdispersion relative to conditionally Poisson counts. When the intra-bin excitation vanishes, $\phi_0=0$, Eq.~\eqref{eq.conditional-covariance} reduces to
$\operatorname{Var}(\bm{X}_n | H_{n-1}) = \operatorname{diag}(\bm{\lambda}^\ast_n )$, 
which corresponds to conditionally independent Poisson counts, as in the binned Poisson approximation described in Appendix~\ref{subset.binnedpoisson}.

We next derive the stationary mean and autocovariance of the spatio-temporal coarse-grained Hawkes process.
\begin{cor}
\label{cor:spatiotemporal-expectation}
Assuming stationarity and $\alpha<1$, the stationary expectation of $\bm{X}_n$ is
    \begin{equation}
        \mathbb{E}[\bm{X}_n] = (I-\alpha F)^{-1}\bm{\mu}\Delta t  .\label{eq.spatiotemporal-expectation}
    \end{equation}
\end{cor}
\begin{proof}
Taking expectations of Eq.~\eqref{eq.spatiotemporal-lambda} under stationarity gives
\[
        \mathbb{E}[\bm{\lambda}_n] = \bm{\mu}\Delta t+\sum_{k=0}^{\infty}\phi_k F \mathbb{E}[\bm{X}_n].
 \]
By Eq.~\eqref{eq.spatiotemporal-first},
$\mathbb{E}[\bm{X}_n]=\mathbb{E}[\bm{\lambda}_n]$, 
and hence
\[
        \left( I-\sum^{\infty}_{k=0}\phi_k F\right)\mathbb{E}[\bm{X}_n]=\bm{\mu}\Delta t.
\]
Using Eq.~\eqref{eq.cg-branching-ratio},
$\sum_{k=0}^{\infty}\phi_k=\alpha$,
we obtain
\[
(I-\alpha F)\mathbb{E}[\bm{X}_n] = \bm{\mu}\Delta t.
\]
Since $F$ is a non-negative matrix whose column sums are equal to one by Eq.~\eqref{eq.spatial-cg-normalization}, its spectral radius satisfies $\rho(F)=1$.
Consequently, $\rho(\alpha F)=\alpha<1$, and therefore $I-\alpha F$ is invertible. This yields Eq.~\eqref{eq.spatiotemporal-expectation}.
\end{proof}

When a stationary process satisfying Definition~\ref{dfn.STCGHP} exists, the condition $\alpha<1$ therefore ensures uniqueness of its non-negative stationary mean for every non-negative background-intensity vector $\bm{\mu}$.

\begin{cor}\label{cor:spatiotemporal-autocovariance}
Assuming stationarity and $\alpha<1$, the autocovariance matrix at lag $s\geq0$ is
\begin{equation}
        \mathrm{Cov}(\bm{X}_n,\bm{X}_{n+s}) = \Delta t\sum_{l=0}^{\infty}\Psi_{l}\mathrm{diag}((I-\alpha F)^{-1}\bm{\mu})\Psi_{l+s}^T,
\end{equation}
where
\[
        \Psi_k = \sum^{\infty}_{j=0}(\phi^{(*j)})_kF^j,
\]
and $(\phi^{(*j)})_k$ denotes the $k$-th element of the $j$-fold discrete convolution of $\{\phi_i\}_i$, with $(\phi^{(*0)})_k=\delta_{k,0}$.
\end{cor}
\begin{proof}
See \cite[Theorem~1]{koyama2025coarse}, with the temporal excitation kernel matrix $\Phi_{n-k}$ replaced by the matrix-valued kernel $\phi_{n-k}F$.
\end{proof}

\subsection{Asymptotic analysis of the statistical moments}
\label{subsec.asymptotics}
In this subsection, we evaluate the approximation errors of the first- and second-order statistical properties of the stationary spatio-temporal coarse-grained Hawkes process relative to those of the corresponding binned Hawkes process. We consider the joint asymptotic regime $\Delta t\to 0$ and $|S_i|\to 0$, where $\Delta t$ is the temporal bin width and $|S_i|$ is the volume of the $i$-th spatial bin.

To control the spatial approximation error, we impose the following additional regularity condition on the spatial partition. There exists a constant $C_r>0$ such that
\begin{equation}
\sup_{\bm{x},\bm{y}\in S_i}
|\bm{x}-\bm{y}|_2
\leq
C_r |S_i|^{1/q},
\label{eq.regularity}
\end{equation}
for every $i=1,\ldots,d$. This condition relates the diameter of each spatial bin to its volume and, in particular, excludes spatial bins that become arbitrarily thin in some directions. 
For example, when the spatial bins are cubic, Eq.~\eqref{eq.regularity} is satisfied if their side lengths converge to zero at the same rate.

Let $\bm{X}^{(\mathrm{cg})}_n$ denote the vector of stationary spatio-temporal coarse-grained Hawkes counts and let $\bm{X}^{(\mathrm{hw})}_n$ denote the vector of stationary binned Hawkes counts aggregated over the same spatio-temporal bins.

\subsubsection{First-order properties}\label{subsec.first-order}
We first consider the stationary mean.
\begin{thm}\label{thm.expectation}
   Assuming stationarity, for each $i=1,\ldots,d$,
    \begin{align}
        \mathbb{E}[X_{i,n}^{(\mathrm{cg})}] = \mathbb{E}[X_{i,n}^{(\mathrm{hw})}] + O(\Delta t |S|^{1/q}|S_i|),
        \label{eq.cghawkes-expectation-asymptotic}
    \end{align}
    where $|S|=\max_{1\leq i\leq d}|S_i|$.
\end{thm}
\begin{proof}
    Under stationarity, Corollary~\ref{cor:spatiotemporal-expectation} gives
    \[
        \mathbb{E}[X_{i,n}^{(\mathrm{cg})}]=\Delta t\sum_{n=0}^{\infty}\alpha^n [F^n\bm{\mu}]_i = \Delta t\sum_{n=0}^{\infty}\alpha^n\left( \sum_{j=1}^d [F^n]_{ij}\mu_j\right).
    \]
    By Lemma~\ref{lem.expectation-convolution}, 
    \[
        \mathbb{E}[X_{i,n}^{(\mathrm{cg})}]= \Delta t\sum_{n=0}^{\infty}\alpha^n\int_{S_i}(f^{(*n)}*\mu)(\bm{x})d\bm{x}+\Delta t \sum_{n=0}^{\infty}\alpha^n O(n|S|^{1/q}|S_i|).
     \]
     The first term on the right-hand side is the stationary mean of the corresponding binned Hawkes process. Since $\alpha<1$,
     \[
     \sum_{n=0}^{\infty}n\alpha^n = \frac{\alpha}{(1-\alpha)^2},
     \]
    and therefore
     \[
        \mathbb{E}[X_{i,n}^{(\mathrm{cg})}]
           = \mathbb{E}[X_{i,n}^{(\mathrm{hw})}] + O(\Delta t |S|^{1/q}|S_i|).
     \]
\end{proof}

Theorem~\ref{thm.expectation} holds for any prescribed spatial densities $p_j$ satisfying the conditions in Eq.~\eqref{eq.normalization-spatial-density}. However, the approximation error in the stationary mean can be eliminated entirely by choosing the spatial densities proportional to the stationary intensity. This observation is summarized in the following theorem.
\begin{thm}\label{thm.exp-stationary}
Assuming stationarity and choosing
\begin{equation}
p_j(\bm{x})=\frac{\bar{\lambda}(\bm{x})}{\int_{S_j}\bar{\lambda}(\bm{y})d\bm{y}},
\qquad
\bm{x}\in S_j,
\label{eq.prescribed-spatial-density}
\end{equation}
the stationary means of the spatio-temporal coarse-grained Hawkes process and the binned Hawkes process coincide:
\begin{equation}
    \mathbb{E}[X_{i,n}^{(\mathrm{cg})}] = \mathbb{E}[X_{i,n}^{(\mathrm{hw})}].
    \label{eq.exact-mean}
\end{equation}
\end{thm}
\begin{proof}
    The stationary intensity $\bar{\lambda}(\bm{x})$ satisfies
    \[
\bar{\lambda}(\bm{x})
= \mu(\bm{x}) +
\alpha (f*\bar{\lambda})(\bm{x}).
    \]
Integrating both sides over $S_i$ gives
    \[
    \int_{S_i}\bar{\lambda}(\bm{x})d\bm{x}=\int_{S_i}\mu(\bm{x})d\bm{x}+\alpha\int_{S_i} (f*\bar{\lambda})(\bm{x})d\bm{x}.
    \]
    Using the partition $A=\bigcup_{j=1}^d S_j$ and the choice in Eq.~\eqref{eq.prescribed-spatial-density},
    \begin{align*}
    \mathbb{E}[X_{i,n}^{(\mathrm{hw})}] &= \Delta t\int_{S_i} \bar{\lambda}(\bm{x})d\bm{x}\\
    &=\Delta t\int_{S_i}\mu(\bm{x})d\bm{x}+\alpha\Delta t\sum^d_{j=1}\int_{S_i}\int_{S_j} f(\bm{x},\bm{y})\bar{\lambda}(\bm{y})d\bm{x}d\bm{y}\\
    &=\mu_i\Delta t + \alpha \sum^{d}_{j=1}F_{ij}\mathbb{E}[X_{j,n}^{(\mathrm{hw})}],
    \end{align*}
    Thus,
    \[
\mathbb{E}[\bm{X}_{n}^{(\mathrm{hw})}] = \bm{\mu}\Delta t+\alpha F\mathbb{E}[\bm{X}_{n}^{(\mathrm{hw})}],
    \]
    and hence
    \begin{align*}
        \mathbb{E}[\bm{X}_{n}^{(\mathrm{hw})}] = (1-\alpha F)^{-1}\bm{\mu}\Delta t.
    \end{align*}
    By Corollary~\ref{cor:spatiotemporal-expectation}, this is also the stationary mean of the spatio-temporal coarse-grained Hawkes process. Therefore Eq.~\eqref{eq.exact-mean} follows.
\end{proof}
The preceding result shows that the spatial approximation error in the stationary mean is not intrinsic to the coarse-graining procedure. Rather, it arises from the choice of the within-bin spatial distributions $p_j$. Choosing these distributions according to the stationary intensity provides an exact match of the stationary first-order property.

We next compare the approximation with the binned Poisson model described in Appendix~\ref{subset.binnedpoisson}. Let $\bm{X}^{(\mathrm{po})}_n$ denote the vector of stationary counts generated by the binned Poisson approximation.
For the asymptotic analysis of its stationary mean, we impose additional regularity assumptions. Specifically, we assume that $f$ is $L_x$-Lipschitz continuous with respect to its first argument and that the background intensity $\mu$ is Lipschitz continuous. The former condition is automatically satisfied when the spatial excitation kernel is shift-invariant and is Lipschitz continuous with respect to its relative spatial coordinate. For the temporal kernel, we additionally assume that
\[
    \Delta t\sum_{k=1}^{\infty}\phi(k\Delta t)
    =\alpha+O(\Delta t)
\]
as $\Delta t\to0$.
Under these assumptions and stationarity,
\begin{align}
        \mathbb{E}[X_{i,n}^{(\mathrm{po})}] = \mathbb{E}[X_{i,n}^{(\mathrm{hw})}]+O(\Delta t|S|^{1/q}|S_i|)+O((\Delta t)^2|S_i|), 
        \label{eq.binnedpoisson-expectation-asymptotic}
\end{align}
as $\Delta t\to0$ and $|S_i|\to0$. The proof is given in Appendix~\ref{appendix.binnedpoisson}.

Comparing Eqs.~\eqref{eq.cghawkes-expectation-asymptotic} and
\eqref{eq.binnedpoisson-expectation-asymptotic}, we see that the stationary mean of the spatio-temporal coarse-grained Hawkes process does not contain the temporal discretization error 
$O((\Delta t)^2|S_i|)$
that appears in the binned Poisson approximation.
This improvement results from explicitly incorporating excitation effects occurring within the same temporal bin through the intra-bin excitation coefficient $\phi_0$.
Moreover, when the spatial densities are chosen according to Eq.~\eqref{eq.prescribed-spatial-density}, the spatial approximation error
$O(\Delta t|S|^{1/q}|S_i|)$
is also eliminated for the stationary mean. Thus, under this choice of spatial density, the stationary mean of the proposed model agrees exactly with that of the binned Hawkes process, whereas the binned Poisson approximation retains both spatial and temporal discretization errors.

\subsubsection{Second-order properties}\label{subsec.second-order}

We next investigate the second-order properties of the stationary spatio-temporal coarse-grained Hawkes process. In particular, we evaluate the approximation error of its autocovariance matrix relative to that of the corresponding binned Hawkes process in the joint limit $\Delta t\to 0$ and $|S_i|\to 0$. 

\begin{thm} \label{thm.cov}
   Assuming stationarity, for each $i,j=1,\ldots,d$ and temporal lag $p\geq 0$,
    \begin{align}
        &\mathrm{Cov}(X_{i,n}^{(\mathrm{cg})},X_{j,n+p}^{(\mathrm{cg})}) = \mathrm{Cov}(X_{i,n}^{(\mathrm{hw})},X_{j,n+p}^{(\mathrm{hw})}) \nonumber\\
        &+\begin{cases}
            O((\Delta t)^3|S_i|^2)+O(\Delta t|S_i||S|^{1/q}), & p=0, \ i= j, \\
            O((\Delta t)^3|S_i||S_j|) + O((\Delta t)^2|S_i||S_j||S|^{1/q}), & \text{otherwise} \\
        \end{cases}
        \label{eq.cghawkes-autocovariance-asymptotic}
    \end{align}    
\end{thm}
The proof of Theorem~\ref{thm.cov} is given in Appendix~\ref{appendix.2nd-order}.

The diagonal zero-lag case, $p=0$ and $i=j$, has a different error order from the other covariance terms. This difference arises because the covariance in this case is the variance of a single bin count. In particular, the same event contributes to both factors in the product defining the second moment. Consequently, the leading diagonal term is proportional to the expected number of events in the bin, and the first-order spatial approximation error in the mean directly contributes to the variance error. This produces the term
$O(\Delta t |S_i||S|^{1/q})$,
which has a lower power of $\Delta t$ than the corresponding error terms for the off-diagonal and nonzero-lag covariances.

The spatial approximation error can again be reduced by choosing the within-bin spatial densities according to Eq.~\eqref{eq.prescribed-spatial-density}.
\begin{cor}
\label{cor.cov-stationary-density}
Assuming stationarity and choosing $p_j$ according to Eq.~\eqref{eq.prescribed-spatial-density},
\begin{align}
&
\operatorname{Cov}
(
X_{i,n}^{(\mathrm{cg})},
X_{j,n+p}^{(\mathrm{cg})}
)
=
\operatorname{Cov}
(
X_{i,n}^{(\mathrm{hw})},
X_{j,n+p}^{(\mathrm{hw})}
)
\nonumber\\
&\qquad+
O((\Delta t)^3|S_i||S_j|)
+
O((\Delta t)^2|S_i||S_j||S|^{1/q})
\label{eq.cghawkes-autocovariance-stationary-density}
\end{align}
for the covariance terms covered by Theorem~\ref{thm.cov}.
\end{cor}

The corresponding approximation errors for the binned Poisson approximation are, in the joint limit $\Delta t\to0$ and $|S_i|\to0$,
\begin{align}
 &\mathrm{Cov}(X_{i,n}^{(\mathrm{po})},X_{j,n+p}^{(\mathrm{po})}) = \mathrm{Cov}(X_{i,n}^{(\mathrm{hw})},X_{j,n+p}^{(\mathrm{hw})}) \nonumber\\
 &+\begin{cases}
 	O((\Delta t)^2|S_i|^2)+O(\Delta t|S_i||S|^{1/q})+O((\Delta t)^2|S_i|), & p=0,\ i= j, \\
	O((\Delta t)^2|S_i||S_j|), & p=0,\ i\neq j, \\
 	O((\Delta t)^{5/2}|S_i||S_j|) + O((\Delta t)^2|S_i||S_j||S|^{1/q}), & p>0. \\
 \end{cases}
 \label{eq.binnedpoisson-autocovariance-asymptotic}
\end{align}
The proof is given in Appendix~\ref{appendix.binnedpoisson-autocovariance}.

Comparing Eqs.~\eqref{eq.cghawkes-autocovariance-asymptotic} and
\eqref{eq.binnedpoisson-autocovariance-asymptotic} reveals a systematic improvement in the temporal order of approximation obtained by the spatio-temporal coarse-grained Hawkes process. In particular, the proposed model incorporates the excitation generated by events occurring within the same temporal bin through the intra-bin excitation term. As a result, the leading temporal discretization errors associated with the binned Poisson approximation are absent or shifted to higher order.

The remaining approximation error of the proposed model is therefore primarily associated with the spatial coarse-graining and with the approximation of the within-bin event configuration. Furthermore, when the within-bin spatial distributions are chosen proportional to the stationary intensity, the leading spatial error terms can be removed, as demonstrated for the first-order property in Theorem~\ref{thm.exp-stationary} and for the second-order property in Corollary~\ref{cor.cov-stationary-density}. This provides a theoretical explanation for the improved approximation accuracy of the proposed spatio-temporal coarse-grained Hawkes process relative to the binned Poisson approximation.

\subsection{Numerical evaluation for varying spatial and temporal bin sizes}
\label{subsec.numerical}
As established in Section~\ref{subsec.asymptotics}, the spatio-temporal coarse-grained Hawkes process provides a higher-order approximation to the first- and second-order moment properties of the binned Hawkes process than the binned Poisson approximation in the asymptotic regime $\Delta t\to0$ and $|S_i|\to0$. In this section, we numerically investigate the validity and robustness of this approximation over a range of temporal and spatial bin sizes.

We consider a spatio-temporal Hawkes process on the two-dimensional spatial domain, 
$A=[0,1]\times[0,1]\subset\mathbb{R}^2$, 
with background intensity
\[
\mu(\bm{x})
=0.5+2.0\left|\cos\left(2\pi\times0.3\times(x_1+x_2)\right)\right|.
\]
The temporal and spatial excitation kernels are, respectively,
\[
\phi(t)=0.9e^{-t}\bm{1}_{\{t\geq0\}},
\]
which has branching ratio $\alpha=0.9$, and
\[
f(\bm{x},\bm{y})
=\sum_{\bm{v}\in\mathbb{Z}^2}
\frac{1}{(2\pi)\times0.05}
\exp\left(
-\frac{\|\bm{x}-\bm{y}+\bm{v}\|_2^2}{2\times0.05}
\right).
\]
The spatial kernel satisfies the normalization conditions in Eqs.~\eqref{eq.normalization-1}--\eqref{eq.normalization-2}. The periodic summation in the definition of $f$ corresponds to periodic boundary conditions on $A$.

We aggregate the spatio-temporal Hawkes process over spatial square bins
\[
S_{k,l}
=[k\Delta x,(k+1)\Delta x]\times[l\Delta x,(l+1)\Delta x],
\]
where $k,l=0,\ldots,\Delta x^{-1}-1$, and use temporal bins of width $\Delta t$. The number of spatial bins is therefore $d=\Delta x^{-2}$ for the values of $\Delta x$ considered below.

\subsubsection{Stationary mean}
We first evaluate the stationary mean of the aggregated counts. We measure the approximation error by the $L^1$ distance between the stationary mean vectors, normalized by the temporal bin width,
\begin{equation}
\mathcal{E}_1^{(\blacklozenge)}(\Delta t,\Delta x)
=
\frac{1}{\Delta t}
\left\|
\mathbb{E}[\bm{X}_n^{(\mathrm{hw})}]
-
\mathbb{E}[\bm{X}_n^{(\blacklozenge)}]
\right\|_1
=
\frac{1}{\Delta t}
\sum_{i=1}^d
\left|
\mathbb{E}[X_{i,n}^{(\mathrm{hw})}]
-
\mathbb{E}[X_{i,n}^{(\blacklozenge)}]
\right|,
\label{eq:numerical-first-error}
\end{equation}
where $\blacklozenge$ denotes either the coarse-grained (cg) or the binned Poisson (po) approximation. The normalization by $\Delta t$ converts the error in bin counts into an error in the corresponding mean event rates. 
For sufficiently small $\Delta t$, Eq.~\eqref{eq:numerical-first-error} approximates the $L^1$ difference between the corresponding stationary intensity fields integrated over the spatial domain. In particular, because the stationary mean of the spatio-temporal coarse-grained Hawkes process is proportional to $\Delta t$, its normalized stationary-mean error is independent of $\Delta t$ for a fixed spatial partition; see Corollary~\ref{cor:spatiotemporal-expectation}. When the spatial densities $p_j$ are proportional to the stationary means, this error is zero according to Theorem~\ref{thm.exp-stationary}.

Figure~\ref{fig:1st-dt} shows the normalized $L^1$ errors as functions of $\Delta t$, with $\Delta x=0.1$. For the spatio-temporal coarse-grained Hawkes process, we use the uniform distribution within each spatial bin for $p_j$. The error remains approximately constant as $\Delta t$ varies. In contrast, the error of the binned Poisson approximation increases with $\Delta t$. This behavior is primarily attributable to the discretization of the temporal excitation kernel: the branching ratio $\alpha=\int_0^\infty\phi(t)\,dt$ is approximated in the binned Poisson model by the Riemann sum, 
$\sum_{k=1}^{\infty}\phi(k\Delta t)\Delta t$.
The proposed coarse-grained construction instead preserves the integrated excitation exactly through Eq.~\eqref{eq.cg-branching-ratio}, including the contribution from intra-bin excitation.

Figure~\ref{fig:1st-dx} shows the corresponding errors as functions of the spatial bin width $\Delta x$, with $\Delta t=0.1$. The error of the spatio-temporal coarse-grained Hawkes process decreases toward zero as $\Delta x$ decreases, consistent with the spatial approximation results in Section~\ref{subsec.asymptotics}. In contrast, the error of the binned Poisson approximation does not vanish as $\Delta x$ decreases because the temporal discretization error in its approximation of the branching ratio remains nonzero for fixed $\Delta t$.

\begin{figure}
    \centering
    \includegraphics[width=0.55\linewidth]{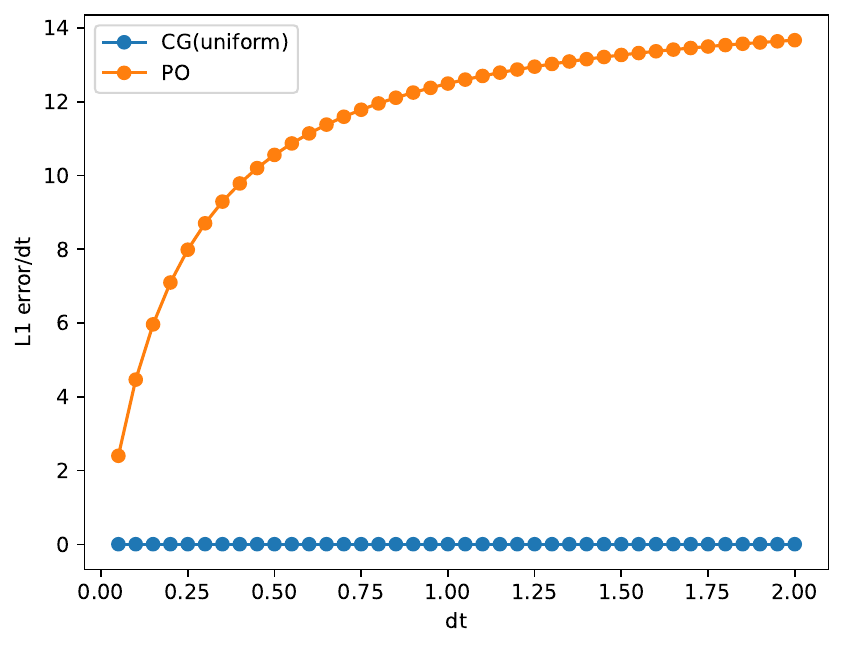}
    \caption{Error in approximating the stationary mean of the binned Hawkes process as a function of the temporal bin width $\Delta t$. The spatial bin width is fixed at $\Delta x=0.1$. The error of the binned Poisson approximation increases with $\Delta t$, whereas the error of the spatio-temporal coarse-grained Hawkes process remains approximately constant at $\approx5.3\times10^{-3}$.}
    \label{fig:1st-dt}
\end{figure}

\begin{figure}
    \centering
    \includegraphics[width=0.55\linewidth]{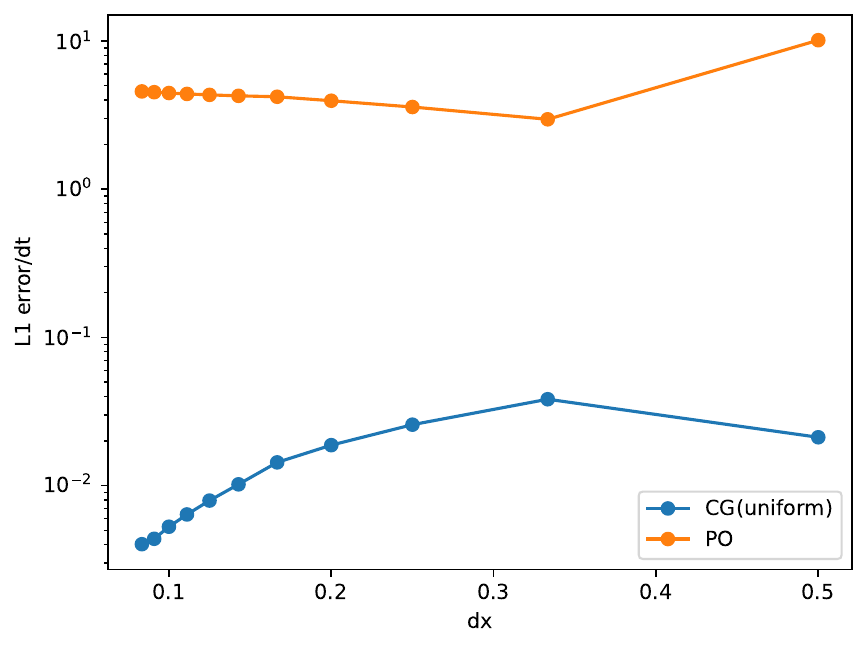}
    \caption{Error in approximating the stationary mean of the binned Hawkes process as a function of the spatial bin width $\Delta x$. The temporal bin width is fixed at $\Delta t=0.1$. The vertical axis is logarithmic. The error of the spatio-temporal coarse-grained Hawkes process decreases toward zero as $\Delta x$ decreases, whereas the error of the binned Poisson approximation remains nonzero.}
    \label{fig:1st-dx}
\end{figure}

\subsubsection{Stationary autocovariance}
We next evaluate the stationary autocovariance. To exclude the zero-lag diagonal terms, whose asymptotic order differs from that of the other covariance terms, we define
\[
\mathcal{A}
=\left\{
(p,i,j):p\in\mathbb{Z},\ i,j=1,\ldots,d,\ (p,i,j)\neq(0,i,i)
\right\}.
\]
We measure the approximation error by the normalized $L^2$ distance
\begin{equation}
\mathcal{E}_2^{(\blacklozenge)}(\Delta t,\Delta x)
=
\left[
\sum_{(p,i,j)\in\mathcal{A}}
\frac{
\left|
\operatorname{Cov}(X_{i,n}^{(\mathrm{hw})},X_{j,n+p}^{(\mathrm{hw})})
-
\operatorname{Cov}(X_{i,n}^{(\blacklozenge)},X_{j,n+p}^{(\blacklozenge)})
\right|^2
}{
(\Delta t)^3|S_i||S_j|
}
\right]^{1/2}.
\label{eq:numerical-second-error}
\end{equation}
The normalization in Eq.~\eqref{eq:numerical-second-error} is chosen according to the asymptotic order of the covariance errors established in Section~\ref{subsec.asymptotics}. In the continuum limit, this discrete measure can be interpreted as a discretized version of the Hilbert--Schmidt norm of the difference between the corresponding covariance kernels \cite{pigoli2014distances},
\begin{equation}
\left\|
 k^{(\mathrm{hw})}-k^{(\blacklozenge)}
\right\|_{\mathrm{HS}}
=
\left(
\int_{\mathbb{R}}
\int_A\int_A
\left[
 k^{(\mathrm{hw})}(t;\bm{x},\bm{y})
-
 k^{(\blacklozenge)}(t;\bm{x},\bm{y})
\right]^2
\,d\bm{x}\,d\bm{y}\,dt
\right)^{1/2}.
\label{eq:hs-covariance-error}
\end{equation}

Figure~\ref{fig:2nd-dt} shows the covariance errors as functions of $\Delta t$, with $\Delta x=0.2$. The errors decrease as the temporal bin width decreases, in agreement with the asymptotic results. Across the range considered, the errors of the spatio-temporal coarse-grained Hawkes process are approximately one to two orders of magnitude smaller than those of the binned Poisson approximation, even when $\Delta t$ is comparable to the characteristic time scale of the temporal excitation kernel. This demonstrates that explicitly incorporating intra-bin excitation can substantially improve the approximation beyond the small-bin regime.

Figure~\ref{fig:2nd-dx} shows the covariance errors as functions of $\Delta x$, with $\Delta t=0.1$. The errors of the spatio-temporal coarse-grained Hawkes process are again approximately one to two orders of magnitude smaller than those of the binned Poisson approximation. Overall, the covariance error of the proposed method is relatively insensitive to the spatial bin width over the range considered. We additionally consider a variant in which the within-bin spatial densities $p_j$ are chosen proportional to the stationary intensity, as in Eq.~\eqref{eq.prescribed-spatial-density}. This choice leads to a modest reduction in the approximation error for relatively large spatial bins.

\begin{figure}
    \centering
    \includegraphics[width=0.55\linewidth]{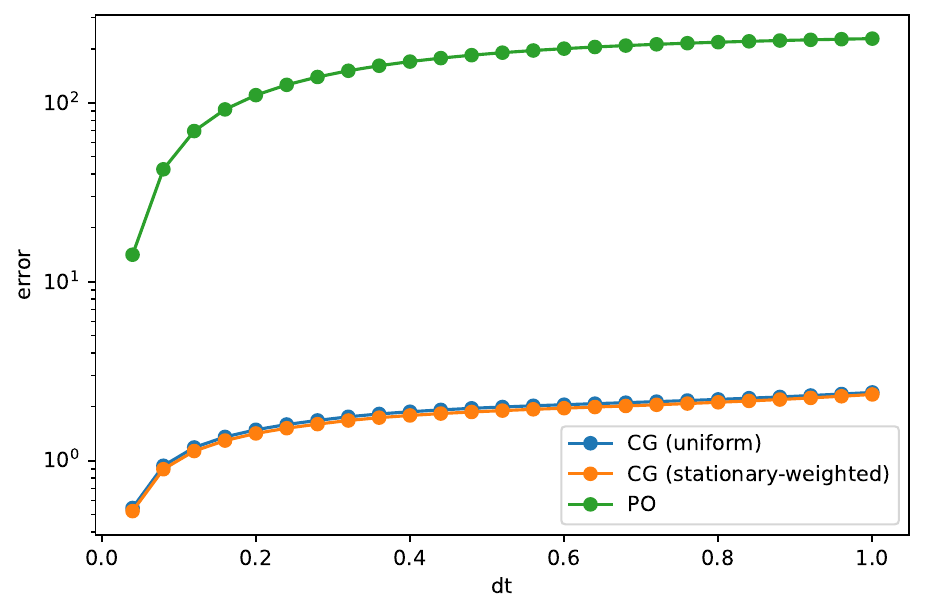}
    \caption{Error of the stationary autocovariance as a function of the temporal bin width $\Delta t$. The spatial bin width is fixed at $\Delta x=0.2$. The vertical axis is logarithmic. The error of the spatio-temporal coarse-grained Hawkes process is approximately one to two orders of magnitude smaller than that of the binned Poisson approximation over the range of $\Delta t$ considered.}
    \label{fig:2nd-dt}
\end{figure}

\begin{figure}
    \centering
    \includegraphics[width=0.55\linewidth]{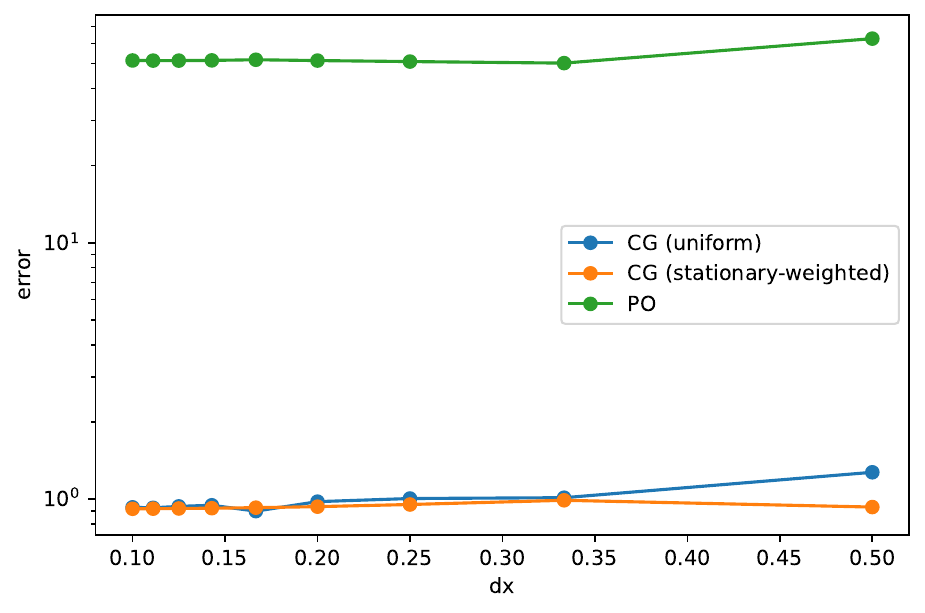}
    \caption{Error of the stationary autocovariance as a function of the spatial bin width $\Delta x$. The temporal bin width is fixed at $\Delta t=0.1$. The vertical axis is logarithmic. The error of the spatio-temporal coarse-grained Hawkes process is approximately one to two orders of magnitude smaller than that of the binned Poisson approximation over the range of $\Delta x$ considered.}
    \label{fig:2nd-dx}
\end{figure}

Overall, these numerical results support the asymptotic analysis and demonstrate that the spatio-temporal coarse-grained Hawkes process provides a substantially more accurate approximation to the first- and second-order properties of the spatially and temporally binned Hawkes process than the binned Poisson approximation over a broad range of temporal and spatial bin widths. In particular, the improvement is most pronounced for the second-order properties, for which accounting explicitly for intra-bin excitation substantially reduces the approximation error even when the temporal bin width is not small relative to the characteristic excitation time scale.

\section{Parameter estimation}\label{sec.estimation}

In this section, we develop a moment-based parameter estimation method for the spatio-temporal coarse-grained Hawkes process and evaluate its performance using earthquake data. We assume that the temporal and spatial excitation kernels, $\phi$ and $f$, are parameterized by a parameter vector $\bm{\theta}$. Accordingly, the coarse-grained temporal kernel and spatial excitation matrix are written as $\phi_k(\bm{\theta})$ and $F(\bm{\theta})$, respectively.

\subsection{Estimation procedure}\label{subsec.lossfunc}

Following the derivation in \cite{koyama2025coarse} and substituting the matrix-valued kernel $\Phi_k$ with $\phi_kF$, the centered count process $\{\bm{X}_n\}$ admits the following $\mathrm{AR}(\infty)$ representation:
\begin{equation}
    \bm{X}_n-\bm{\lambda} = \sum_{k=1}^{\infty}\phi_k(\bm{\theta})F^*(\bm{\theta})(\bm{X}_{n-k}-\bm{\lambda})+\Delta \bm{M}^*_n(\bm{\theta}),
\end{equation}
where $\bm{\lambda}$ is the stationary mean count vector,
\[
    F^*(\bm{\theta})
    =\left(I-\phi_0(\bm{\theta})F(\bm{\theta})\right)^{-1}F(\bm{\theta}),
\]
and 
\[
    \Delta \bm{M}_n^*(\bm{\theta})
    =\left(I-\phi_0(\bm{\theta})F(\bm{\theta})\right)^{-1}
    \Delta \bm{M}_n(\bm{\theta}).
\]
Because Definition~\ref{dfn.STCGHP} specifies only the first- and second-order conditional moment properties, the following objective function should not be interpreted as an exact likelihood. Instead, we construct a moment-based loss function from the $\mathrm{AR}(\infty)$ representation and the conditional moment properties in Eqs.~\eqref{eq.spatiotemporal-first} and \eqref{eq.spatiotemporal-second}:
\begin{align}
    \mathcal{L}(\bm{\theta})&=\sum_{k=1}^n \left( \bm{\Upsilon}_k-\sum_{l=1}^{k-1}\phi_{k-l}(\bm{\theta})F^*(\bm{\theta})\bm{\Upsilon}_l\right)^{T}\Lambda^*(\bm{\theta})^{-1}\left( \bm{\Upsilon}_k-\sum_{l=1}^{k-1}\phi_{k-l}(\bm{\theta})F^*(\bm{\theta})\bm{\Upsilon}_l\right) \nonumber\\
    &\quad +n\log \det\Lambda^*(\bm{\theta}),
\end{align}
where $\bm{\Upsilon}_k=\bm{X}_k-\bm{\lambda}$ and
\[
 \Lambda^*(\bm{\theta})=(I-\phi_0(\bm{\theta})F(\bm{\theta}))^{-1}\mathrm{diag}(\bm{\lambda})((I-\phi_0(\bm{\theta})F(\bm{\theta}))^{-1})^T.
\]
This objective can be interpreted as a discrete, moment-based analogue of a least-squares criterion for the conditional intensity. In particular, it is partly related asymptotically to the weighted least-squares loss for point-process intensities \cite{cheng2025deep} in the limit $\Delta t\to0$ and $|S|\to0$.

In practice, the stationary mean $\bm{\lambda}$ is unknown and is replaced by the empirical mean
\[
\hat{\bm{\lambda}}
=
\frac{1}{n}\sum_{k=1}^n\bm{X}_k.
\]
After substituting the empirical mean and omitting terms that are independent of $\bm{\theta}$, the objective function can be written as
\begin{align}
    \mathcal{L}(\bm{\theta})
    &=
    \sum_{k=1}^n
    \left(
    \bm{\Upsilon}_k
    -
    \sum_{l=1}^{k}
    \phi_{k-l}(\bm{\theta})
    F(\bm{\theta})
    \bm{\Upsilon}_l
    \right)^{T}
    \mathrm{diag}(\hat{\bm{\lambda}})^{-1}
    \left(
    \bm{\Upsilon}_k
    -
    \sum_{l=1}^{k}
    \phi_{k-l}(\bm{\theta})
    F(\bm{\theta})
    \bm{\Upsilon}_l
    \right)
    \nonumber\\
    &\quad
    -2n\log
    \det\left(
    I-\phi_0(\bm{\theta})F(\bm{\theta})
    \right),
\end{align}
where $\bm{\Upsilon}_k = \bm{X}_k-\hat{\bm{\lambda}}$ in the empirical implementation. The inclusion of the $l=k$ term accounts for excitation occurring within the current temporal bin through the intra-bin coefficient $\phi_0$. The second term acts as a regularization term that prevents the objective from being minimized by a degenerate solution 
$\phi_k(\bm{\theta})F(\bm{\theta}) = \delta_{k0}I$.

The parameter-estimation procedure is summarized in Algorithm~\ref{algorithm}.

\begin{algorithm}[H]
    \caption{Procedure for parameter estimation}
    \label{algorithm}
    \begin{algorithmic}[1]
    \STATE From the count data $\{\bm{X}_1,\ldots,\bm{X}_n\}$, compute the empirical mean
    $\hat{\bm{\lambda}}=n^{-1}\sum_{k=1}^n\bm{X}_k$ and set
    $\bm{\Upsilon}_k=\bm{X}_k-\hat{\bm{\lambda}}$.
    \STATE For a given parameter vector $\bm{\theta}$, calculate $\phi_k(\bm{\theta})$ and $F(\bm{\theta})$, and evaluate $\mathcal{L}(\bm{\theta})$.
    \STATE Minimize $\mathcal{L}(\bm{\theta})$ over the parameter space and obtain the estimate $\hat{\bm{\theta}}$.
    \end{algorithmic}
\end{algorithm}

\subsection{Application to earthquake forecasting}
\label{subsec.earthquake}

Earthquake occurrence is a natural application of spatio-temporal Hawkes processes. In this subsection, we formulate an ETAS-type model within the proposed coarse-grained framework and evaluate its forecasting performance using aggregated earthquake counts.

The Epidemic-Type Aftershock Sequence (ETAS) model \cite{ogata1988statistical,ogata1998space} is a spatio-temporal Hawkes model for earthquake occurrence in which the excitation generated by an earthquake depends on its magnitude. Its conditional intensity can be written as
\begin{equation}
    \lambda(t,\bm{x}\mid\mathcal{H}_t)
    =\mu(\bm{x})
    +\sum_{t_i<t}
    g(t-t_i,\bm{x}-\bm{x}_i,M_i)\bm{1}_{M_i\geq M_0},
\end{equation}
where $\{(t_i,\bm{x}_i,M_i)\}_i$ denotes the sequence of earthquake occurrence times, locations, and magnitudes, and $M_0$ is the magnitude threshold.

Various parametric forms have been proposed based on empirical laws in seismology. Following \cite{stockman2026earthquakenpp}, we adopt the intensity function
\begin{equation}
    \lambda(t,\bm{x}\mid\mathcal{H}_t)
    =\mu
    +\sum_{t_i<t}
    \frac{
    e^{-(t-t_i)/\tau}\,k\,e^{a(M_i-M_0)}
    }{
    (t-t_i+c)^{1+\omega}
    \left(\|\bm{x}-\bm{x}_i\|^2+d\,e^{\gamma(M_i-M_0)}\right)^{1+\rho}
    }
    \bm{1}_{M_i\geq M_0},
    \label{eq.ETASintensity}
\end{equation}
with parameters $\bm{\theta} = (\mu,k,a,c,\omega,\tau,d,\gamma,\rho)$.

Following the benchmark of \cite{stockman2026earthquakenpp}, we apply the spatio-temporal coarse-grained Hawkes process to earthquake forecasting based on Eq.~\eqref{eq.ETASintensity}. We compare its forecasting performance with that of the ETAS model, Neural Spatio-Temporal Point Process (NSTPP) \cite{chen2021neural}, Deep Spatio-Temporal Point Process (DeepSTPP) \cite{zhou2022neural}, and Automatic Integration for Spatiotemporal Neural Point Processes (AutoSTPP) \cite{zhou2023automatic}.

The purpose of this experiment is twofold: to assess how much forecasting accuracy can be retained when only aggregated count data are available for parameter estimation, and to examine whether aggregation provides practical computational savings. The benchmark methods use complete earthquake data, including occurrence times and locations, for training. In contrast, our method uses only counts aggregated over spatio-temporal bins; earthquake magnitudes are not observed at the aggregated level and are therefore modeled according to the Gutenberg--Richter distribution,
\[
    p(M)=\beta e^{-\beta(M-M_0)}\bm{1}_{M\geq M_0},
\]
where $\beta$ is first estimated empirically from the training data.
The temporal and spatial coarse-grained kernels are then obtained as
\begin{equation}
\begin{aligned}
    \phi_{n-k}
    ={}&\frac{1}{\Delta t}
    \int_{(n-1)\Delta t}^{n\Delta t}
    \int_{(k-1)\Delta t}^{k\Delta t}
    \frac{e^{-(t-s)/\tau}}{(t-s+c)^{1+\omega}}
    \bm{1}_{t\geq s}\,ds\,dt,
\end{aligned}
\end{equation}
and
\begin{align}
    f_{ij}
    ={}&\frac{1}{|S_j|}
    \int_{S_i}\int_{S_j}
    \mathbb{E}_{M\sim p(M)}
    \left[
    k\frac{e^{a(M-M_0)}}
    {\left(\|\bm{x}-\bm{y}\|^2+d\,e^{\gamma(M-M_0)}\right)^{1+\rho}}
    \right]
    d\bm{y}\,d\bm{x}.
\end{align}
We use the uniform distribution within each spatial bin for $p_j$. This choice is consistent with the stationary-intensity-based distribution discussed in Section~\ref{subsec.asymptotics}, because the background intensity is spatially homogeneous and the spatial kernel is normalized and shift-invariant under periodic boundary conditions.

We use earthquake data from two regions in California in the QTM catalog collected by SCSN and reanalyzed by \cite{ross2019QTM}. The data are divided into training and test periods. We evaluate the models using test log-likelihood scores, with parameters estimated from the training period, and by measuring the training time.

Although our parameters are estimated using only aggregated count data, the estimated parameters are subsequently substituted into the continuous-time ETAS intensity, and the test log-likelihood is evaluated on the same complete test catalog used for the benchmark methods. Thus, the forecasting evaluation is performed on event-level data even though our parameter estimation uses only aggregated observations.

\begin{table}
    \centering
    \begin{tabular}{cccccc} \hline
        \textbf{Region} &\textbf{Ours} & \textbf{ETAS}  & \textbf{NSTPP} & \textbf{DeepSTPP} & \textbf{AutoSTPP} \\ \hline
        Salton Sea & $-2.745$ & $0.016$ & $-3.579$ & $-1.875$ & $-1.285$ \\ \hline
        San Jacinto & $-5.287$ & $-4.265$ & $-7.078$ & $-5.848$ & $-5.334$ \\ \hline
    \end{tabular}
    \caption{Test log-likelihood scores for the spatio-temporal coarse-grained Hawkes process, the ETAS model, and existing neural point-process models. Scores for the benchmark methods are taken from \cite{stockman2026earthquakenpp}. Our method is estimated using only aggregated count data, whereas the benchmark methods use complete event-time and location data for parameter estimation.}
    \label{tab:loglikelihood}
\end{table}

Table~\ref{tab:loglikelihood} summarizes the test log-likelihood scores. The spatio-temporal coarse-grained Hawkes process yields lower test log-likelihood scores than the ETAS model, as expected from the information loss caused by coarse-graining. Nevertheless, its scores are comparable to or better than those of the neural point-process models in these two regions, despite using substantially coarser information for parameter estimation.

Table~\ref{tab:trainingtime} reports the training times and the numbers of training events. Compared with ETAS, aggregation substantially reduces the effective data size for parameter estimation, so the computational cost of our method is governed primarily by the numbers of temporal and spatial bins rather than by the total number of observed events. Compared with the three neural point-process models, our method achieves competitive forecasting performance with substantially lower computational requirements. Although the computational environments differ across methods, these results demonstrate a practical trade-off among the information required for estimation, forecasting performance, and computational cost.

\begin{table}
    \centering
    \begin{tabular}{ccccccc} \hline
        \textbf{Region} & \textbf{Training Events} &\textbf{Ours} & \textbf{ETAS}  & \textbf{NSTPP} & \textbf{DeepSTPP} & \textbf{AutoSTPP} \\ \hline
        Salton Sea & $44,042$ & 00:00:44 & 07:28:28 & 2days, 00:26:45 & 00:26:46 & 01:45:34 \\ \hline
        San Jacinto & $18,664$ & 00:19:50 & 00:32:40 & 1days, 22:06:33 & 00:09:31 & 00:37:03 \\ \hline
    \end{tabular}
    \caption{Training times for parameter estimation. Our method was run on a Mac Pro (MacPro7,1) with a 16-core Intel Xeon W CPU at 3.2 GHz and 48 GB of memory. The benchmark methods were run on an HPC system equipped with four NVIDIA A100 40 GB SXM GPUs and an AMD EPYC 7543P 32-core processor, as reported in \cite{stockman2026earthquakenpp}.}
    \label{tab:trainingtime}
\end{table}

\section{Discussion}

This study introduced a class of discrete-time stochastic processes for spatio-temporally aggregated event-count data that retain the self-exciting behavior of continuous-time Hawkes processes. We defined this class by coarse-graining continuous-time Hawkes processes and requiring the resulting residual process to satisfy the first- and second-order conditional moment properties inherited from the martingale increments.

We evaluated the approximation errors of the stationary mean and autocovariance relative to those of the corresponding binned Hawkes process in the joint limit $\Delta t\to0$ and $|S|\to0$. The resulting asymptotic analysis shows that the proposed model has smaller discretization errors than the binned Poisson approximation. Numerical experiments further demonstrate that these errors can be one to two orders of magnitude smaller than those of the binned Poisson approximation, even for relatively coarse temporal and spatial discretizations.

We also applied the proposed framework to earthquake occurrence data modeled by the ETAS model, a continuous-time spatio-temporal Hawkes process. Although our method did not match the ETAS model in test log-likelihood, which is expected because parameter estimation was based only on aggregated occurrence data, it achieved predictive performance comparable to that of existing neural point-process models. At the same time, the proposed method required substantially less computational time in the experiments. These results suggest that the proposed framework can provide a practical alternative when event-level temporal and spatial information is unavailable or too costly to process.

The proposed framework has advantages over the binned Poisson approximation from both temporal and spatial perspectives. Temporally, it explicitly incorporates excitation among events occurring within the same temporal bin through the intra-bin excitation coefficient $\phi_0$. This effect is particularly important for temporal kernels with substantial mass near zero, such as exponential kernels, for which an event can strongly influence subsequent events within the same temporal bin. Spatially, the framework represents event locations within each spatial bin through a prescribed distribution. This allows available spatial information to be incorporated into the coarse-grained model and accommodates spatial bins with different sizes and shapes. Moreover, because the spatial excitation kernel is defined through an integral over each bin, the total spatial excitation is preserved under the normalization conditions imposed on the kernel. The numerical experiments indicate that these features make the proposed model relatively robust to coarse spatio-temporal aggregation, which is particularly useful when fine-scale event data are unavailable.

From a computational perspective, the proposed method provides a direct way to control computational cost through the temporal and spatial resolutions. A single evaluation of the loss function has computational complexity $O\left(n^2d+nd^2+d^3\right)$, where $n$ and $d$ denote the numbers of temporal and spatial bins, respectively. Importantly, this complexity depends on the number of bins rather than on the total number of events. Consequently, aggregation can substantially reduce the computational burden when the original event catalog is large.

This study has several limitations. First, the proposed process is defined only through its first- and second-order conditional moment properties rather than through a complete probability law. Consequently, likelihood-based inference and direct simulation of the process require additional distributional assumptions. Second, the theoretical properties of the proposed parameter estimator remain to be established. In particular, future work should investigate the consistency and asymptotic distribution of the estimator, as well as its dynamical and dependence properties. It would also be valuable to develop explicit distributional specifications for the coarse-grained process and to extend the framework to nonlinear and nonstationary Hawkes processes and to other forms of aggregated event data.

\begin{appendices}
\section{Proofs}
\subsection{Proofs of the first-order properties}
\begin{lem}\label{lem.weighttedsum-Fn}
    For every $n=0,1,2,\ldots$ and $i=1,\ldots,d$,
    \begin{equation}
        \sum_{j=1}^d|S_j|[F^n]_{ij}
        \leq\|f\|_\infty|A||S_i|.
    \end{equation}
    Moreover, for every $n=1,2,\ldots$ and $i,j=1,\ldots,d$,
    \begin{equation}
        [F^n]_{ij}\leq\|f\|_\infty|S_i|.
    \end{equation}
\end{lem}
\begin{proof}
    By Eq.~\eqref{eq.normalization-2}, for every $\bm{x}\in A$,
    \[
        1=\int_Af(\bm{x},\bm{y})d\bm{y}\leq\|f\|_\infty|A|.
    \]
    Therefore, for $n=0$, since $F^0=I$, we have
    \[
        \sum_{j=1}^d|S_j|[F^0]_{ij}=|S_i|\leq\|f\|_\infty|A||S_i|.
    \]
    For $n=1$, since $[F^1]_{ij}=f_{ij}$, we have
    \begin{align*}
        [F^1]_{ij}
        &=\int_{S_i}\int_{S_j}f(\bm{x},\bm{y})p_j(\bm{y})d\bm{y}d\bm{x}
        \leq\|f\|_\infty|S_i|.
    \end{align*}
    Therefore,
    \begin{align*}
        \sum_{j=1}^d|S_j|[F^1]_{ij}
        \leq\|f\|_\infty|S_i|\sum_{j=1}^d|S_j|
        =\|f\|_\infty|A||S_i|.
    \end{align*}
    Applying Eq.~\eqref{eq.spatial-cg-normalization} repeatedly yields, for every $n\geq1$ and $j=1,\ldots,d$,
    \begin{align*}
        \sum_{l=1}^d[F^n]_{lj}
        &=\sum_{m=1}^d\left(\sum_{l=1}^df_{lm}\right)[F^{n-1}]_{mj}
        =\sum_{m=1}^d[F^{n-1}]_{mj}
        =\cdots=1.
    \end{align*}
    Hence, for every $n\geq2$,
    \begin{align*}
        [F^n]_{ij}
        &=\sum_{l=1}^df_{il}[F^{n-1}]_{lj}
        \leq\|f\|_\infty|S_i|\sum_{l=1}^d[F^{n-1}]_{lj}
        =\|f\|_\infty|S_i|.
    \end{align*}
    Consequently,
    \begin{align*}
        \sum_{j=1}^d|S_j|[F^n]_{ij}
        \leq\|f\|_\infty|S_i|\sum_{j=1}^d|S_j|
        =\|f\|_\infty|A||S_i|
    \end{align*}
    for every $n\geq1$, which completes the proof.
\end{proof}
\begin{lem}\label{lem.expectation-convolution}
    For each $n=0,1,2,\ldots$ and $i=1,\ldots,d$,
    \begin{equation}
        \sum_{j=1}^d[F^n]_{ij}\mu_j
        =\int_{S_i}(f^{(*n)}*\mu)(\bm{x})d\bm{x}
        +O(n|S|^{1/q}|S_i|),
    \end{equation}
    where the constant in the $O$ term is independent of $n$, $i$, and the spatial partition.
\end{lem}
\begin{proof}
    For $n=0$, the assertion follows immediately from $F^0=I$ and $f^{(*0)}*\mu=\mu$. For $n=1$, we have
    \begin{align*}
        \sum_{j=1}^d[F^1]_{ij}\mu_j
        &=\sum_{j=1}^d\int_{S_i}\int_{S_j}f(\bm{x},\bm{y})p_j(\bm{y})d\bm{y}d\bm{x}\int_{S_j}\mu(\bm{z})d\bm{z}\\
        &=\sum_{j=1}^d\int_{S_i}\int_{S_j}f(\bm{x},\bm{y})\mu(\bm{y})d\bm{y}d\bm{x}\\
        &\quad-\sum_{j=1}^d\int_{S_i}\int_{S_j}\left(f(\bm{x},\bm{y})-\int_{S_j}f(\bm{x},\bm{z})p_j(\bm{z})d\bm{z}\right)\mu(\bm{y})d\bm{y}d\bm{x}\\
        &=\int_{S_i}(f*\mu)(\bm{x})d\bm{x}
        -\sum_{j=1}^d\int_{S_i}\int_{S_j}\left(f(\bm{x},\bm{y})-\mathbb{E}_{P_j}[f(\bm{x},\bm{Y})]\right)\mu(\bm{y})d\bm{y}d\bm{x},
    \end{align*}
    where $P_j$ denotes the distribution with density $p_j$. Similarly, for $n=2$,
    \begin{align*}
        \sum_{j=1}^d[F^2]_{ij}\mu_j
        &=\sum_{j=1}^d\left(\sum_{k=1}^df_{ik}f_{kj}\right)\mu_j
        =\sum_{k=1}^df_{ik}\left(\sum_{j=1}^df_{kj}\mu_j\right)\\
        &=\sum_{k=1}^d\int_{S_i}\int_{S_k}f(\bm{x},\bm{y})p_k(\bm{y})d\bm{y}d\bm{x}\int_{S_k}(f*\mu)(\bm{y})d\bm{y}\\
        &\quad-\sum_{k=1}^d\sum_{j=1}^df_{ik}\int_{S_k}\int_{S_j}\left(f(\bm{y},\bm{z})-\mathbb{E}_{P_j}[f(\bm{y},\bm{Z})]\right)\mu(\bm{z})d\bm{z}d\bm{y}\\
        &=\int_{S_i}(f^{(*2)}*\mu)(\bm{x})d\bm{x}\\
        &\quad-\sum_{k=1}^d\int_{S_i}\int_{S_k}\left(f(\bm{x},\bm{y})-\mathbb{E}_{P_k}[f(\bm{x},\bm{Y})]\right)(f*\mu)(\bm{y})d\bm{y}d\bm{x}\\
        &\quad-\sum_{k=1}^d\sum_{j=1}^d[F^1]_{ik}\int_{S_k}\int_{S_j}\left(f(\bm{y},\bm{z})-\mathbb{E}_{P_j}[f(\bm{y},\bm{Z})]\right)\mu(\bm{z})d\bm{z}d\bm{y}.
    \end{align*}
    Repeating this calculation gives
    \begin{align*}
        \sum_{j=1}^d[F^n]_{ij}\mu_j
        &=\int_{S_i}(f^{(*n)}*\mu)(\bm{x})d\bm{x}\\
        &\quad-\sum_{m=1}^n\sum_{k=1}^d\sum_{j=1}^d[F^{m-1}]_{ik}
        \int_{S_k}\int_{S_j}\left(f(\bm{y},\bm{z})-\mathbb{E}_{P_j}[f(\bm{y},\bm{Z})]\right)\\
        &\qquad\times(f^{*(n-m)}*\mu)(\bm{z})d\bm{z}d\bm{y}.
    \end{align*}
    By the Lipschitz continuity of $f$ with respect to its second argument and the regularity condition in Eq.~\eqref{eq.regularity}, there exists a constant $C_r>0$ such that
    \begin{align*}
        \left|f(\bm{y},\bm{z})-\mathbb{E}_{P_j}[f(\bm{y},\bm{Z})]\right|
        &\leq\mathbb{E}_{P_j}\left[\left|f(\bm{y},\bm{z})-f(\bm{y},\bm{Z})\right|\right]\\
        &\leq L_y\mathbb{E}_{P_j}\left[\|\bm{z}-\bm{Z}\|\right]
        \leq C_rL_y|S|^{1/q}.
    \end{align*}
    Moreover, Eq.~\eqref{eq.normalization-1} implies that, for $n-m\geq1$,
    \begin{align*}
        \sum_{j=1}^d\int_{S_j}(f^{*(n-m)}*\mu)(\bm{z})d\bm{z}
        &=\int_A\int_Af(\bm{z},\bm{w})(f^{*(n-m-1)}*\mu)(\bm{w})d\bm{w}d\bm{z}\\
        &=\int_A(f^{*(n-m-1)}*\mu)(\bm{w})d\bm{w}
        =\cdots=\int_A\mu(\bm{z})d\bm{z}.
    \end{align*}
    The same equality is immediate when $n-m=0$. Therefore, the absolute value of the error term is bounded by
    \begin{align*}
        &\sum_{m=1}^n\sum_{k=1}^d\sum_{j=1}^d[F^{m-1}]_{ik}
        \int_{S_k}\int_{S_j}\left|f(\bm{y},\bm{z})-\mathbb{E}_{P_j}[f(\bm{y},\bm{Z})]\right|\\
        &\qquad\times(f^{*(n-m)}*\mu)(\bm{z})d\bm{z}d\bm{y}\\
        &\quad\leq C_rL_y|S|^{1/q}\sum_{m=1}^n\sum_{k=1}^d\sum_{j=1}^d[F^{m-1}]_{ik}
        \int_{S_k}\int_{S_j}(f^{*(n-m)}*\mu)(\bm{z})d\bm{z}d\bm{y}\\
        &\quad=C_rL_y|S|^{1/q}\sum_{m=1}^n\sum_{k=1}^d[F^{m-1}]_{ik}|S_k|\|\mu\|_{1}\\
        &\quad\leq nC_rL_y\|f\|_\infty|A|\|\mu\|_{1}|S|^{1/q}|S_i|,
    \end{align*}
    where the last inequality follows from Lemma~\ref{lem.weighttedsum-Fn}. Hence, the assertion follows.
\end{proof}

\subsection{Proofs of the second-order properties}\label{appendix.2nd-order}
In this subsection, we first establish several lemmas needed for the proof of Theorem~\ref{thm.cov} and then present the proof of the theorem.

\begin{lem}\label{lem.phi-convolution}
    For $l=0,1,2,\ldots$ and $i=0,1,2,\ldots$,
    \begin{align}
        &\left|\dfrac{1}{\Delta t}\int_{(n-1)\Delta t}^{n\Delta t}
        \int_{(n-i-1)\Delta t}^{(n-i)\Delta t}
        \phi^{(*l)}(t-s)dsdt-(\phi^{(*l)})_i\right|\nonumber\\
        &\quad\leq
        \begin{cases}
            0, & l=0,1,\\
            Cl\alpha^{l-2}(\Delta t)^2, & l\geq2,
        \end{cases}
    \end{align}
    where $C>0$ is independent of $l$, $i$, and $\Delta t$. Moreover, when $i=0$ and $l\geq2$, the left-hand side is also $O((\Delta t)^l)$.
\end{lem}
\begin{proof}
    We proceed by induction on $l$. For $l=0,1$, the assertion follows immediately from the definitions. We extend $\phi$ by zero on $\mathbb{R}_{-}$. For $i\geq1$ and $l\geq2$, the integral can be written as
    \[
        \dfrac{1}{\Delta t}\int_{(n-1)\Delta t}^{n\Delta t}
        \int_{(n-i-1)\Delta t}^{(n-i)\Delta t}\phi^{(*l)}(t-s)dsdt=\dfrac{1}{\Delta t}\int_{(i-1)\Delta t}^{(i+1)\Delta t}
        (\Delta t-|t-i\Delta t|)\phi^{(*l)}(t)dt.
    \]

    For $l=2$, $(\phi^{(*2)})_i$ can be written as
    \begin{align*}
        (\phi^{(*2)})_i
        &=\sum_{j=0}^i\phi_{i-j}\phi_j\\
        &=\sum_{j=0}^i\dfrac{1}{\Delta t}
        \int_{(i-j)\Delta t}^{(i-j+1)\Delta t}\int_0^{\Delta t}\phi(t-u)dudt
        \dfrac{1}{\Delta t}\int_{j\Delta t}^{(j+1)\Delta t}\int_0^{\Delta t}\phi(s-v)dvds\\
        &=\sum_{j=0}^i\dfrac{1}{\Delta t}
        \int_{(i-j-1)\Delta t}^{(i-j+1)\Delta t}
        (\Delta t-|t-(i-j)\Delta t|)\phi(t)dt\\
        &\qquad\times\dfrac{1}{\Delta t}
        \int_{(j-1)\Delta t}^{(j+1)\Delta t}
        (\Delta t-|s-j\Delta t|)\phi(s)ds\\
        &=\sum_{j=0}^i\dfrac{1}{\Delta t}
        \int_{(i-1)\Delta t}^{(i+1)\Delta t}
        (\Delta t-|t-i\Delta t|)\phi(t-j\Delta t)dt\\
        &\qquad\times\dfrac{1}{\Delta t}
        \int_{(j-1)\Delta t}^{(j+1)\Delta t}
        (\Delta t-|s-j\Delta t|)\phi(s)ds\\
        &=\dfrac{1}{\Delta t}\int_{(i-1)\Delta t}^{(i+1)\Delta t}
        (\Delta t-|t-i\Delta t|)
        \sum_{j=0}^i\int_{(j-1)\Delta t}^{(j+1)\Delta t}
        \dfrac{\Delta t-|s-j\Delta t|}{\Delta t}
        \phi(t-j\Delta t)\phi(s)dsdt\\
        &=\dfrac{1}{\Delta t}\int_{(i-1)\Delta t}^{(i+1)\Delta t}
        (\Delta t-|t-i\Delta t|)
        \sum_{j=0}^i\int_{(j-1)\Delta t}^{(j+1)\Delta t}
        \dfrac{\Delta t-|s-j\Delta t|}{\Delta t}
        \phi(t-s)\phi(s)dsdt\\
        &\quad+\dfrac{1}{\Delta t}\int_{(i-1)\Delta t}^{(i+1)\Delta t}
        (\Delta t-|t-i\Delta t|)
        \sum_{j=0}^i\int_{(j-1)\Delta t}^{(j+1)\Delta t}
        \dfrac{\Delta t-|s-j\Delta t|}{\Delta t}\\
        &\qquad\times\left(\phi(t-j\Delta t)-\phi(t-s)\right)\phi(s)dsdt.
    \end{align*}

    For every $t\in[(i-1)\Delta t,(i+1)\Delta t]$, the integrand in the first term satisfies
    \begin{align*}
        &\left|\phi^{(*2)}(t)-\sum_{j=0}^i\int_{(j-1)\Delta t}^{(j+1)\Delta t}
        \dfrac{\Delta t-|s-j\Delta t|}{\Delta t}\phi(t-s)\phi(s)ds\right|\\
        &=\left|\int_0^t\phi(t-s)\phi(s)ds-\int_0^{i\Delta t}\phi(t-s)\phi(s)ds
        -\int_{i\Delta t}^{(i+1)\Delta t}
        \dfrac{\Delta t-|s-i\Delta t|}{\Delta t}\phi(t-s)\phi(s)ds\right|\\
        &\leq\left|\int_{i\Delta t}^t\phi(t-s)\phi(s)ds\right|
        +\int_{i\Delta t}^{(i+1)\Delta t}
        \dfrac{\Delta t-|s-i\Delta t|}{\Delta t}\phi(t-s)\phi(s)ds\\
        &\leq2\Delta t\|\phi\|_\infty^2.
    \end{align*}
    Therefore,
    \begin{align*}
        &\left|\dfrac{1}{\Delta t}\int_{(i-1)\Delta t}^{(i+1)\Delta t}
        (\Delta t-|t-i\Delta t|)\phi^{(*2)}(t)dt\right.\\
        &\quad\left.-\dfrac{1}{\Delta t}\int_{(i-1)\Delta t}^{(i+1)\Delta t}
        (\Delta t-|t-i\Delta t|)
        \sum_{j=0}^i\int_{(j-1)\Delta t}^{(j+1)\Delta t}
        \dfrac{\Delta t-|s-j\Delta t|}{\Delta t}\phi(t-s)\phi(s)dsdt\right|\\
        &\leq\dfrac{1}{\Delta t}\times(\Delta t)^2\times
        2\Delta t\|\phi\|_\infty^2
        =2\|\phi\|_\infty^2(\Delta t)^2.
    \end{align*}

    Taking absolute values and decomposing the sum over $j$ into the ranges $j=0,\ldots,i-2$ and $j=i-1,i$, we bound the second term by
    \begin{align*}
        &\dfrac{1}{\Delta t}\int_{(i-1)\Delta t}^{(i+1)\Delta t}
        (\Delta t-|t-i\Delta t|)
        \sum_{j=0}^{i-2}\int_{(j-1)\Delta t}^{(j+1)\Delta t}
        \dfrac{\Delta t-|s-j\Delta t|}{\Delta t}\\
        &\qquad\times|\phi(t-j\Delta t)-\phi(t-s)|\phi(s)dsdt\\
        &\quad+\dfrac{1}{\Delta t}\int_{(i-1)\Delta t}^{(i+1)\Delta t}
        (\Delta t-|t-i\Delta t|)
        \sum_{j=i-1}^{i}\int_{(j-1)\Delta t}^{(j+1)\Delta t}
        \dfrac{\Delta t-|s-j\Delta t|}{\Delta t}\\
        &\qquad\times|\phi(t-j\Delta t)-\phi(t-s)|\phi(s)dsdt,
    \end{align*}
    where the first sum is understood to be empty when $i=1$. By the Cauchy--Schwarz inequality,
    \begin{align*}
        &\sum_{j=0}^{i-2}\int_{(j-1)\Delta t}^{(j+1)\Delta t}
        \dfrac{\Delta t-|s-j\Delta t|}{\Delta t}
        |\phi(t-j\Delta t)-\phi(t-s)|\phi(s)ds\\
        &\leq\left(\sum_{j=0}^{i-2}\int_{(j-1)\Delta t}^{(j+1)\Delta t}
        \left(\dfrac{\Delta t-|s-j\Delta t|}{\Delta t}\right)^2
        |\phi(t-j\Delta t)-\phi(t-s)|^2\phi(s)ds\right)^{1/2}\\
        &\qquad\times\left(\sum_{j=0}^{i-2}
        \int_{(j-1)\Delta t}^{(j+1)\Delta t}\phi(s)ds\right)^{1/2}\\
        &\leq\sqrt{2\alpha\|\phi\|_\infty}\\
        &\qquad\times\left(\sum_{j=0}^{i-2}\int_{(j-1)\Delta t}^{(j+1)\Delta t}
        \left(\dfrac{\Delta t-|s-j\Delta t|}{\Delta t}\right)^2
        |s-j\Delta t|\left|\int_{t-s}^{t-j\Delta t}|\phi'(u)|^2du\right|ds\right)^{1/2}\\
        &=\sqrt{2\alpha\|\phi\|_\infty}
        \left(\sum_{j=0}^{i-2}\int_{-\Delta t}^{\Delta t}
        \left(1-\dfrac{|s|}{\Delta t}\right)^2|s|
        \left|\int_{t-s-j\Delta t}^{t-j\Delta t}|\phi'(u)|^2du\right|ds\right)^{1/2}\\
        &=\sqrt{2\alpha\|\phi\|_\infty}
        \left(\int_{-\Delta t}^{\Delta t}
        \left(1-\dfrac{|s|}{\Delta t}\right)^2|s|
        \int_0^\infty\sum_{j=0}^{i-2}\bm I_j(s,t,u)|\phi'(u)|^2duds\right)^{1/2},
    \end{align*}
    where
    \[
        \bm I_j(s,t,u)=
        \begin{cases}
            1, & u\text{ is between }t-j\Delta t\text{ and }t-s-j\Delta t,\\
            0, & \text{otherwise}.
        \end{cases}
    \]
    For fixed $s$, $t$, and $u$, the intervals corresponding to $\bm I_j$ have uniformly bounded overlap, and in particular
    \[
        \sum_{j=0}^{i-2}\bm I_j(s,t,u)\leq3.
    \]
    Hence,
    \begin{align*}
        &\sum_{j=0}^{i-2}\int_{(j-1)\Delta t}^{(j+1)\Delta t}
        \dfrac{\Delta t-|s-j\Delta t|}{\Delta t}
        |\phi(t-j\Delta t)-\phi(t-s)|\phi(s)ds\\
        &\leq\sqrt{6\alpha\|\phi\|_\infty}
        \left(\int_{-\Delta t}^{\Delta t}
        \left(1-\dfrac{|s|}{\Delta t}\right)^2|s|ds
        \int_0^\infty|\phi'(u)|^2du\right)^{1/2}\\
        &=\sqrt{\alpha\|\phi\|_\infty}\|\phi'\|_2\Delta t.
    \end{align*}
    Here, we used
    \begin{align*}
        |\phi(t-j\Delta t)-\phi(t-s)|^2
        &=\left|\int_{t-s}^{t-j\Delta t}\phi'(u)du\right|^2\\
        &\leq|s-j\Delta t|
        \left|\int_{t-s}^{t-j\Delta t}|\phi'(u)|^2du\right|.
    \end{align*}
    Therefore,
    \begin{align*}
        &\dfrac{1}{\Delta t}\int_{(i-1)\Delta t}^{(i+1)\Delta t}
        (\Delta t-|t-i\Delta t|)
        \sum_{j=0}^{i-2}\int_{(j-1)\Delta t}^{(j+1)\Delta t}
        \dfrac{\Delta t-|s-j\Delta t|}{\Delta t}\\
        &\qquad\times|\phi(t-j\Delta t)-\phi(t-s)|\phi(s)dsdt\\
        &\leq\dfrac{1}{\Delta t}\times(\Delta t)^2\times
        \sqrt{\alpha\|\phi\|_\infty}\|\phi'\|_2\Delta t\\
        &=\sqrt{\alpha\|\phi\|_\infty}\|\phi'\|_2(\Delta t)^2.
    \end{align*}

    The second sum is bounded by
    \begin{align*}
        &\dfrac{1}{\Delta t}\int_{(i-1)\Delta t}^{(i+1)\Delta t}
        (\Delta t-|t-i\Delta t|)
        \sum_{j=i-1}^{i}\int_{(j-1)\Delta t}^{(j+1)\Delta t}
        \dfrac{\Delta t-|s-j\Delta t|}{\Delta t}\\
        &\qquad\times|\phi(t-j\Delta t)-\phi(t-s)|\phi(s)dsdt\\
        &\leq\dfrac{1}{\Delta t}\int_{(i-1)\Delta t}^{(i+1)\Delta t}
        (\Delta t-|t-i\Delta t|)dt
        \times2\|\phi\|_\infty^2
        \sum_{j=i-1}^{i}\int_{(j-1)\Delta t}^{(j+1)\Delta t}ds\\
        &\leq\dfrac{1}{\Delta t}\times(\Delta t)^2\times
        2\|\phi\|_\infty^2\times4\Delta t
        =8\|\phi\|_\infty^2(\Delta t)^2.
    \end{align*}
    Thus, for $l=2$, the total error is bounded by
    \[
        \left(10\|\phi\|_\infty^2+
        \sqrt{\alpha\|\phi\|_\infty}\|\phi'\|_2\right)(\Delta t)^2.
    \]

        Suppose that the assertion holds for $l=k\geq2$. From the discrete convolution,
    \[
        (\phi^{*(k+1)})_i=\sum_{j=0}^i(\phi^{(*k)})_{i-j}\phi_j.
    \]
    Inserting the corresponding averaged continuous convolution gives
    \begin{align*}
        &\left|\dfrac{1}{\Delta t}\int_{(i-1)\Delta t}^{(i+1)\Delta t}
        (\Delta t-|t-i\Delta t|)\phi^{*(k+1)}(t)dt
        -(\phi^{*(k+1)})_i\right|\\
        &\leq
        \left|\dfrac{1}{\Delta t}\int_{(i-1)\Delta t}^{(i+1)\Delta t}
        (\Delta t-|t-i\Delta t|)\phi^{*(k+1)}(t)dt\right.\\
        &\qquad\left.-\sum_{j=0}^i\left\{\dfrac{1}{\Delta t}
        \int_{(i-j)\Delta t}^{(i-j+1)\Delta t}
        \int_0^{\Delta t}\phi^{(*k)}(t-u)dudt\right\}\phi_j\right|\\
        &\quad+
        \left|\sum_{j=0}^i\left\{\dfrac{1}{\Delta t}
        \int_{(i-j)\Delta t}^{(i-j+1)\Delta t}
        \int_0^{\Delta t}\phi^{(*k)}(t-u)dudt\right\}\phi_j
        -\sum_{j=0}^i(\phi^{(*k)})_{i-j}\phi_j\right|.
    \end{align*}
    By the induction hypothesis, the second term is bounded by
    \begin{align*}
        &\sum_{j=0}^i\left|
        \dfrac{1}{\Delta t}\int_{(i-j)\Delta t}^{(i-j+1)\Delta t}
        \int_0^{\Delta t}\phi^{(*k)}(t-u)dudt
        -(\phi^{(*k)})_{i-j}\right|\phi_j\\
        &\quad\leq Ck\alpha^{k-2}(\Delta t)^2\sum_{j=0}^i\phi_j
        \leq Ck\alpha^{k-1}(\Delta t)^2.
    \end{align*}
    Applying the same argument as in the case $l=2$, with $\phi^{(*k)}$ in place of the first factor $\phi$, the first term is bounded by
    \[
        \left(10\|\phi^{(*k)}\|_\infty\|\phi\|_\infty
        +\sqrt{\alpha\|\phi\|_\infty}\|(\phi^{(*k)})'\|_2\right)(\Delta t)^2.
    \]
    By Young's convolution inequality,
    \[
        \|\phi^{(*k)}\|_\infty\leq\alpha^{k-1}\|\phi\|_\infty
    \]
    and
    \begin{align*}
        \|(\phi^{(*k)})'\|_2
        &\leq\alpha^{k-2}\|(\phi^{(*2)})'\|_2\\
        &\leq\alpha^{k-2}\left(\alpha\|\phi'\|_2+|\phi(0)|\|\phi\|_2\right)\\
        &\leq\alpha^{k-2}\left(\alpha\|\phi'\|_2
        +|\phi(0)|\sqrt{\alpha\|\phi\|_\infty}\right).
    \end{align*}
    Hence, the first term is bounded by
    \[
        \alpha^{k-1}\left(10\|\phi\|_\infty^2
        +\sqrt{\alpha\|\phi\|_\infty}\|\phi'\|_2
        +|\phi(0)|\|\phi\|_\infty\right)(\Delta t)^2.
    \]
    Taking $C$ sufficiently large so that
    \[
        10\|\phi\|_\infty^2
        +\sqrt{\alpha\|\phi\|_\infty}\|\phi'\|_2
        +|\phi(0)|\|\phi\|_\infty\leq C,
    \]
    the first term is bounded by $C\alpha^{k-1}(\Delta t)^2$. Combining the two bounds gives
    \[
        C\alpha^{k-1}(\Delta t)^2
        +Ck\alpha^{k-1}(\Delta t)^2=C(k+1)\alpha^{k-1}(\Delta t)^2.
    \]
    Therefore, the assertion holds for $l=k+1$.

    Finally, let $i=0$ and $l\geq2$. Repeated convolution gives
    \[
        \phi^{(*l)}(t)\leq\dfrac{\|\phi\|_\infty^l}{(l-1)!}t^{l-1}.
    \]
    Therefore,
    \begin{align*}
        \dfrac{1}{\Delta t}\int_0^{\Delta t}\int_0^t
        \phi^{(*l)}(t-s)dsdt
        &\leq\dfrac{\|\phi\|_\infty^l}{(l-1)!\Delta t}
        \int_0^{\Delta t}\int_0^t(t-s)^{l-1}dsdt\\
        &=O((\Delta t)^l).
    \end{align*}
    On the other hand,
    \[
        (\phi^{(*l)})_0=\phi_0^l=O((\Delta t)^l).
    \]
    Thus, the error is $O((\Delta t)^l)$. For sufficiently small $\Delta t$, this also satisfies the bound $Cl\alpha^{l-2}(\Delta t)^2$ after enlarging $C$ if necessary. This completes the proof.
\end{proof}

\begin{lem}\label{lem.phiconvolution-bound}
    For each $l=0,1,2,\ldots$,
    \begin{equation}
        \sum_{m=0}^{\infty}(\phi^{(*l)})_m=\alpha^l.
    \end{equation}
\end{lem}
\begin{proof}
    For $l=0$, we have
    \[
        \sum_{m=0}^{\infty}(\phi^{(*0)})_m
        =\sum_{m=0}^{\infty}\delta_{m,0}
        =1.
    \]
    For $l=1,2,\ldots$, changing the order of summation and setting $r=m-n$, we obtain
    \begin{align*}
        \sum_{m=0}^{\infty}(\phi^{(*l)})_m
        &=\sum_{m=0}^{\infty}\sum_{n=0}^m\phi_n(\phi^{*(l-1)})_{m-n}\\
        &=\sum_{n=0}^{\infty}\phi_n\sum_{r=0}^{\infty}(\phi^{*(l-1)})_r\\
        &=\alpha\sum_{r=0}^{\infty}(\phi^{*(l-1)})_r.
    \end{align*}
    Applying this relation repeatedly, we obtain
    \[
        \sum_{m=0}^{\infty}(\phi^{(*l)})_m=\alpha^l.
    \]
    This completes the proof.
\end{proof}

\begin{lem}\label{lem.decompose-time}
    Let $m,n\in\mathbb{Z}$ with $m\leq n$, and let $p,l,l'=0,1,2,\ldots$. Then, the following expansions hold as $\Delta t\to 0$.
    \begin{enumerate}[(i)]
        \item For $m<n$,
        \begin{align}
            &\int_{(m-1)\Delta t}^{m\Delta t}
            \left(\int_{(n-1)\Delta t}^{n\Delta t}\phi^{(*l)}(t-s)dt\right)
            \left(\int_{(n+p-1)\Delta t}^{(n+p)\Delta t}\phi^{(*l')}(t'-s)dt'\right)ds \nonumber\\
            &=\Delta t(\phi^{(*l)})_{n-m}(\phi^{(*l')})_{n+p-m}\nonumber\\
            &\quad+\begin{cases}
                \begin{aligned}
                    &(\phi^{(*l)})_{n-m}\bm{1}_{l'\geq2}l'\alpha^{l'-2}O((\Delta t)^3)\\
                    &\quad+\bm{1}_{l\geq2}l\alpha^{l-2}O((\Delta t)^3)
                    \dfrac{1}{\Delta t}\int_{(n+p-1)\Delta t}^{(n+p)\Delta t}
                    \int_{(m-1)\Delta t}^{m\Delta t}\phi^{(*l')}(t'-s')ds'dt'\\
                    &\quad+O\left((\Delta t)^2
                    \|\partial_sD_{l,m,n}\|_{L^2((m-1)\Delta t,m\Delta t]}
                    \|\partial_sD_{l',m,n+p}\|_{L^2((m-1)\Delta t,m\Delta t]}\right),
                \end{aligned}
                & ll'>0,\\
                0, & ll'=0,
            \end{cases}
            \label{eq.decompose-time-error}
        \end{align}
        where
        \[
            D_{l,m,n}(s)
            =\int_{(n-1)\Delta t}^{n\Delta t}
            \left(\phi^{(*l)}(t-s)
            -\dfrac{1}{\Delta t}\int_{(m-1)\Delta t}^{m\Delta t}
            \phi^{(*l)}(t-s')ds'\right)dt.
        \]

        \item For $m=n$,
        \begin{align}
            &\int_{(n-1)\Delta t}^{n\Delta t}
            \left(\int_{(n-1)\Delta t}^{n\Delta t}\phi^{(*l)}(t-s)dt\right)
            \left(\int_{(n+p-1)\Delta t}^{(n+p)\Delta t}\phi^{(*l')}(t'-s)dt'\right)ds \nonumber\\
            &=\Delta t(\phi^{(*l)})_0(\phi^{(*l')})_p\nonumber\\
            &\quad+\begin{cases}
                \left(\bm{1}_{l\geq2}l\alpha^{l-2}
                +\bm{1}_{l'\geq2}l'\alpha^{l'-2}\right)O((\Delta t)^3),
                & ll'=0,\\
                \left(\bm{1}_{l\geq2}l+\bm{1}_{l'\geq2}l'\right)
                \alpha^{l+l'-3}O((\Delta t)^4),
                & ll'>0\text{ and }(l,l')\neq(1,1),\\
                O((\Delta t)^{7/2}),
                & p>0\text{ and }(l,l')=(1,1),\\
                O((\Delta t)^3),
                & p=0\text{ and }(l,l')=(1,1).
            \end{cases}
        \end{align}
    \end{enumerate}
    The constants implicit in the $O$ terms depend only on $\phi$ and are independent of $m$, $n$, $p$, $l$, $l'$, and $\Delta t$.
\end{lem}
\begin{proof}
    We first consider the case $ll'=0$. If $l=0$, then, for almost every $s\in((m-1)\Delta t,m\Delta t]$,
    \[
        \int_{(n-1)\Delta t}^{n\Delta t}\phi^{(*0)}(t-s)dt
        =\int_{(n-1)\Delta t}^{n\Delta t}\delta(t-s)dt
        =\delta_{n,m},
        \qquad
        (\phi^{(*0)})_{n-m}=\delta_{n,m}.
    \]
    Therefore, Lemma~\ref{lem.phi-convolution} gives
    \begin{align*}
        &\int_{(m-1)\Delta t}^{m\Delta t}
        \left(\int_{(n-1)\Delta t}^{n\Delta t}\phi^{(*0)}(t-s)dt\right)
        \left(\int_{(n+p-1)\Delta t}^{(n+p)\Delta t}\phi^{(*l')}(t'-s)dt'\right)ds\\
        &=\delta_{n,m}\int_{(m-1)\Delta t}^{m\Delta t}
        \int_{(n+p-1)\Delta t}^{(n+p)\Delta t}\phi^{(*l')}(t'-s)dt'ds\\
        &=\Delta t(\phi^{(*0)})_{n-m}(\phi^{(*l')})_{n+p-m}
        +\begin{cases}
            \delta_{n,m}l'\alpha^{l'-2}O((\Delta t)^3), & l'\geq2,\\
            0, & l'=0,1.
        \end{cases}
    \end{align*}

    Similarly, if $l'=0$, then
    \[
        \int_{(n+p-1)\Delta t}^{(n+p)\Delta t}\phi^{(*0)}(t'-s)dt'
        =\delta_{n+p,m},
        \qquad
        (\phi^{(*0)})_{n+p-m}=\delta_{n+p,m}
    \]
    for almost every $s\in((m-1)\Delta t,m\Delta t]$. Hence,
    \begin{align*}
        &\int_{(m-1)\Delta t}^{m\Delta t}
        \left(\int_{(n-1)\Delta t}^{n\Delta t}\phi^{(*l)}(t-s)dt\right)
        \left(\int_{(n+p-1)\Delta t}^{(n+p)\Delta t}\phi^{(*0)}(t'-s)dt'\right)ds\\
        &=\delta_{n+p,m}\int_{(m-1)\Delta t}^{m\Delta t}
        \int_{(n-1)\Delta t}^{n\Delta t}\phi^{(*l)}(t-s)dtds\\
        &=\Delta t(\phi^{(*l)})_{n-m}(\phi^{(*0)})_{n+p-m}
        +\begin{cases}
            \delta_{n,m}l\alpha^{l-2}O((\Delta t)^3), & p=0\text{ and }l\geq2,\\
            0, & \text{otherwise}.
        \end{cases}
    \end{align*}
    These identities prove the assertion when $ll'=0$.\\

    We next consider the case $ll'>0$. For convenience, define the bin-averaged coefficient
    \[
        B_{l,m,n}
        :=\dfrac{1}{\Delta t}\int_{(m-1)\Delta t}^{m\Delta t}
        \int_{(n-1)\Delta t}^{n\Delta t}\phi^{(*l)}(t-s')dtds'
    \]
    and
    \[
        D_{l,m,n}(s)
        :=\int_{(n-1)\Delta t}^{n\Delta t}\phi^{(*l)}(t-s)dt-B_{l,m,n}.
    \]
    By definition,
    \[
        \int_{(m-1)\Delta t}^{m\Delta t}D_{l,m,n}(s)ds=0.
    \]
    Therefore,
    \begin{align*}
        &\int_{(m-1)\Delta t}^{m\Delta t}
        \left(\int_{(n-1)\Delta t}^{n\Delta t}\phi^{(*l)}(t-s)dt\right)
        \left(\int_{(n+p-1)\Delta t}^{(n+p)\Delta t}\phi^{(*l')}(t'-s)dt'\right)ds\\
        &=\int_{(m-1)\Delta t}^{m\Delta t}
        \left(B_{l,m,n}+D_{l,m,n}(s)\right)
        \left(B_{l',m,n+p}+D_{l',m,n+p}(s)\right)ds\\
        &=\Delta t B_{l,m,n}B_{l',m,n+p}
        +\int_{(m-1)\Delta t}^{m\Delta t}
        D_{l,m,n}(s)D_{l',m,n+p}(s)ds.
    \end{align*}
    From Lemma~\ref{lem.phi-convolution},
    \[
        B_{l,m,n}
        =(\phi^{(*l)})_{n-m}
        +\bm{1}_{l\geq2}l\alpha^{l-2}O((\Delta t)^2).
    \]
    Hence,
    \begin{align*}
        &\Delta t B_{l,m,n}B_{l',m,n+p}\\
        &=\Delta t
        \left\{(\phi^{(*l)})_{n-m}
        +\bm{1}_{l\geq2}l\alpha^{l-2}O((\Delta t)^2)\right\}
        B_{l',m,n+p}\\
        &=\Delta t(\phi^{(*l)})_{n-m}
        \left\{(\phi^{(*l')})_{n+p-m}
        +\bm{1}_{l'\geq2}l'\alpha^{l'-2}O((\Delta t)^2)\right\}\\
        &\quad+\bm{1}_{l\geq2}l\alpha^{l-2}O((\Delta t)^3)B_{l',m,n+p}\\
        &=\Delta t(\phi^{(*l)})_{n-m}(\phi^{(*l')})_{n+p-m}\\
        &\quad+(\phi^{(*l)})_{n-m}\bm{1}_{l'\geq2}l'\alpha^{l'-2}O((\Delta t)^3)\\
        &\quad+\bm{1}_{l\geq2}l\alpha^{l-2}O((\Delta t)^3)
        \dfrac{1}{\Delta t}\int_{(m-1)\Delta t}^{m\Delta t}
        \int_{(n+p-1)\Delta t}^{(n+p)\Delta t}
        \phi^{(*l')}(t'-s')dt'ds'.
    \end{align*}
    Consequently,
    \begin{align*}
        &\int_{(m-1)\Delta t}^{m\Delta t}
        \left(\int_{(n-1)\Delta t}^{n\Delta t}\phi^{(*l)}(t-s)dt\right)
        \left(\int_{(n+p-1)\Delta t}^{(n+p)\Delta t}\phi^{(*l')}(t'-s)dt'\right)ds\\
        &=\Delta t(\phi^{(*l)})_{n-m}(\phi^{(*l')})_{n+p-m}\\
        &\quad+(\phi^{(*l)})_{n-m}\bm{1}_{l'\geq2}l'\alpha^{l'-2}O((\Delta t)^3)\\
        &\quad+\bm{1}_{l\geq2}l\alpha^{l-2}O((\Delta t)^3)
        \dfrac{1}{\Delta t}\int_{(m-1)\Delta t}^{m\Delta t}
        \int_{(n+p-1)\Delta t}^{(n+p)\Delta t}
        \phi^{(*l')}(t'-s')dt'ds'\\
        &\quad+\int_{(m-1)\Delta t}^{m\Delta t}
        D_{l,m,n}(s)D_{l',m,n+p}(s)ds.
    \end{align*}

    For $m<n$, the functions $D_{l,m,n}$ and $D_{l',m,n+p}$ belong to
    $W^{1,2}((m-1)\Delta t,m\Delta t]$ and have zero mean on this interval.
    Therefore, the Cauchy--Schwarz and Poincaré inequalities give
    \begin{align*}
        &\left|\int_{(m-1)\Delta t}^{m\Delta t}
        D_{l,m,n}(s)D_{l',m,n+p}(s)ds\right|\\
        &\quad\leq
        \|D_{l,m,n}\|_{L^2((m-1)\Delta t,m\Delta t]}
        \|D_{l',m,n+p}\|_{L^2((m-1)\Delta t,m\Delta t]}\\
        &\quad\leq C(\Delta t)^2
        \|\partial_sD_{l,m,n}\|_{L^2((m-1)\Delta t,m\Delta t]}
        \|\partial_sD_{l',m,n+p}\|_{L^2((m-1)\Delta t,m\Delta t]}
    \end{align*}
    for some constant $C>0$ independent of $m$, $n$, $p$, $l$, $l'$, and $\Delta t$.
    Combining this bound with the preceding decomposition proves part~(i) when $ll'>0$.\\
    
    For $m=n$ and $l\geq2$, we first note that
    \[
        \phi^{(*2)}(u)=\int_0^u\phi(u-v)\phi(v)dv\leq\|\phi\|_\infty^2u.
    \]
    Therefore,
    \begin{align*}
        \phi^{(*l)}(u)
        &=\int_0^u\phi^{*(l-2)}(u-v)\phi^{(*2)}(v)dv\\
        &\leq\|\phi\|_\infty^2\int_0^u\phi^{*(l-2)}(u-v)v\,dv\\
        &\leq\|\phi\|_\infty^2u\int_0^u\phi^{*(l-2)}(u-v)dv\\
        &\leq\|\phi\|_\infty^2\alpha^{l-2}u.
    \end{align*}
    Using this bound and extending $\phi^{(*l)}$ by zero on $\mathbb{R}_-$, we obtain
    \begin{align*}
        |D_{l,n,n}(s)|
        &\leq\int_{(n-1)\Delta t}^{n\Delta t}
        \dfrac{1}{\Delta t}\int_{(n-1)\Delta t}^{n\Delta t}
        |\phi^{(*l)}(t-s)-\phi^{(*l)}(t-s')|ds'dt\\
        &\leq\dfrac{\|\phi\|_\infty^2\alpha^{l-2}}{\Delta t}
        \int_{(n-1)\Delta t}^{n\Delta t}
        \int_{(n-1)\Delta t}^{n\Delta t}
        \left\{\bm{1}_{t\geq s}(t-s)+\bm{1}_{t\geq s'}(t-s')\right\}ds'dt\\
        &=\|\phi\|_\infty^2\alpha^{l-2}
        \left\{\int_{(n-1)\Delta t}^{n\Delta t}
        \bm{1}_{t\geq s}(t-s)dt+\dfrac{(\Delta t)^2}{6}\right\}.
    \end{align*}
    Moreover, for every $l'\geq1$,
    \[
        |D_{l',n,n+p}(s)|
        \leq2\Delta t\|\phi^{(*l')}\|_\infty
        \leq2\alpha^{l'-1}\|\phi\|_\infty\Delta t.
    \]
    Hence, for $l\geq2$ and $l'\geq1$,
    \begin{align*}
        &\left|\int_{(n-1)\Delta t}^{n\Delta t}
        D_{l,n,n}(s)D_{l',n,n+p}(s)ds\right|\\
        &\leq2\|\phi\|_\infty^3\alpha^{l+l'-3}\Delta t
        \int_{(n-1)\Delta t}^{n\Delta t}
        \left\{\int_{(n-1)\Delta t}^{n\Delta t}
        \bm{1}_{t\geq s}(t-s)dt+\dfrac{(\Delta t)^2}{6}\right\}ds\\
        &=\dfrac{2}{3}\|\phi\|_\infty^3
        \alpha^{l+l'-3}(\Delta t)^4.
    \end{align*}

    We next evaluate the error terms arising from Lemma~\ref{lem.phi-convolution}. Since
    $\phi_0\leq\alpha$ and $\phi_0\leq\|\phi\|_\infty\Delta t$, we have
    \[
        (\phi^{(*l)})_0=\phi_0^l
        \leq\alpha^{l-1}\|\phi\|_\infty\Delta t.
    \]
    Moreover, for $l'\geq1$,
    \begin{align*}
        (\phi^{(*l')})_p
        &=\sum_{r=0}^p\phi_r(\phi^{*(l'-1)})_{p-r}\\
        &\leq\left(\sup_{r\geq0}\phi_r\right)
        \sum_{r=0}^p(\phi^{*(l'-1)})_{p-r}\\
        &\leq\alpha^{l'-1}\|\phi\|_\infty\Delta t,
    \end{align*}
    where the last inequality follows from Lemma~\ref{lem.phiconvolution-bound}. Similarly,
    \[
        B_{l',n,n+p}
        \leq\Delta t\|\phi^{(*l')}\|_\infty
        \leq\alpha^{l'-1}\|\phi\|_\infty\Delta t.
    \]
    Consequently, when $l\geq2$ and $l'\geq1$,
    \begin{align*}
        &(\phi^{(*l)})_0\bm{1}_{l'\geq2}l'\alpha^{l'-2}O((\Delta t)^3)
        +\bm{1}_{l\geq2}l\alpha^{l-2}O((\Delta t)^3)B_{l',n,n+p}\\
        &\quad=
        \left(\bm{1}_{l\geq2}l+\bm{1}_{l'\geq2}l'\right)
        \alpha^{l+l'-3}O((\Delta t)^4).
    \end{align*}
        Combining this estimate with the preceding bound for the product of the $D$ terms, and using $l\geq2$, the total error is
    \[
        \left(\bm{1}_{l\geq2}l+\bm{1}_{l'\geq2}l'\right)
        \alpha^{l+l'-3}O((\Delta t)^4).
    \]

        It remains to consider the case $l=1$ and $l'\geq2$. Since
    \[
        \|(\phi^{(*l')})'\|_\infty
        \leq\alpha^{l'-2}\|(\phi^{(*2)})'\|_\infty
        \leq\alpha^{l'-2}
        \left(|\phi(0)|\|\phi\|_\infty+\|\phi'\|_2\|\phi\|_2\right),
    \]
    the same averaging argument as above gives
    \begin{align*}
        |D_{l',n,n+p}(s)|
        &\leq\dfrac{1}{\Delta t}
        \int_{(n-1)\Delta t}^{n\Delta t}
        \int_{(n+p-1)\Delta t}^{(n+p)\Delta t}
        |\phi^{(*l')}(t'-s)-\phi^{(*l')}(t'-s')|dt'ds'\\
        &\leq\|(\phi^{(*l')})'\|_\infty(\Delta t)^2
        =\alpha^{l'-2}O((\Delta t)^2).
    \end{align*}
    Together with
    \[
        |D_{1,n,n}(s)|\leq2\|\phi\|_\infty\Delta t,
    \]
    this yields
    \[
        \left|\int_{(n-1)\Delta t}^{n\Delta t}
        D_{1,n,n}(s)D_{l',n,n+p}(s)ds\right|
        =\alpha^{l'-2}O((\Delta t)^4).
    \]
    Moreover, since $(\phi^{(*1)})_0=\phi_0\leq\|\phi\|_\infty\Delta t$, the error arising from Lemma~\ref{lem.phi-convolution} satisfies
    \[
        (\phi^{(*1)})_0l'\alpha^{l'-2}O((\Delta t)^3)
        =l'\alpha^{l'-2}O((\Delta t)^4).
    \]
    Since $l'\geq2$, the first error can be absorbed into the second one. Therefore, the total error in this case is
    \[
        l'\alpha^{l'-2}O((\Delta t)^4),
    \]
    which agrees with
    \[
        \left(\bm{1}_{l\geq2}l+\bm{1}_{l'\geq2}l'\right)
        \alpha^{l+l'-3}O((\Delta t)^4)
    \]
    for $l=1$ and $l'\geq2$.\\
    
    Finally, consider the case $(l,l')=(1,1)$. In this case, the error arising from Lemma~\ref{lem.phi-convolution} is zero, so it remains only to evaluate
    \[
        \int_{(n-1)\Delta t}^{n\Delta t}
        D_{1,n,n}(s)D_{1,n,n+p}(s)ds.
    \]
    For $p=0$, since
    \[
        |D_{1,n,n}(s)|\leq2\|\phi\|_\infty\Delta t,
    \]
    we have
    \begin{align*}
        \left|\int_{(n-1)\Delta t}^{n\Delta t}
        D_{1,n,n}(s)^2ds\right|
        &\leq\int_{(n-1)\Delta t}^{n\Delta t}
        \left(2\|\phi\|_\infty\Delta t\right)^2ds\\
        &=4\|\phi\|_\infty^2(\Delta t)^3.
    \end{align*}

    For $p>0$, the $L^1$ Poincaré inequality and the zero-mean property of
    $D_{1,n,n+p}$ give
    \begin{align*}
        &\left|\int_{(n-1)\Delta t}^{n\Delta t}
        D_{1,n,n}(s)D_{1,n,n+p}(s)ds\right|\\
        &\quad\leq2\|\phi\|_\infty\Delta t
        \int_{(n-1)\Delta t}^{n\Delta t}|D_{1,n,n+p}(s)|ds\\
        &\quad\leq2C\|\phi\|_\infty(\Delta t)^2
        \int_{(n-1)\Delta t}^{n\Delta t}
        |\partial_sD_{1,n,n+p}(s)|ds.
    \end{align*}
    Moreover,
    \begin{align*}
        |\partial_sD_{1,n,n+p}(s)|
        &=\left|\phi((n+p)\Delta t-s)
        -\phi((n+p-1)\Delta t-s)\right|\\
        &=\left|\int_{(n+p-1)\Delta t-s}^{(n+p)\Delta t-s}
        \phi'(u)du\right|\\
        &\leq\sqrt{\Delta t}\|\phi'\|_2.
    \end{align*}
    Therefore,
    \begin{align*}
        &\left|\int_{(n-1)\Delta t}^{n\Delta t}
        D_{1,n,n}(s)D_{1,n,n+p}(s)ds\right|\\
        &\quad\leq2C\|\phi\|_\infty\|\phi'\|_2(\Delta t)^{7/2}.
    \end{align*}
    This proves the assertion for $(l,l')=(1,1)$ and completes the proof.
\end{proof}

\begin{lem}\label{lem.phi-sum-evaluation}
    For $p,l,l'=0,1,2,\ldots$, the following expansion holds as $\Delta t\to0$:
    \begin{align}
        &\sum_{m=-\infty}^n\int_{(m-1)\Delta t}^{m\Delta t}
        \left(\int_{(n-1)\Delta t}^{n\Delta t}\phi^{(*l)}(t-s)dt\right)
        \left(\int_{(n+p-1)\Delta t}^{(n+p)\Delta t}\phi^{(*l')}(t'-s)dt'\right)ds \nonumber\\
        &=\Delta t\sum_{m=-\infty}^n
        (\phi^{(*l)})_{n-m}(\phi^{(*l')})_{n+p-m} \nonumber\\
        &\quad+\begin{cases}
            \left(\bm{1}_{l\geq2}l\alpha^{l-2}
            +\bm{1}_{l'\geq2}l'\alpha^{l'-2}\right)
            O((\Delta t)^3),
            & ll'=0,\\
            (l+l')\alpha^{l+l'-2}O((\Delta t)^3),
            & ll'>0\text{ and }(l,l')\neq(1,1),\\
            O((\Delta t)^{7/2}),
            & p>0\text{ and }(l,l')=(1,1),\\
            O((\Delta t)^3),
            & p=0\text{ and }(l,l')=(1,1).
        \end{cases}
    \end{align}
    The constants implicit in the $O$ terms depend only on $\phi$ and are independent of $n$, $p$, $l$, $l'$, and $\Delta t$.
\end{lem}
\begin{proof}
    We first consider the case $ll'=0$. By Lemma~\ref{lem.decompose-time}, the expansion is exact for every $m<n$. Therefore, only the term $m=n$ contributes to the error. Applying the second part of Lemma~\ref{lem.decompose-time}, we obtain
    \[
        \left(\bm{1}_{l\geq2}l\alpha^{l-2}
        +\bm{1}_{l'\geq2}l'\alpha^{l'-2}\right)
        O((\Delta t)^3).
    \]

    Next, suppose that $ll'>0$. By Lemma~\ref{lem.phiconvolution-bound}, the sum of the first two error terms in Eq.~\eqref{eq.decompose-time-error} is bounded by
    \begin{align*}
        &\left|\sum_{m=-\infty}^{n-1}\left\{
        (\phi^{(*l)})_{n-m}\bm{1}_{l'\geq2}l'\alpha^{l'-2}O((\Delta t)^3)\right.\right.\\
        &\left.\left.\quad+\bm{1}_{l\geq2}l\alpha^{l-2}O((\Delta t)^3)
        \dfrac{1}{\Delta t}\int_{(n+p-1)\Delta t}^{(n+p)\Delta t}
        \int_{(m-1)\Delta t}^{m\Delta t}\phi^{(*l')}(t'-s')ds'dt'
        \right\}\right|\\
        &\quad\leq\left(
        \bm{1}_{l'\geq2}l'\alpha^{l+l'-2}
        +\bm{1}_{l\geq2}l\alpha^{l-2}\|\phi^{(*l')}\|_1
        \right)O((\Delta t)^3)\\
        &\quad=\left(
        \bm{1}_{l\geq2}l+\bm{1}_{l'\geq2}l'
        \right)\alpha^{l+l'-2}O((\Delta t)^3).
    \end{align*}

    By the Cauchy--Schwarz inequality, the sum of the third error terms in Eq.~\eqref{eq.decompose-time-error} is bounded by
    \begin{align*}
        &\sum_{m=-\infty}^{n-1}C(\Delta t)^2
        \|\partial_sD_{l,m,n}\|_{L^2((m-1)\Delta t,m\Delta t]}
        \|\partial_sD_{l',m,n+p}\|_{L^2((m-1)\Delta t,m\Delta t]}\\
        &\quad\leq C(\Delta t)^2
        \left(\sum_{m=-\infty}^{n-1}
        \|\partial_sD_{l,m,n}\|_{L^2((m-1)\Delta t,m\Delta t]}^2\right)^{1/2}
        \left(\sum_{m=-\infty}^{n-1}
        \|\partial_sD_{l',m,n+p}\|_{L^2((m-1)\Delta t,m\Delta t]}^2\right)^{1/2}\\
        &=C(\Delta t)^2
        \left(\sum_{m=-\infty}^{n-1}
        \left\|\int_{(n-1)\Delta t}^{n\Delta t}
        (\phi^{(*l)})'(t-\cdot)dt
        \right\|_{L^2((m-1)\Delta t,m\Delta t]}^2\right)^{1/2}\\
        &\qquad\times
        \left(\sum_{m=-\infty}^{n-1}
        \left\|\int_{(n+p-1)\Delta t}^{(n+p)\Delta t}
        (\phi^{(*l')})'(t'-\cdot)dt'
        \right\|_{L^2((m-1)\Delta t,m\Delta t]}^2\right)^{1/2}\\
        &\leq C(\Delta t)^2
        \int_{(n-1)\Delta t}^{n\Delta t}\|(\phi^{(*l)})'\|_2dt
        \int_{(n+p-1)\Delta t}^{(n+p)\Delta t}\|(\phi^{(*l')})'\|_2dt'\\
        &=C(\Delta t)^4
        \|(\phi^{(*l)})'\|_2\|(\phi^{(*l')})'\|_2.
    \end{align*}

    Suppose first that $l,l'\geq2$. By Young's convolution inequality,
    \[
        \|(\phi^{(*l)})'\|_2
        \leq\alpha^{l-2}
        \left(\alpha\|\phi'\|_2+|\phi(0)|\|\phi\|_2\right),
    \]
    and the same bound holds with $l'$ in place of $l$. Hence, the sum of the third error terms for $m<n$ is
    \[
        \alpha^{l+l'-4}O((\Delta t)^4).
    \]
    By the second part of Lemma~\ref{lem.decompose-time}, the error corresponding to $m=n$ is
    \[
        \left(\bm{1}_{l\geq2}l+\bm{1}_{l'\geq2}l'\right)
        \alpha^{l+l'-3}O((\Delta t)^4).
    \]
    Combining these estimates, the total error is bounded by
    \begin{align*}
        &\left(\bm{1}_{l\geq2}l+\bm{1}_{l'\geq2}l'\right)
        \alpha^{l+l'-2}O((\Delta t)^3)
        +\alpha^{l+l'-4}O((\Delta t)^4)\\
        &\quad+\left(\bm{1}_{l\geq2}l+\bm{1}_{l'\geq2}l'\right)
        \alpha^{l+l'-3}O((\Delta t)^4)\\
        &=\left(\bm{1}_{l\geq2}l+\bm{1}_{l'\geq2}l'\right)
        \alpha^{l+l'-2}O((\Delta t)^3).
    \end{align*}

    If $l=1$ and $l'\geq2$, then
    \[
        \|(\phi^{(*l)})'\|_2=\|\phi'\|_2,
        \qquad
        \|(\phi^{(*l')})'\|_2
        \leq\alpha^{l'-2}
        \left(\alpha\|\phi'\|_2+|\phi(0)|\|\phi\|_2\right).
    \]
    Therefore, the sum of the third error terms for $m<n$ is
    \[
        \alpha^{l'-2}O((\Delta t)^4)
        =\alpha^{l+l'-3}O((\Delta t)^4).
    \]
    The same estimate holds when $l\geq2$ and $l'=1$. Adding the error corresponding to $m=n$, the total error is bounded by
    \begin{align*}
        &\left(\bm{1}_{l\geq2}l+\bm{1}_{l'\geq2}l'\right)
        \alpha^{l+l'-2}O((\Delta t)^3)
        +\alpha^{l+l'-3}O((\Delta t)^4)\\
        &\quad+\left(\bm{1}_{l\geq2}l+\bm{1}_{l'\geq2}l'\right)
        \alpha^{l+l'-3}O((\Delta t)^4)\\
        &=\left(\bm{1}_{l\geq2}l+\bm{1}_{l'\geq2}l'\right)
        \alpha^{l+l'-2}O((\Delta t)^3).
    \end{align*}
    Since $\alpha>0$ is fixed, the $O((\Delta t)^4)$ terms in both cases can be absorbed into the $O((\Delta t)^3)$ term for sufficiently small $\Delta t$.

    Finally, let $(l,l')=(1,1)$. The first two error terms in Eq.~\eqref{eq.decompose-time-error} vanish, and the preceding Cauchy--Schwarz estimate shows that the sum of the third error terms for $m<n$ is $O((\Delta t)^4)$. By the second part of Lemma~\ref{lem.decompose-time}, the error corresponding to $m=n$ is $O((\Delta t)^{7/2})$ when $p>0$ and $O((\Delta t)^3)$ when $p=0$. Consequently, the total error is $O((\Delta t)^{7/2})$ when $p>0$ and $O((\Delta t)^3)$ when $p=0$. This completes the proof.
\end{proof}

\begin{lem}\label{lem.lipschitz-space}
    For every $l=1,2,\ldots$, $f^{(*l)}(\bm{x},\bm{y})$ is $L_y$-Lipschitz continuous with respect to its second argument.
\end{lem}
\begin{proof}
    We prove the assertion by induction on $l$. For $l=1$, it follows directly from the assumption on $f$. Suppose that the assertion holds for $l=k$. Then, for every $\bm{x},\bm{y},\bm{z}\in A$,
    \begin{align*}
        &|f^{*(k+1)}(\bm{x},\bm{y})-f^{*(k+1)}(\bm{x},\bm{z})|\\
        &\quad=\left|\int_Af(\bm{x},\bm{w})
        \left(f^{(*k)}(\bm{w},\bm{y})-f^{(*k)}(\bm{w},\bm{z})\right)d\bm{w}\right|\\
        &\quad\leq\int_Af(\bm{x},\bm{w})
        |f^{(*k)}(\bm{w},\bm{y})-f^{(*k)}(\bm{w},\bm{z})|d\bm{w}\\
        &\quad\leq L_y\|\bm{y}-\bm{z}\|_2
        \int_Af(\bm{x},\bm{w})d\bm{w}\\
        &\quad=L_y\|\bm{y}-\bm{z}\|_2,
    \end{align*}
    where the last equality follows from Eq.~\eqref{eq.normalization-2}. Thus, the assertion also holds for $l=k+1$, which completes the proof.
\end{proof}

\begin{lem}\label{lem.decompose-space}
    For $i,j,k=1,\ldots,d$ and $l,l'=0,1,2,\ldots$, the following expansion holds as $|S|\to0$:
    \begin{align}
        &\int_{S_k}
        \left(\int_{S_i}f^{(*l)}(\bm{x},\bm{y})d\bm{x}\right)
        \left(\int_{S_j}f^{(*l')}(\bm{x}',\bm{y})d\bm{x}'\right)
        \bar{\lambda}(\bm{y})d\bm{y}\nonumber\\
        &=\left(\int_{S_i}\int_{S_k}
        f^{(*l)}(\bm{x},\bm{y}')p_k(\bm{y}')
        d\bm{y}'d\bm{x}\right)
        \left(\int_{S_j}\int_{S_k}
        f^{(*l')}(\bm{x}',\bm{y}')p_k(\bm{y}')
        d\bm{y}'d\bm{x}'\right)
        \int_{S_k}\bar{\lambda}(\bm{y})d\bm{y}\nonumber\\
        &\quad+\begin{cases}
            O(|S_i||S_j||S_k||S|^{1/q}),
            & ll'>0,\\
            \delta_{i,k}O(|S_j||S_k||S|^{1/q}),
            & l=0,\ l'>0,\\
            \delta_{j,k}O(|S_i||S_k||S|^{1/q}),
            & l>0,\ l'=0,\\
            0,
            & (l,l')=(0,0).
        \end{cases}
    \end{align}
    The constants implicit in the $O$ terms depend only on $f$, $\bar{\lambda}$, and the regularity constant of the spatial partition, and are independent of $i$, $j$, $k$, $l$, $l'$, and the spatial partition.
\end{lem}
\begin{proof}
    We first consider the case $l\geq1$ and $l'\geq1$.  For every $\bm{y}\in S_k$, we decompose the first spatial integral as
    \begin{align*}
        \int_{S_i}f^{(*l)}(\bm{x},\bm{y})d\bm{x}
        &=\int_{S_i}\int_{S_k}
        f^{(*l)}(\bm{x},\bm{y}')p_k(\bm{y}')
        d\bm{y}'d\bm{x}\\
        &\quad+\int_{S_i}\left(
        f^{(*l)}(\bm{x},\bm{y})
        -\int_{S_k}f^{(*l)}(\bm{x},\bm{y}')
        p_k(\bm{y}')d\bm{y}'
        \right)d\bm{x}.
    \end{align*}
    Similarly,
    \begin{align*}
        \int_{S_j}f^{(*l')}(\bm{x}',\bm{y})d\bm{x}'
        &=\int_{S_j}\int_{S_k}
        f^{(*l')}(\bm{x}',\bm{y}')p_k(\bm{y}')
        d\bm{y}'d\bm{x}'\\
        &\quad+\int_{S_j}\left(
        f^{(*l')}(\bm{x}',\bm{y})
        -\int_{S_k}f^{(*l')}(\bm{x}',\bm{y}')
        p_k(\bm{y}')d\bm{y}'
        \right)d\bm{x}'.
    \end{align*}
    Multiplying these decompositions and integrating over $S_k$, we obtain
    \begin{align*}
        &\int_{S_k}
        \left(\int_{S_i}f^{(*l)}(\bm{x},\bm{y})d\bm{x}\right)
        \left(\int_{S_j}f^{(*l')}(\bm{x}',\bm{y})d\bm{x}'\right)
        \bar{\lambda}(\bm{y})d\bm{y}\\
        &=\left(\int_{S_i}\int_{S_k}
        f^{(*l)}(\bm{x},\bm{y}')p_k(\bm{y}')
        d\bm{y}'d\bm{x}\right)
        \left(\int_{S_j}\int_{S_k}
        f^{(*l')}(\bm{x}',\bm{y}')p_k(\bm{y}')
        d\bm{y}'d\bm{x}'\right)
        \int_{S_k}\bar{\lambda}(\bm{y})d\bm{y}\\
        &\quad+(\text{error terms}).
    \end{align*}

    By Lemma~\ref{lem.lipschitz-space} and the regularity condition on the spatial partition,
    \begin{align*}
        &\left|
        \int_{S_i}\left(
        f^{(*l)}(\bm{x},\bm{y})
        -\int_{S_k}f^{(*l)}(\bm{x},\bm{y}')p_k(\bm{y}')d\bm{y}'
        \right)d\bm{x}
        \right|\\
        &\quad\leq
        \int_{S_i}\int_{S_k}
        |f^{(*l)}(\bm{x},\bm{y})
        -f^{(*l)}(\bm{x},\bm{y}')|
        p_k(\bm{y}')d\bm{y}'d\bm{x}\\
        &\quad\leq
        L_y|S_i|\int_{S_k}
        \|\bm{y}-\bm{y}'\|_2p_k(\bm{y}')d\bm{y}'\\
        &\quad\leq C_rL_y|S_i||S|^{1/q}.
    \end{align*}
    Similarly,
    \begin{align*}
        &\left|
        \int_{S_j}\left(
        f^{(*l')}(\bm{x}',\bm{y})
        -\int_{S_k}f^{(*l')}(\bm{x}',\bm{y}')p_k(\bm{y}')d\bm{y}'
        \right)d\bm{x}'
        \right|\\
        &\quad\leq C_rL_y|S_j||S|^{1/q}.
    \end{align*}

    Moreover, the normalization condition in Eq.~\eqref{eq.normalization-1} implies that
    \[
        \|f^{(*l)}\|_\infty\leq\|f\|_\infty
    \]
    for every $l\geq1$. Indeed, for $l\geq2$,
    \begin{align*}
        f^{(*l)}(\bm{x},\bm{y})
        &=\int_Af(\bm{x},\bm{w})
        f^{*(l-1)}(\bm{w},\bm{y})d\bm{w}\\
        &\leq\|f\|_\infty
        \int_Af^{*(l-1)}(\bm{w},\bm{y})d\bm{w}
        =\|f\|_\infty.
    \end{align*}
    Therefore,
    \[
        \int_{S_i}\int_{S_k}
        f^{(*l)}(\bm{x},\bm{y}')p_k(\bm{y}')
        d\bm{y}'d\bm{x}
        =O(|S_i|)
    \]
    and
    \[
        \int_{S_j}\int_{S_k}
        f^{(*l')}(\bm{x}',\bm{y}')p_k(\bm{y}')
        d\bm{y}'d\bm{x}'
        =O(|S_j|).
    \]
    Since $\bar{\lambda}$ is bounded,
    \[
        \int_{S_k}\bar{\lambda}(\bm{y})d\bm{y}=O(|S_k|).
    \]
    Since the two differences are bounded by
    $O(|S_i||S|^{1/q})$ and
    $O(|S_j||S|^{1/q})$, respectively, while the two averaged terms are bounded by $O(|S_i|)$ and $O(|S_j|)$, all the error terms obtained by multiplying the two decompositions are bounded by
    \[
        O(|S_i||S_j||S_k||S|^{1/q}).
    \] Thus, when $ll'>0$, the total error is
    \[
        O(|S_i||S_j||S_k||S|^{1/q}).
    \]

    Next, suppose that $l=0$ and $l'\geq1$. For almost every $\bm{y}\in S_k$,
    \[
        \int_{S_i}f^{(*0)}(\bm{x},\bm{y})d\bm{x}
        =\int_{S_i}\delta(\bm{x}-\bm{y})d\bm{x}
        =\delta_{i,k},
    \]
    and
    \[
        \int_{S_i}\int_{S_k}
        f^{(*0)}(\bm{x},\bm{y}')p_k(\bm{y}')
        d\bm{y}'d\bm{x}
        =\delta_{i,k}.
    \]
    Therefore, the difference between the left-hand side and the main term is
    \begin{align*}
        &\delta_{i,k}\int_{S_k}\left\{
        \int_{S_j}f^{(*l')}(\bm{x}',\bm{y})d\bm{x}'
        -\int_{S_j}\int_{S_k}
        f^{(*l')}(\bm{x}',\bm{y}')p_k(\bm{y}')
        d\bm{y}'d\bm{x}'
        \right\}\bar{\lambda}(\bm{y})d\bm{y}.
    \end{align*}
    By Lemma~\ref{lem.lipschitz-space}, its absolute value is bounded by
    \[
        \delta_{i,k}O(|S_j||S_k||S|^{1/q}).
    \]

    The case $l\geq1$ and $l'=0$ follows symmetrically, and the corresponding error is
    \[
        \delta_{j,k}O(|S_i||S_k||S|^{1/q}).
    \]

    Finally, suppose that $(l,l')=(0,0)$. Then,
    \begin{align*}
        &\int_{S_k}
        \left(\int_{S_i}f^{(*0)}(\bm{x},\bm{y})d\bm{x}\right)
        \left(\int_{S_j}f^{(*0)}(\bm{x}',\bm{y})d\bm{x}'\right)
        \bar{\lambda}(\bm{y})d\bm{y}\\
        &=\delta_{i,k}\delta_{j,k}
        \int_{S_k}\bar{\lambda}(\bm{y})d\bm{y}\\
        &=\left(\int_{S_i}\int_{S_k}
        f^{(*0)}(\bm{x},\bm{y}')p_k(\bm{y}')
        d\bm{y}'d\bm{x}\right)
        \left(\int_{S_j}\int_{S_k}
        f^{(*0)}(\bm{x}',\bm{y}')p_k(\bm{y}')
        d\bm{y}'d\bm{x}'\right)
        \int_{S_k}\bar{\lambda}(\bm{y})d\bm{y}.
    \end{align*}
    Thus, the equality is exact when $(l,l')=(0,0)$, which completes the proof.
\end{proof}

\begin{lem}\label{lem.Fnij}
    For every $i,j=1,\ldots,d$, the following equality holds for $n=0,1$:
    \begin{equation}
        [F^n]_{ij}
        =\int_{S_i}\int_{S_j}
        f^{(*n)}(\bm{x},\bm{y})p_j(\bm{y})
        d\bm{y}d\bm{x}.
    \end{equation}
    For every $n=2,3,\ldots$,
    \begin{equation}
        \left|
        [F^n]_{ij}
        -\int_{S_i}\int_{S_j}
        f^{(*n)}(\bm{x},\bm{y})p_j(\bm{y})
        d\bm{y}d\bm{x}
        \right|
        \leq (n-1)C_rL_y|S_i||S|^{1/q}.
    \end{equation}
\end{lem}

\begin{proof}
    For $n=0$, since $F^0=I$ and $f^{(*0)}(\bm{x},\bm{y})=\delta(\bm{x}-\bm{y})$, we have
    \begin{align*}
        [F^0]_{ij}
        &=\delta_{i,j}=\int_{S_i}\int_{S_j}
        \delta(\bm{x}-\bm{y})p_j(\bm{y})
        d\bm{y}d\bm{x}.
    \end{align*}
    For $n=1$, the equality follows directly from the definition of $F$.

    For $n=2$, we have
    \begin{align*}
        [F^2]_{ij}
        &=\sum_{k=1}^d
        \int_{S_i}\int_{S_k}
        f(\bm{x},\bm{y})p_k(\bm{y})
        d\bm{y}d\bm{x}
        \int_{S_k}\int_{S_j}
        f(\bm{y}',\bm{z})p_j(\bm{z})
        d\bm{z}d\bm{y}'\\
        &=\int_{S_i}\int_{S_j}
        f^{(*2)}(\bm{x},\bm{z})p_j(\bm{z})
        d\bm{z}d\bm{x}\\
        &\quad-\sum_{k=1}^d
        \int_{S_i}\int_{S_k}
        \left(
        f(\bm{x},\bm{y}')
        -\mathbb{E}_{P_k}[f(\bm{x},\bm{Y})]
        \right)
        \mathbb{E}_{P_j}[f(\bm{y}',\bm{Z})]
        d\bm{y}'d\bm{x}.
    \end{align*}
    More generally, an induction on $n$ gives
    \begin{align*}
        [F^n]_{ij}
        &=\int_{S_i}\int_{S_j}
        f^{(*n)}(\bm{x},\bm{z})p_j(\bm{z})
        d\bm{z}d\bm{x}\\
        &\quad-\sum_{m=1}^{n-1}\sum_{k=1}^d\sum_{l=1}^d
        [F^{m-1}]_{kj}
        \int_{S_i}\int_{S_l}
        \left(
        f^{*(n-m)}(\bm{x},\bm{y}')
        -\mathbb{E}_{P_l}
        [f^{*(n-m)}(\bm{x},\bm{Y})]
        \right)\\
        &\hspace{45mm}\times
        \mathbb{E}_{P_k}[f(\bm{y}',\bm{Z})]
        d\bm{y}'d\bm{x}.
    \end{align*}
    Indeed, the case $n=2$ was shown above, and the identity for $n+1$ follows by multiplying the identity for $n$ by $F$ and applying the same decomposition once more.

    By Lemma~\ref{lem.lipschitz-space}, for every $\bm{y}'\in S_l$,
    \begin{align*}
        &\left|
        f^{*(n-m)}(\bm{x},\bm{y}')
        -\mathbb{E}_{P_l}
        [f^{*(n-m)}(\bm{x},\bm{Y})]
        \right|\\
        &\quad\leq
        \mathbb{E}_{P_l}
        \left[
        \left|
        f^{*(n-m)}(\bm{x},\bm{y}')
        -f^{*(n-m)}(\bm{x},\bm{Y})
        \right|
        \right]\\
        &\quad\leq
        L_y\mathbb{E}_{P_l}
        [\|\bm{y}'-\bm{Y}\|_2]
        \leq C_rL_y|S|^{1/q}.
    \end{align*}
    Therefore, the absolute value of the error term is bounded by
    \begin{align*}
        &C_rL_y|S_i||S|^{1/q}
        \sum_{m=1}^{n-1}\sum_{k=1}^d\sum_{l=1}^d
        [F^{m-1}]_{kj}
        \int_{S_l}
        \mathbb{E}_{P_k}[f(\bm{y}',\bm{Z})]
        d\bm{y}'\\
        &=C_rL_y|S_i||S|^{1/q}
        \sum_{m=1}^{n-1}\sum_{k=1}^d
        [F^{m-1}]_{kj}
        \mathbb{E}_{P_k}
        \left[
        \int_A f(\bm{y}',\bm{Z})d\bm{y}'
        \right]\\
        &=C_rL_y|S_i||S|^{1/q}
        \sum_{m=1}^{n-1}\sum_{k=1}^d
        [F^{m-1}]_{kj}\\
        &=(n-1)C_rL_y|S_i||S|^{1/q}.
    \end{align*}
    Here, the third equality follows from Eq.~\eqref{eq.normalization-1}, and the last equality follows by repeatedly applying Eq.~\eqref{eq.spatial-cg-normalization}, which gives
    \[
        \sum_{k=1}^d[F^{m-1}]_{kj}=1.
    \]
    This completes the proof.
\end{proof}

\begin{lem}\label{lem.Psi-evaluation}
    Let
    \[
        \Psi_m:=\sum_{l=0}^{\infty}(\phi^{(*l)})_mF^l.
    \]
    Then, uniformly over $m=0,1,2,\ldots$ and $i,k=1,\ldots,d$,
    \begin{equation}
        [\Psi_m]_{ik}
        =\delta_{m,0}\delta_{i,k}
        +O(\Delta t|S_i|)
    \end{equation}
    as $\Delta t\to0$. The constant implicit in the $O$ term is independent of $m$, $i$, $k$, $\Delta t$, and the spatial partition.
\end{lem}
\begin{proof}
    Since $\phi$ is non-negative and bounded, Eqs.~\eqref{eq.temporal-cgkernel-1} and \eqref{eq.temporal-cgkernel-0} imply that
    \[
        0\leq\phi_m\leq\|\phi\|_\infty\Delta t
    \]
    for every $m=0,1,2,\ldots$. Therefore, for every $l\geq1$,
    \begin{align*}
        (\phi^{(*l)})_m
        &=\sum_{r=0}^m\phi_r(\phi^{*(l-1)})_{m-r}\\
        &\leq\|\phi\|_\infty\Delta t
        \sum_{r=0}^m(\phi^{*(l-1)})_{m-r}\\
        &\leq\|\phi\|_\infty\Delta t\,\alpha^{l-1},
    \end{align*}
    where the last inequality follows from Lemma~\ref{lem.phiconvolution-bound}.

    Moreover, since $f_{ij}\leq\|f\|_\infty|S_i|$ and the column sums of $F^{l-1}$ are equal to one, we have
    \begin{align*}
        [F^l]_{ik}
        &=\sum_{j=1}^df_{ij}[F^{l-1}]_{jk}\\
        &\leq\|f\|_\infty|S_i|
        \sum_{j=1}^d[F^{l-1}]_{jk}\\
        &=\|f\|_\infty|S_i|
    \end{align*}
    for every $l\geq1$.

    It follows that
    \begin{align*}
        0
        \leq[\Psi_m]_{ik}-\delta_{m,0}\delta_{i,k}&=\sum_{l=1}^{\infty}
        (\phi^{(*l)})_m[F^l]_{ik}\\
        &\leq
        \|\phi\|_\infty\|f\|_\infty
        \Delta t|S_i|
        \sum_{l=1}^{\infty}\alpha^{l-1}\\
        &=\frac{\|\phi\|_\infty\|f\|_\infty}
        {1-\alpha}\Delta t|S_i|.
    \end{align*}
    Thus,
    \[
        [\Psi_m]_{ik}
        =\delta_{m,0}\delta_{i,k}
        +O(\Delta t|S_i|)
    \]
    uniformly over $m$, $i$, and $k$, which completes the proof.
\end{proof}

\begin{lem}\label{lem.Psi-sum-evaluation}
    Let
    \[
        \Psi_m:=\sum_{l=0}^{\infty}(\phi^{(*l)})_mF^l.
    \]
    Then, for every $i,j=1,\ldots,d$ and $p=0,1,2,\ldots$,
    \begin{align}
        \sum_{m=0}^{\infty}\sum_{k=1}^d
        [\Psi_m]_{ik}|S_k|
        &=O(|S_i|),\\
        \sum_{m=0}^{\infty}\sum_{k=1}^d
        [\Psi_m]_{ik}[\Psi_{m+p}]_{jk}|S_k|
        &=\delta_{p,0}\delta_{i,j}|S_i|
        +O(\Delta t|S_i||S_j|).
    \end{align}
    The constants implicit in the $O$ terms are independent of $i$, $j$, $p$, $\Delta t$, and the spatial partition.
\end{lem}
\begin{proof}
    By Lemmas~\ref{lem.weighttedsum-Fn} and \ref{lem.phiconvolution-bound},
    \begin{align*}
        \sum_{m=0}^{\infty}\sum_{k=1}^d
        [\Psi_m]_{ik}|S_k|
        &=\sum_{m=0}^{\infty}\sum_{l=0}^{\infty}
        (\phi^{(*l)})_m
        \sum_{k=1}^d[F^l]_{ik}|S_k|\\
        &\leq\|f\|_\infty|A||S_i|
        \sum_{l=0}^{\infty}
        \sum_{m=0}^{\infty}(\phi^{(*l)})_m\\
        &=\|f\|_\infty|A||S_i|
        \sum_{l=0}^{\infty}\alpha^l\\
        &=O(|S_i|).
    \end{align*}

    We next prove the second evaluation. Suppose first that $p>0$. Since $m+p\geq1$, Lemma~\ref{lem.Psi-evaluation} gives
    \[
        [\Psi_{m+p}]_{jk}
        =O(\Delta t|S_j|)
    \]
    uniformly over $m$ and $k$. Therefore,
    \begin{align*}
        \sum_{m=0}^{\infty}\sum_{k=1}^d
        [\Psi_m]_{ik}[\Psi_{m+p}]_{jk}|S_k|
        &\leq O(\Delta t|S_j|)
        \sum_{m=0}^{\infty}\sum_{k=1}^d
        [\Psi_m]_{ik}|S_k|\\
        &=O(\Delta t|S_i||S_j|).
    \end{align*}

    Suppose next that $p=0$. We decompose
    \begin{align*}
        &\sum_{m=0}^{\infty}\sum_{k=1}^d
        [\Psi_m]_{ik}[\Psi_m]_{jk}|S_k|\\
        &\quad=[\Psi_0]_{ij}|S_j|
        +\sum_{m=0}^{\infty}\sum_{k=1}^d
        [\Psi_m]_{ik}
        \left([\Psi_m]_{jk}-\delta_{m,0}\delta_{j,k}\right)|S_k|.
    \end{align*}
    By Lemma~\ref{lem.Psi-evaluation},
    \begin{align*}
        [\Psi_0]_{ij}|S_j|
        &=\left(\delta_{i,j}+O(\Delta t|S_i|)\right)|S_j|\\
        &=\delta_{i,j}|S_i|
        +O(\Delta t|S_i||S_j|).
    \end{align*}
    Moreover, Lemma~\ref{lem.Psi-evaluation} and the first evaluation give
    \begin{align*}
        &\left|
        \sum_{m=0}^{\infty}\sum_{k=1}^d
        [\Psi_m]_{ik}
        \left([\Psi_m]_{jk}-\delta_{m,0}\delta_{j,k}\right)|S_k|
        \right|\\
        &\quad\leq O(\Delta t|S_j|)
        \sum_{m=0}^{\infty}\sum_{k=1}^d
        [\Psi_m]_{ik}|S_k|=O(\Delta t|S_i||S_j|).
    \end{align*}
    Combining the cases $p=0$ and $p>0$ yields
    \[
        \sum_{m=0}^{\infty}\sum_{k=1}^d
        [\Psi_m]_{ik}[\Psi_{m+p}]_{jk}|S_k|
        =\delta_{p,0}\delta_{i,j}|S_i|
        +O(\Delta t|S_i||S_j|),
    \]
    which completes the proof.
\end{proof}

We are now ready to prove Theorem~\ref{thm.cov} by combining the preceding estimates with the autocovariance representation in Lemma~\ref{lem.autocovariance-binnedhawkes}.

\begin{proof}[Proof of Theorem~\ref{thm.cov}]
    Fix $p\geq0$. By Lemma~\ref{lem.autocovariance-binnedhawkes} (iii) and the causality of $\phi$, the autocovariance of the binned Hawkes process can be written as
    \begin{align*}
        &\mathrm{Cov}(X_{i,n}^{(\mathrm{hw})},X_{j,n+p}^{(\mathrm{hw})})\\
        &=\sum_{l,l'=0}^{\infty}
        \left\{
        \sum_{m=-\infty}^{n}
        \int_{(m-1)\Delta t}^{m\Delta t}
        \left(\int_{(n-1)\Delta t}^{n\Delta t}
        \phi^{(*l)}(t-s)dt\right)
        \left(\int_{(n+p-1)\Delta t}^{(n+p)\Delta t}
        \phi^{(*l')}(t'-s)dt'\right)ds
        \right\}\\
        &\quad\times
        \left\{
        \sum_{k=1}^d\int_{S_k}
        \left(\int_{S_i}f^{(*l)}(\bm{x},\bm{y})d\bm{x}\right)
        \left(\int_{S_j}f^{(*l')}(\bm{x}',\bm{y})d\bm{x}'\right)
        \bar{\lambda}(\bm{y})d\bm{y}
        \right\}.
    \end{align*}

    First, Lemma~\ref{lem.phi-sum-evaluation} gives
    \begin{align*}
        &\sum_{m=-\infty}^{n}
        \int_{(m-1)\Delta t}^{m\Delta t}
        \left(\int_{(n-1)\Delta t}^{n\Delta t}
        \phi^{(*l)}(t-s)dt\right)
        \left(\int_{(n+p-1)\Delta t}^{(n+p)\Delta t}
        \phi^{(*l')}(t'-s)dt'\right)ds\\
        &\quad=\Delta t\sum_{m=-\infty}^{n}
        (\phi^{(*l)})_{n-m}
        (\phi^{(*l')})_{n+p-m}
        +(\text{temporal error}),
    \end{align*}
    where the temporal error is given by the four cases in Lemma~\ref{lem.phi-sum-evaluation}.

    Second, Lemmas~\ref{lem.decompose-space} and \ref{lem.Fnij}, together with $\int_{S_k}\bar{\lambda}(\bm{y})d\bm{y}=\mathbb{E}[X_{k,n}^{(\mathrm{hw})}]/\Delta t$, imply that
    \begin{align*}
        &\sum_{k=1}^d\int_{S_k}
        \left(\int_{S_i}f^{(*l)}(\bm{x},\bm{y})d\bm{x}\right)
        \left(\int_{S_j}f^{(*l')}(\bm{x}',\bm{y})d\bm{x}'\right)
        \bar{\lambda}(\bm{y})d\bm{y}\\
        &\quad=\sum_{k=1}^d
        [F^l]_{ik}[F^{l'}]_{jk}
        \frac{\mathbb{E}[X_{k,n}^{(\mathrm{hw})}]}{\Delta t}+\bm{1}_{(l,l')\neq(0,0)}
        (l+1)(l'+1)
        O(|S_i||S_j||S|^{1/q}).
    \end{align*}
    Indeed, the errors arising from Lemma~\ref{lem.Fnij} are bounded by
    \[
        \bm{1}_{l\geq2}(l-1)O(|S_i||S|^{1/q})
        \quad\text{and}\quad
        \bm{1}_{l'\geq2}(l'-1)O(|S_j||S|^{1/q}),
    \]
    respectively. After multiplication and summation over $k$, Lemma~\ref{lem.weighttedsum-Fn} shows that these errors, together with the additive error in Lemma~\ref{lem.decompose-space}, are bounded by
    \[
        \bm{1}_{(l,l')\neq(0,0)}
        (l+1)(l'+1)
        O(|S_i||S_j||S|^{1/q}).
    \]
    When $(l,l')=(0,0)$, both Lemmas~\ref{lem.decompose-space} and \ref{lem.Fnij} give exact equalities.

    For $(l,l')\neq(0,0)$, the bound established in the proof of Lemma~\ref{lem.Psi-evaluation} and Lemma~\ref{lem.phiconvolution-bound} give
    \begin{align*}
        \sum_{m=-\infty}^{n}
        (\phi^{(*l)})_{n-m}
        (\phi^{(*l')})_{n+p-m}
        \leq\|\phi\|_\infty\Delta t\,\alpha^{l+l'-1},
    \end{align*}
    where the pointwise bound is applied to whichever of $l$ and $l'$ is positive. Since
    \[
        \sum_{\substack{l,l'\geq0\\(l,l')\neq(0,0)}}
        (l+1)(l'+1)\alpha^{l+l'-1}<\infty,
    \]
    the terms obtained by multiplying the leading temporal term by at least one spatial approximation error contribute $O((\Delta t)^2|S_i||S_j||S|^{1/q})$.

    We next evaluate the temporal errors. Since $\mathbb{E}[X_{k,n}^{(\mathrm{hw})}]=O(\Delta t|S_k|)$, Lemma~\ref{lem.weighttedsum-Fn} gives
    \[
        \sum_{k=1}^d
        [F^l]_{ik}[F^{l'}]_{jk}
        \frac{\mathbb{E}[X_{k,n}^{(\mathrm{hw})}]}{\Delta t}
        =O(|S_i||S_j|)
    \]
    whenever $(l,l')\neq(0,0)$. For example, if $l,l'\geq1$,
    \begin{align*}
        \sum_{k=1}^d
        [F^l]_{ik}[F^{l'}]_{jk}|S_k|
        &\leq
        \sup_k[F^l]_{ik}
        \sum_{k=1}^d[F^{l'}]_{jk}|S_k|\\
        &=O(|S_i||S_j|).
    \end{align*}
    The cases $l=0$, $l'>0$ and $l>0$, $l'=0$ follow directly from $F^0=I$.

    Moreover,
    \[
        \sum_{l=2}^{\infty}l\alpha^{l-2}<\infty
    \]
    and
    \[
        \sum_{\substack{l,l'\geq1\\(l,l')\neq(1,1)}}
        (l+l')\alpha^{l+l'-2}<\infty.
    \]
    Hence, All terms containing a temporal approximation error in Lemma~\ref{lem.phi-sum-evaluation}, except for the case $(l,l')=(1,1)$, including its cross terms with the spatial approximation errors, contribute
    \[
    O((\Delta t)^3|S_i||S_j|).
    \]
    For $(l,l')=(1,1)$, the error is $O((\Delta t)^3)$ when $p=0$ and $O((\Delta t)^{7/2})$ when $p>0$. The latter is also absorbed into $O((\Delta t)^3)$ for sufficiently small $\Delta t$. The products of the temporal and spatial errors are of smaller order and are absorbed into the same bounds.

    Combining the temporal and spatial estimates yields
    \begin{align*}
        &\mathrm{Cov}(X_{i,n}^{(\mathrm{hw})},X_{j,n+p}^{(\mathrm{hw})})\\
        &\quad=
        \sum_{m=-\infty}^{n}\sum_{k=1}^d
        [\Psi_{n-m}]_{ik}
        \mathbb{E}[X_{k,n}^{(\mathrm{hw})}]
        [\Psi_{n+p-m}]_{jk}\\
        &\qquad+
        O((\Delta t)^3|S_i||S_j|)
        +O((\Delta t)^2|S_i||S_j||S|^{1/q}).
    \end{align*}

    By Theorem~\ref{thm.expectation},
    \[
        \mathbb{E}[X_{k,n}^{(\mathrm{hw})}]
        -\mathbb{E}[X_{k,n}^{(\mathrm{cg})}]
        =O(\Delta t|S_k||S|^{1/q}).
    \]
    Therefore, after setting $r=n-m$ and applying Lemma~\ref{lem.Psi-sum-evaluation},
    \begin{align*}
        &\sum_{m=-\infty}^{n}\sum_{k=1}^d
        [\Psi_{n-m}]_{ik}
        \left(
        \mathbb{E}[X_{k,n}^{(\mathrm{hw})}]
        -\mathbb{E}[X_{k,n}^{(\mathrm{cg})}]
        \right)
        [\Psi_{n+p-m}]_{jk}\\
        &\quad=
        \delta_{p,0}\delta_{i,j}
        O(\Delta t|S_i||S|^{1/q})
        +O((\Delta t)^2|S_i||S_j||S|^{1/q}).
    \end{align*}
    Consequently,
    \begin{align*}
        &\mathrm{Cov}(X_{i,n}^{(\mathrm{hw})},X_{j,n+p}^{(\mathrm{hw})})\\
        &\quad=
        \sum_{m=-\infty}^{n}\sum_{k=1}^d
        [\Psi_{n-m}]_{ik}
        \mathbb{E}[X_{k,n}^{(\mathrm{cg})}]
        [\Psi_{n+p-m}]_{jk}\\
        &\qquad+
        O((\Delta t)^3|S_i||S_j|)
        +O((\Delta t)^2|S_i||S_j||S|^{1/q})\\
        &\qquad+
        \delta_{p,0}\delta_{i,j}
        O(\Delta t|S_i||S|^{1/q}).
    \end{align*}
    By Corollary~\ref{cor:spatiotemporal-autocovariance}, the first term on the right-hand side equals $\mathrm{Cov}(X_{i,n}^{(\mathrm{cg})},X_{j,n+p}^{(\mathrm{cg})})$.
    Thus,
    \begin{align*}
        &\mathrm{Cov}(X_{i,n}^{(\mathrm{hw})},X_{j,n+p}^{(\mathrm{hw})})\\
        &\quad=
        \mathrm{Cov}(X_{i,n}^{(\mathrm{cg})},X_{j,n+p}^{(\mathrm{cg})})\\
        &\qquad+
        O((\Delta t)^3|S_i||S_j|)
        +O((\Delta t)^2|S_i||S_j||S|^{1/q})\\
        &\qquad+
        \delta_{p,0}\delta_{i,j}
        O(\Delta t|S_i||S|^{1/q}).
    \end{align*}

    If $p=0$ and $i=j$, the second error term is absorbed into
    \[
        O(\Delta t|S_i||S|^{1/q}),
    \]
    and hence
    \begin{align*}
        &\mathrm{Cov}(X_{i,n}^{(\mathrm{cg})},X_{i,n}^{(\mathrm{cg})})
        =\mathrm{Cov}(X_{i,n}^{(\mathrm{hw})},X_{i,n}^{(\mathrm{hw})})+
        O((\Delta t)^3|S_i|^2)
        +O(\Delta t|S_i||S|^{1/q}).
    \end{align*}
    In all other cases, $\delta_{p,0}\delta_{i,j}=0$, and therefore
    \begin{align*}
        &\mathrm{Cov}(X_{i,n}^{(\mathrm{cg})},X_{j,n+p}^{(\mathrm{cg})})
        =\mathrm{Cov}(X_{i,n}^{(\mathrm{hw})},X_{j,n+p}^{(\mathrm{hw})})+
        O((\Delta t)^3|S_i||S_j|)
        +O((\Delta t)^2|S_i||S_j||S|^{1/q}).
    \end{align*}
    This completes the proof.
\end{proof}

\section{Properties of binned Hawkes processes}
In this appendix, we derive the autocovariance of spatio-temporal binned Hawkes processes from the martingale representation, following the approach of \cite{bacry2012non}. Let $\bar{\lambda}(\bm{x})$ denote the stationary intensity and write
\[
    dM(t,\bm{x})=dN(t,\bm{x})-\lambda(t,\bm{x})dtd\bm{x}.
\]
Define
\[
    \Omega(t,\bm{x},\bm{y})
    =\sum_{l=0}^{\infty}\phi^{(*l)}(t)f^{(*l)}(\bm{x},\bm{y}),
    \qquad
    \delta(t,\bm{x},\bm{y})=\delta(t)\delta(\bm{x}-\bm{y}),
\]
where the term $l=0$ in $\Omega$ is $\delta$. Identities involving Dirac delta terms below are understood in the sense of distributions.

\begin{lem}\label{lem.autocovariance-binnedhawkes}
For the stationary Hawkes process, the following identities hold.
\begin{enumerate}[(i)]
    \item The intensity admits the representation
    \[
        \lambda(t,\bm{x})=\bar{\lambda}(\bm{x})+((\Omega-\delta)*dM)(t,\bm{x}).
    \]

    \item The second moment of the counting measure satisfies
    \begin{align}
        &\mathbb{E}[dN(t,\bm{x})dN(t',\bm{x}')]\nonumber\\
        &=\Bigg\{\bar{\lambda}(\bm{x})\bar{\lambda}(\bm{x}')
        +\int_{\mathbb{R}}\int_A
        \Omega(t-s,\bm{x},\bm{y})\Omega(t'-s,\bm{x}',\bm{y})
        \bar{\lambda}(\bm{y})d\bm{y}ds\Bigg\}
        dtd\bm{x}dt'd\bm{x}'.
    \end{align}

    \item For $\Delta t>0$ and $S_i,S_j\in\mathcal{B}(A)$, let
    \[
        N_{t,\Delta t}(S_i)=N(t,S_i)-N(t-\Delta t,S_i).
    \]
    Then
    \begin{align}
        &\operatorname{Cov}(N_{t,\Delta t}(S_i),N_{u,\Delta t}(S_j))\nonumber\\
        &=\int_{\mathbb{R}}\int_A\bar{\lambda}(\bm{y})
        \left[\int_{t-\Delta t}^{t}\int_{S_i}
        \Omega(\tau-s,\bm{x},\bm{y})d\bm{x}d\tau\right]\nonumber\\
        &\qquad\times\left[\int_{u-\Delta t}^{u}\int_{S_j}
        \Omega(\tau'-s,\bm{x}',\bm{y})d\bm{x}'d\tau'\right]d\bm{y}ds.
    \end{align}
\end{enumerate}
\end{lem}
\begin{proof}
For (i), the proof of \cite[Proposition~1]{bacry2012non} applies with matrix multiplication replaced by integration over the spatial variable. The normalization of $f$ and $\alpha<1$ ensure the convergence of the corresponding convolution series.\\

For (ii), following \cite[Proposition~2]{bacry2012non}, the martingale decomposition gives
\begin{align*}
    &\mathbb{E}[dN(t,\bm{x})dN(t',\bm{x}')]\\
    &=\mathbb{E}[dM(t,\bm{x})dM(t',\bm{x}')]\\
    &\quad+\mathbb{E}[\lambda(t,\bm{x})dM(t',\bm{x}')]dtd\bm{x}
    +\mathbb{E}[dM(t,\bm{x})\lambda(t',\bm{x}')]dt'd\bm{x}'\\
    &\quad+\mathbb{E}[\lambda(t,\bm{x})\lambda(t',\bm{x}')]dtd\bm{x}dt'd\bm{x}'.
\end{align*}
The first term is given by
\[
    \mathbb{E}[dM(t,\bm{x})dM(t',\bm{x}')]
    =\bar{\lambda}(\bm{x})\delta(t-t',\bm{x},\bm{x}')dtd\bm{x}dt'd\bm{x}'.
\]
Using (i) and the fact that the martingale measure has zero mean, the second term satisfies
\begin{align*}
    \mathbb{E}[\lambda(t,\bm{x})dM(t',\bm{x}')]
    &=\mathbb{E}[((\Omega-\delta)*dM)(t,\bm{x})dM(t',\bm{x}')]\\
    &=(\Omega-\delta)(t-t',\bm{x},\bm{x}')\bar{\lambda}(\bm{x}')dt'd\bm{x}'.
\end{align*}
Similarly,
\[
    \mathbb{E}[dM(t,\bm{x})\lambda(t',\bm{x}')]
    =\bar{\lambda}(\bm{x})(\Omega-\delta)(t'-t,\bm{x}',\bm{x})dtd\bm{x}.
\]
For the fourth term, applying (i) and the same martingale second-moment identity yields
\begin{align*}
    &\mathbb{E}[\lambda(t,\bm{x})\lambda(t',\bm{x}')]\\
    &=\bar{\lambda}(\bm{x})\bar{\lambda}(\bm{x}')
    +\int_{\mathbb{R}}\int_A
    (\Omega-\delta)(t-s,\bm{x},\bm{y})
    (\Omega-\delta)(t'-s,\bm{x}',\bm{y})
    \bar{\lambda}(\bm{y})d\bm{y}ds.
\end{align*}
Expanding the last integral, we obtain
\begin{align*}
    &\int_{\mathbb{R}}\int_A
    (\Omega-\delta)(t-s,\bm{x},\bm{y})
    (\Omega-\delta)(t'-s,\bm{x}',\bm{y})
    \bar{\lambda}(\bm{y})d\bm{y}ds\\
    &=\int_{\mathbb{R}}\int_A
    \Omega(t-s,\bm{x},\bm{y})\Omega(t'-s,\bm{x}',\bm{y})
    \bar{\lambda}(\bm{y})d\bm{y}ds\\
    &\quad-\Omega(t-t',\bm{x},\bm{x}')\bar{\lambda}(\bm{x}')
    -\bar{\lambda}(\bm{x})\Omega(t'-t,\bm{x}',\bm{x})\\
    &\quad+\bar{\lambda}(\bm{x})\delta(t-t',\bm{x},\bm{x}').
\end{align*}
Substituting these expressions into the four-term expansion, the terms outside the integral involving $\Omega$ and $\delta$ cancel, proving (ii).\\

For (iii), stationarity gives
\[
    \mathbb{E}[N_{t,\Delta t}(S_i)]
    =\Delta t\int_{S_i}\bar{\lambda}(\bm{x})d\bm{x}.
\]
By the definition of covariance,
\begin{align*}
    &\operatorname{Cov}(N_{t,\Delta t}(S_i),N_{u,\Delta t}(S_j))\\
    &=\mathbb{E}\left[
    \left(\int_{(t-\Delta t,t]\times S_i}dN(\tau,\bm{x})\right)
    \left(\int_{(u-\Delta t,u]\times S_j}dN(\tau',\bm{x}')\right)\right]\\
    &\quad-(\Delta t)^2
    \left(\int_{S_i}\bar{\lambda}(\bm{x})d\bm{x}\right)
    \left(\int_{S_j}\bar{\lambda}(\bm{x}')d\bm{x}'\right).
\end{align*}
Substituting (ii), the product of the stationary intensities cancels with the product of the means, yielding
\begin{align*}
    &\operatorname{Cov}(N_{t,\Delta t}(S_i),N_{u,\Delta t}(S_j))\\
    &=\int_{t-\Delta t}^{t}\int_{S_i}
    \int_{u-\Delta t}^{u}\int_{S_j}
    \Bigg[\int_{\mathbb{R}}\int_A
    \Omega(\tau-s,\bm{x},\bm{y})\\
    &\qquad\times\Omega(\tau'-s,\bm{x}',\bm{y})
    \bar{\lambda}(\bm{y})d\bm{y}ds\Bigg]
    d\bm{x}'d\tau'd\bm{x}d\tau.
\end{align*}
Changing the order of integration gives (iii).
\end{proof}

\section{Proofs for binned Poisson approximation}\label{appendix.binnedpoisson}
\subsection{Binned Poisson approximation}\label{subset.binnedpoisson}
In this appendix, we briefly describe the binned Poisson approximation for aggregated Hawkes processes, which provides a basis for conventional discretization-based approaches \cite{shlomovich2022parameter,sheen2022tensor,mark2018network}.

In this approximation, the conditional intensity of the spatio-temporal Hawkes process is discretized directly in both time and space. Specifically,
\begin{align}
    \lambda_{i,n} &= \mu(\bm{r}_i)\Delta t|S_i|+\sum_{k\leq n-1}\sum_l\phi((n-k)\Delta t)f(\bm{r}_i,\bm{r}_l)X_{l,k}\Delta t|S_i|, 
    \label{eq.binnedpoisson-lambda-component}
\end{align}
where $\bm{r}_i$ is a representative point of $S_i$. In vector form,
\begin{align}
    \bm{\lambda}_n &= \tilde{\bm{\mu}}\Delta t+\sum_{k\leq n-1}\phi((n-k)\Delta t)\Delta t\tilde{F}\bm{X}_{k},\label{eq.binnedpoisson-lambda}
\end{align}
where 
$[\tilde{\bm{\mu}}]_i=\mu(\bm{r}_i)|S_i|$ and  $[\tilde{F}]_{il}=f(\bm{r}_i,\bm{r}_l)|S_i| $. 
Given the history of the binned counts, each component is assumed to follow a Poisson distribution,
\begin{equation}
X_{i,n} \mid \bm{X}_1,...,\bm{X}_{n-1}\ \sim\ \mathrm{Poisson}(\lambda_{i,n}), 
\label{eq.binnedpoisson-poisson}
\end{equation}
and the components $X_{1,n},\ldots,X_{d,n}$ are conditionally independent.
Comparing Eq.~\eqref{eq.binnedpoisson-lambda} with Eq.~\eqref{eq.spatiotemporal-lambda}, the binned Poisson approximation differs from the proposed model in two important respects. First, it evaluates the spatial kernel at representative points rather than averaging it over spatial regions. Second, it excludes the $k=n$ term and therefore does not account for intra-bin excitation.

Because $\bm{\lambda}_n$ in the binned Poisson approximation is measurable with respect to the history $H_{n-1}$, the conditional Poisson distributions completely specify the probability law of the count process. Consequently, the conditional log-likelihood is
\begin{equation}
    \log L(\bm{X}_1,...,\bm{X}_n)=\sum_{m=1}^n\sum_i (X_{i,m}\log \lambda_{i,m}-\lambda_{i,m}) + \text{const}.
    \label{eq.binnedpoisson-likelihood} 
\end{equation}

The stationary expectation in the binned Poisson approximation is
\begin{equation}
    \mathbb{E}[\bm{X}_n]=\left( I-\sum_{m=1}^{\infty}\phi(m\Delta t)\Delta t\tilde{F}\right)^{-1}\tilde{\bm{\mu}}\Delta t,
    \label{eq.binnedpoisson-expectation}
\end{equation}
provided that
\begin{equation}
\rho\left(
\sum_{m=1}^{\infty}
\phi(m\Delta t)\Delta t\tilde{F}
\right)
<1,
\label{eq.binnedpoisson-stationarity}
\end{equation}
where $\rho(\cdot)$ denotes the spectral radius. This condition is the standard stability condition for the corresponding Poisson autoregressive process; see, for example, \cite{doukhan2012weak,fokianos2020multivariate}.

The stationary autocovariance is
\begin{equation}
\operatorname{Cov}(\bm{X}_n,\bm{X}_{n+s})=\Delta t \sum^{\infty}_{l=0}\tilde{\Psi}_l \mathrm{diag}\left( \left(I-\sum_{m=1}^{\infty}\phi(m\Delta t)\Delta t\tilde{F}\right)^{-1}\tilde{\bm{\mu}}\right)\tilde{\Psi}_{l+s}^T,
\label{eq.binnedpoisson-autocovariance}
\end{equation}
where
\[
    \tilde{\Psi}_k=\sum^{\infty}_{j=0}(\tilde{\phi}^{(*j)})_k\tilde{F}^j,
\]
and $(\tilde{\phi}^{(*j)})_k$ denotes the $k$-th element of the $j$-fold discrete convolution of $\{\tilde{\phi}_i\}_i$, with
\[
    \tilde{\phi}_i = \begin{cases}
        0 & i=0\\
        \phi(i\Delta t)\Delta t & i=1,2,\ldots,
    \end{cases}
\] 
and $(\tilde{\phi}^{(*0)})_k = \delta_{k,0}$.

\subsection{Proof of Eq.~(\ref{eq.binnedpoisson-expectation-asymptotic})}\label{appendix.binnedpoisson-thm1}
We first establish the counterpart of Lemma \ref{lem.Fnij} for the binned Poisson approximation.
\begin{lem}\label{lem.binnedpoisson-Fnij}
For $n=0$, we have $[\tilde{F}^0]_{ij}=\delta_{i,j}$. There exists a constant $C>0$, independent of $n$, $i$, $j$, and the spatial partition, such that, for every $n\geq1$ and sufficiently small $|S|$,
\[
    \left|[\tilde{F}^n]_{ij}-\int_{S_i}f^{(*n)}(\bm{x},\bm{r}_j)d\bm{x}\right|
    \leq Cn\left(1+C|S|^{1/q}\right)^{n-1}|S_i||S|^{1/q}.
\]
\end{lem}
\begin{proof}
The assertion for $n=0$ follows from $\tilde{F}^0=I$. For $n=1$, the Lipschitz continuity of $f$ in its first argument gives
\[
    \left|\tilde{f}_{ij}-\int_{S_i}f(\bm{x},\bm{r}_j)d\bm{x}\right|
    \leq C|S_i||S|^{1/q}.
\]
Together with the normalization of $f$, this also gives
\[
    \sum_{i=1}^d\tilde{f}_{ij}\leq1+C|S|^{1/q}.
\]
Here and below, $C$ is chosen sufficiently large and independently of the convolution order and the spatial partition.

The boundedness and normalization of $f$ give $\|f^{(*n)}\|_\infty\leq\|f\|_\infty$. Moreover, by Lemma~\ref{lem.lipschitz-space},
\[
    |f^{(*n)}(\bm{x},\bm{y})-f^{(*n)}(\bm{x},\bm{z})|
    \leq L_y\|\bm{y}-\bm{z}\|
\]
uniformly for $n\geq1$. Consequently,
\begin{align*}
    &\left|\sum_{k=1}^d\tilde{f}_{kj}\int_{S_i}f^{(*n)}(\bm{x},\bm{r}_k)d\bm{x}
    -\int_{S_i}f^{*(n+1)}(\bm{x},\bm{r}_j)d\bm{x}\right|\\
    &=\left|\sum_{k=1}^d\int_{S_i}\int_{S_k}
    \left\{f^{(*n)}(\bm{x},\bm{r}_k)f(\bm{r}_k,\bm{r}_j)
    -f^{(*n)}(\bm{x},\bm{y})f(\bm{y},\bm{r}_j)\right\}d\bm{y}d\bm{x}\right|\\
    &\leq C|S_i||S|^{1/q}.
\end{align*}

Suppose that the assertion holds for some $n\geq1$. Using the preceding estimate, we obtain
\begin{align*}
    &\left|[\tilde{F}^{n+1}]_{ij}-\int_{S_i}f^{*(n+1)}(\bm{x},\bm{r}_j)d\bm{x}\right|\\
    &\leq\sum_{k=1}^d\tilde{f}_{kj}
    \left|[\tilde{F}^n]_{ik}-\int_{S_i}f^{(*n)}(\bm{x},\bm{r}_k)d\bm{x}\right|
    +C|S_i||S|^{1/q}\\
    &\leq Cn\left(1+C|S|^{1/q}\right)^{n-1}|S_i||S|^{1/q}
    \sum_{k=1}^d\tilde{f}_{kj}+C|S_i||S|^{1/q}\\
    &\leq C(n+1)\left(1+C|S|^{1/q}\right)^n|S_i||S|^{1/q}.
\end{align*}
Thus, the assertion holds for $n+1$, completing the proof.
\end{proof}
From Eqs.~\eqref{eq.binnedpoisson-lambda} and \eqref{eq.binnedpoisson-poisson} and stationarity, we have
\begin{align*}
    \mathbb{E}[\bm{\lambda}_n]
    &=\mathbb{E}[\bm{X}_n^{(\mathrm{po})}]
    =\tilde{\bm{\mu}}\Delta t
    +\sum_{m=1}^{\infty}\phi(m\Delta t)\Delta t
    \tilde{F}\mathbb{E}[\bm{X}_n^{(\mathrm{po})}].
\end{align*}
Therefore,
\[
    \left(
    I-\sum_{m=1}^{\infty}\phi(m\Delta t)\Delta t\tilde{F}
    \right)
    \mathbb{E}[\bm{X}_n^{(\mathrm{po})}]
    =\tilde{\bm{\mu}}\Delta t.
\]
Under the stability condition in Eq.~\eqref{eq.binnedpoisson-stationarity}, we obtain the expansion
\begin{align*}
    \mathbb{E}[X_{i,n}^{(\mathrm{po})}]
    &=\Delta t\sum_{l=0}^{\infty}
    \left(
    \sum_{m=1}^{\infty}\phi(m\Delta t)\Delta t
    \right)^l
    [\tilde{F}^l\tilde{\bm{\mu}}]_i\\
    &=\Delta t\sum_{l=0}^{\infty}
    (\alpha+O(\Delta t))^l
    \left(
    \sum_{j=1}^d[\tilde{F}^l]_{ij}
    \mu(\bm{r}_j)|S_j|
    \right),
\end{align*}
where we used the assumed temporal approximation.

We next evaluate the last summation. For $l=0$,
\begin{align*}
    \sum_{j=1}^d[\tilde{F}^0]_{ij}\mu(\bm{r}_j)|S_j|
    &=\mu(\bm{r}_i)|S_i|\\
    &=\int_{S_i}\mu(\bm{x})d\bm{x}
    +\int_{S_i}\left(\mu(\bm{r}_i)-\mu(\bm{x})\right)d\bm{x}\\
    &=\int_{S_i}\mu(\bm{x})d\bm{x}
    +O(|S|^{1/q}|S_i|),
\end{align*}
where the last equality follows from the Lipschitz continuity of $\mu$.

For $l=1$,
\begin{align*}
    \sum_{j=1}^d[\tilde{F}]_{ij}\mu(\bm{r}_j)|S_j|
    &=\sum_{j=1}^df(\bm{r}_i,\bm{r}_j)
    \mu(\bm{r}_j)|S_i||S_j|\\
    &=\int_{S_i}(f*\mu)(\bm{x})d\bm{x}\\
    &\quad-\int_{S_i}\sum_{j=1}^d\int_{S_j}
    \left(
    f(\bm{x},\bm{y})\mu(\bm{y})
    -f(\bm{r}_i,\bm{r}_j)\mu(\bm{r}_j)
    \right)d\bm{y}d\bm{x}.
\end{align*}
The absolute value of the error term is bounded by
\begin{align*}
    &\int_{S_i}\sum_{j=1}^d\int_{S_j}
    |f(\bm{x},\bm{y})-f(\bm{r}_i,\bm{r}_j)|
    \mu(\bm{y})d\bm{y}d\bm{x}\\
    &\quad+\int_{S_i}\sum_{j=1}^d\int_{S_j}
    f(\bm{r}_i,\bm{r}_j)
    |\mu(\bm{y})-\mu(\bm{r}_j)|
    d\bm{y}d\bm{x}.
\end{align*}
For every $\bm{x}\in S_i$ and $\bm{y}\in S_j$, the Lipschitz continuity of $f$ gives
\begin{align*}
    |f(\bm{x},\bm{y})-f(\bm{r}_i,\bm{r}_j)|
    &\leq |f(\bm{x},\bm{y})-f(\bm{x},\bm{r}_j)|
    +|f(\bm{x},\bm{r}_j)-f(\bm{r}_i,\bm{r}_j)|\\
    &\leq L_y\|\bm{y}-\bm{r}_j\|
    +L_x\|\bm{x}-\bm{r}_i\|
    =O(|S|^{1/q}).
\end{align*}
Hence,
\begin{align*}
    &\int_{S_i}\sum_{j=1}^d\int_{S_j}
    |f(\bm{x},\bm{y})-f(\bm{r}_i,\bm{r}_j)|
    \mu(\bm{y})d\bm{y}d\bm{x}=O(|S|^{1/q}|S_i|)\|\mu\|_1.
\end{align*}
Similarly, the Lipschitz continuity of $\mu$ and Eq.~\eqref{eq.normalization-2} give
\begin{align*}
    &\int_{S_i}\sum_{j=1}^d\int_{S_j}
    f(\bm{r}_i,\bm{r}_j)
    |\mu(\bm{y})-\mu(\bm{r}_j)|
    d\bm{y}d\bm{x}\\
    &\leq O(|S|^{1/q}|S_i|)
    \sum_{j=1}^df(\bm{r}_i,\bm{r}_j)|S_j|\\
    &=O(|S|^{1/q}|S_i|)
    \left(1+O(|S|^{1/q})\right)
    =O(|S|^{1/q}|S_i|).
\end{align*}
Thus,
\[
    \sum_{j=1}^d[\tilde{F}]_{ij}
    \mu(\bm{r}_j)|S_j|
    =\int_{S_i}(f*\mu)(\bm{x})d\bm{x}
    +O(|S|^{1/q}|S_i|).
\]

For $l\geq2$, Lemma~\ref{lem.binnedpoisson-Fnij} gives
\[
    [\tilde{F}^l]_{ij}
    =\int_{S_i}f^{(*l)}(\bm{x},\bm{r}_j)d\bm{x}
    +O\left(l(1+C|S|^{1/q})^{l-1}|S|^{1/q}|S_i|\right).
\]
Multiplying by $\mu(\bm{r}_j)|S_j|$ and summing over $j$, we obtain
\begin{align*}
    \sum_{j=1}^d[\tilde{F}^l]_{ij}\mu(\bm{r}_j)|S_j|
    &=\int_{S_i}\sum_{j=1}^df^{(*l)}(\bm{x},\bm{r}_j)\mu(\bm{r}_j)|S_j|d\bm{x}\\
    &\quad+O\left(l(1+C|S|^{1/q})^{l-1}|S|^{1/q}|S_i|\right).
\end{align*}
By Lemma~\ref{lem.lipschitz-space}, the Lipschitz continuity of $\mu$, and the bound $\|f^{(*l)}\|_\infty\leq\|f\|_\infty$, replacing the spatial sum by the corresponding integral introduces an error of $O(|S|^{1/q}|S_i|)$, uniformly in $l$. This error is absorbed into the preceding error term. Together with the case $l=1$, this gives
\[
    \sum_{j=1}^d[\tilde{F}^l]_{ij}\mu(\bm{r}_j)|S_j|
    =\int_{S_i}(f^{(*l)}*\mu)(\bm{x})d\bm{x}
    +O\left(l(1+C|S|^{1/q})^{l-1}|S|^{1/q}|S_i|\right)
\]
for every $l\geq1$, with constants independent of $l$, $i$, and the spatial partition.

Fix a constant $\hat{\alpha}$ satisfying $\alpha<\hat{\alpha}<1$. For sufficiently small $\Delta t$, the mean-value theorem gives
\[
    \left|(\alpha+O(\Delta t))^l-\alpha^l\right|
    \leq Cl\hat{\alpha}^{l-1}\Delta t,
    \qquad l\geq1,
\]
where $C>0$ is independent of $l$ and $\Delta t$. Choose $|S|_0>0$ sufficiently small that
\[
    \hat{\alpha}(1+C|S|_0^{1/q})<1.
\]
This ensures that the error terms below are summable over $l$, uniformly for $|S|\leq|S|_0$.\\

Moreover, the boundedness of $\mu$ and $f$ and the normalization of $f$ imply
\[
    \int_{S_i}(f^{(*l)}*\mu)(\bm{x})d\bm{x}=O(|S_i|)
\]
uniformly in $l$. Consequently,
\begin{align*}
    \mathbb{E}[X_{i,n}^{(\mathrm{po})}]
    &=\Delta t\int_{S_i}\mu(\bm{x})d\bm{x}\\
    &\quad+\Delta t\sum_{l=1}^{\infty}
    \left(\alpha^l+l\hat{\alpha}^{l-1}O(\Delta t)\right)\\
    &\qquad\times\left(
    \int_{S_i}(f^{(*l)}*\mu)(\bm{x})d\bm{x}
    +O\left(l(1+C|S|^{1/q})^{l-1}|S|^{1/q}|S_i|\right)
    \right)\\
    &\quad+O(\Delta t|S|^{1/q}|S_i|)\\
    &=\Delta t\sum_{l=0}^{\infty}
    \alpha^l\int_{S_i}(f^{(*l)}*\mu)(\bm{x})d\bm{x}\\
    &\quad+O(\Delta t|S|^{1/q}|S_i|)
    +O((\Delta t)^2|S_i|)
    +O((\Delta t)^2|S|^{1/q}|S_i|)\\
    &=\Delta t\sum_{l=0}^{\infty}
    \alpha^l\int_{S_i}(f^{(*l)}*\mu)(\bm{x})d\bm{x}\\
    &\quad+O(\Delta t|S|^{1/q}|S_i|)
    +O((\Delta t)^2|S_i|).
\end{align*}
Here, the last equality follows by absorbing the product of the temporal and spatial approximation errors into the other error terms. Since
\[
    \mathbb{E}[X_{i,n}^{(\mathrm{hw})}]
    =\Delta t\sum_{l=0}^{\infty}
    \alpha^l\int_{S_i}(f^{(*l)}*\mu)(\bm{x})d\bm{x},
\]
we conclude that
\[
    \mathbb{E}[X_{i,n}^{(\mathrm{po})}]
    =\mathbb{E}[X_{i,n}^{(\mathrm{hw})}]
    +O(\Delta t|S|^{1/q}|S_i|)
    +O((\Delta t)^2|S_i|).
\]

\subsection{Proof of Eq.~(\ref{eq.binnedpoisson-autocovariance-asymptotic})}\label{appendix.binnedpoisson-autocovariance}
We first establish the counterpart of Lemma~\ref{lem.phi-convolution} for the binned Poisson approximation.
\begin{lem}\label{lem.binnedpoisson-phi-convolution}
    Fix any constant $\hat{\alpha}$ satisfying $\alpha<\hat{\alpha}<1$. For $l=0,1,2,\ldots$ and $i=1,2,\ldots$, the following bound holds for sufficiently small $\Delta t$:
    \begin{align}
        &\left|
        \dfrac{1}{\Delta t}
        \int_{(n-1)\Delta t}^{n\Delta t}
        \int_{(n-i-1)\Delta t}^{(n-i)\Delta t}
        \phi^{(*l)}(t-s)dsdt
        -(\tilde{\phi}^{(*l)})_i
        \right|\nonumber\\
        &\quad\leq
        \begin{cases}
            0, & l=0,\\
            C(\Delta t)^{3/2}, & l=1,\\
            Cl\hat{\alpha}^{l-2}(\Delta t)^2, & l\geq2,
        \end{cases}
    \end{align}
    where $C>0$ is independent of $l$, $i$, $n$, and $\Delta t$. Moreover, when $i=0$, the left-hand side is zero for $l=0$ and is $O((\Delta t)^l)$ for $l\geq1$.
\end{lem}
\begin{proof}
    First, let $i\geq1$. For $l=0$, both terms in the left-hand side are zero, and hence the assertion follows immediately.

    For $l=1$, the integral and $\tilde{\phi}_i$ can be written as
    \begin{align*}
        &\dfrac{1}{\Delta t}
        \int_{(n-1)\Delta t}^{n\Delta t}
        \int_{(n-i-1)\Delta t}^{(n-i)\Delta t}
        \phi(t-s)dsdt
        =\dfrac{1}{\Delta t}
        \int_{-\Delta t}^{\Delta t}
        (\Delta t-|u|)\phi(i\Delta t+u)du,\\
        &\tilde{\phi}_i
        =\dfrac{1}{\Delta t}
        \int_{-\Delta t}^{\Delta t}
        (\Delta t-|u|)\phi(i\Delta t)du.
    \end{align*}
    Therefore, the Cauchy--Schwarz inequality gives
    \begin{align*}
        &\left|
        \dfrac{1}{\Delta t}
        \int_{(n-1)\Delta t}^{n\Delta t}
        \int_{(n-i-1)\Delta t}^{(n-i)\Delta t}
        \phi(t-s)dsdt-\tilde{\phi}_i
        \right|\\
        &\quad\leq
        \dfrac{1}{\Delta t}
        \int_{-\Delta t}^{\Delta t}
        (\Delta t-|u|)
        |\phi(i\Delta t+u)-\phi(i\Delta t)|du\\
        &\quad\leq
        \dfrac{\|\phi'\|_2}{\Delta t}
        \int_{-\Delta t}^{\Delta t}
        (\Delta t-|u|)|u|^{1/2}du
        =O((\Delta t)^{3/2}).
    \end{align*}

    For $l=2$, we have
    \[
        (\tilde{\phi}^{(*2)})_i
        =\sum_{j=1}^{i-1}(\Delta t)^2
        \phi((i-j)\Delta t)\phi(j\Delta t),
    \]
    whereas
\begin{align*}
    &\dfrac{1}{\Delta t}
    \int_{(n-1)\Delta t}^{n\Delta t}
    \int_{(n-i-1)\Delta t}^{(n-i)\Delta t}
    \phi^{(*2)}(t-s)dsdt\\
    &=\dfrac{1}{\Delta t}
    \int_{-\Delta t}^{\Delta t}
    (\Delta t-|u|)
    \int_0^{i\Delta t+u}
    \phi(i\Delta t+u-v)\phi(v)dvdu\\
    &=\dfrac{1}{\Delta t}
    \int_{-\Delta t}^{\Delta t}
    (\Delta t-|u|)
    \sum_{j=1}^{i-1}
    \int_{(j-\frac12)\Delta t}^{(j+\frac12)\Delta t}
    \phi(i\Delta t+u-v)\phi(v)dvdu\\
    &\quad+\dfrac{1}{\Delta t}
    \int_{-\Delta t}^{\Delta t}
    (\Delta t-|u|)
    \int_0^{\frac12\Delta t}
    \phi(i\Delta t+u-v)\phi(v)dvdu\\
    &\quad+\dfrac{1}{\Delta t}
    \int_{-\Delta t}^{\Delta t}
    (\Delta t-|u|)
    \int_{(i-\frac12)\Delta t}^{i\Delta t+u}
    \phi(i\Delta t+u-v)\phi(v)dvdu.
\end{align*}
Here, $\phi$ is extended by zero to negative arguments, and the sum is interpreted as zero when $i=1$.
For each $j=1,\ldots,i-1$, we decompose the difference over the corresponding cell as
    \begin{align*}
        &\int_{(j-\frac12)\Delta t}^{(j+\frac12)\Delta t}
        \phi(i\Delta t+u-v)\phi(v)dv
        -\Delta t\phi((i-j)\Delta t)\phi(j\Delta t)\\
        &=
        \int_{(j-\frac12)\Delta t}^{(j+\frac12)\Delta t}
        \left\{
        \phi(i\Delta t+u-v)\phi(v)
        -\phi(u+(i-j)\Delta t)\phi(j\Delta t)
        \right\}dv\\
        &\quad+\Delta t
        \left\{
        \phi(u+(i-j)\Delta t)-\phi((i-j)\Delta t)
        \right\}\phi(j\Delta t).
    \end{align*}
    For $i\geq2$, we split the sum of the first terms into $j=1,\ldots,i-2$ and $j=i-1$. For $j\leq i-2$, both arguments of $\phi$ remain nonnegative throughout the corresponding cell. Thus, the Poincaré and Hölder inequalities give
\begin{align*}
    &\sum_{j=1}^{i-2}\left|\int_{(j-\frac12)\Delta t}^{(j+\frac12)\Delta t}
    \{\phi(i\Delta t+u-v)\phi(v)-\phi(u+(i-j)\Delta t)\phi(j\Delta t)\}dv\right|\\
    &\quad\leq C\Delta t\left(\|\phi'(i\Delta t+u-\cdot)\phi\|_1
    +\|\phi(i\Delta t+u-\cdot)\phi'\|_1\right)
    \leq 2C\Delta t\|\phi\|_2\|\phi'\|_2,
\end{align*}
where the sum is interpreted as zero when $i=2$. For $j=i-1$, the boundedness of $\phi$ gives
\begin{align*}
    &\left|\int_{(i-\frac32)\Delta t}^{(i-\frac12)\Delta t}
    \{\phi(i\Delta t+u-v)\phi(v)-\phi(u+\Delta t)\phi((i-1)\Delta t)\}dv\right|\\
    &\quad\leq 2\Delta t\|\phi\|_\infty^2.
\end{align*}
Since
\[
    \frac{1}{\Delta t}\int_{-\Delta t}^{\Delta t}(\Delta t-|u|)du=\Delta t,
\]
the total contribution of these terms after integration against the outer triangular weight is $O((\Delta t)^2)$. When $i=1$, the original sum is empty.

    For the second terms, the Cauchy--Schwarz inequality gives
    \begin{align*}
        &\sum_{j=1}^{i-1}
        |\phi(u+(i-j)\Delta t)-\phi((i-j)\Delta t)|
        \phi(j\Delta t)\\
        &\quad\leq
        \left(
        \sum_{j=1}^{i-1}
        |\phi(u+(i-j)\Delta t)-\phi((i-j)\Delta t)|^2
        \phi(j\Delta t)
        \right)^{1/2}
        \left(
        \sum_{j=1}^{i-1}\phi(j\Delta t)
        \right)^{1/2}\\
        &\quad\leq
        \sqrt{\|\phi\|_{\infty}}\left(
        \sum_{j=1}^{i-1}
        |\phi(u+(i-j)\Delta t)-\phi((i-j)\Delta t)|^2
        \right)^{1/2}
        \left(
        \sum_{j=1}^{i-1}\phi(j\Delta t)
        \right)^{1/2}.
    \end{align*}
    Since
    \[
        |\phi(u+(i-j)\Delta t)-\phi((i-j)\Delta t)|^2
        \leq |u|
        \left|
        \int_{(i-j)\Delta t}^{u+(i-j)\Delta t}
        |\phi'(v)|^2dv
        \right|
    \]
    and the intervals in the last integrals have disjoint interiors for fixed $u$, we obtain
    \[
        \sum_{j=1}^{i-1}
        |\phi(u+(i-j)\Delta t)-\phi((i-j)\Delta t)|^2
        \leq |u|\|\phi'\|_2^2.
    \]
    Since $\alpha<\hat{\alpha}$, the assumed temporal approximation implies that, for sufficiently small $\Delta t$,
\[
    \sum_{j=1}^{\infty}\tilde{\phi}_j
    =\sum^{\infty}_{j=1}\phi(j\Delta t)\Delta t=\alpha+O(\Delta t)\leq\hat{\alpha}.
\]
    and hence
\[
    \sum_{j=1}^{i-1}\phi(j\Delta t)
    \leq\frac{\hat{\alpha}}{\Delta t}.
\]
Together with the preceding estimate, this yields
    \[
        \sum_{j=1}^{i-1}
        |\phi(u+(i-j)\Delta t)-\phi((i-j)\Delta t)|
        \phi(j\Delta t)
        \leq C\sqrt{\frac{|u|}{\Delta t}}.
    \]
    Consequently,
    \[
        \int_{-\Delta t}^{\Delta t}
        (\Delta t-|u|)
        \sum_{j=1}^{i-1}
        |\phi(u+(i-j)\Delta t)-\phi((i-j)\Delta t)|
        \phi(j\Delta t)du
        =O((\Delta t)^2).
    \]
    The two remaining boundary integrals have intervals of length at most $3\Delta t/2$ and are also $O((\Delta t)^2)$. Therefore,
    \[
        \left|
        \dfrac{1}{\Delta t}
        \int_{(n-1)\Delta t}^{n\Delta t}
        \int_{(n-i-1)\Delta t}^{(n-i)\Delta t}
        \phi^{(*2)}(t-s)dsdt
        -(\tilde{\phi}^{(*2)})_i
        \right|
        \leq C(\Delta t)^2.
    \]

    We next proceed by induction. Suppose that, for some $l\geq2$,
    \[
        \left|
        \dfrac{1}{\Delta t}
        \int_{(n-1)\Delta t}^{n\Delta t}
        \int_{(n-i-1)\Delta t}^{(n-i)\Delta t}
        \phi^{(*l)}(t-s)dsdt
        -(\tilde{\phi}^{(*l)})_i
        \right|
        \leq Cl\hat{\alpha}^{l-2}(\Delta t)^2.
    \]
    Using the same decomposition as in the case $l=2$, with $\phi$ replaced by $\phi^{(*l)}$, the additional approximation error is bounded by
    \[
        C(\Delta t)^2
        \left(
        \|(\phi^{(*l)})'\|_2\| \phi\|_2 + \|\phi^{(*l)}\|_2\|\phi'\|_2+\sqrt{\|\phi\|_{\infty}}\|(\phi^{(*l)})'\|_2
        +\|\phi\|_\infty\|\phi^{(*l)}\|_\infty
        \right).
    \]
    The error propagated from the induction hypothesis is bounded by
    \[
        Cl\hat{\alpha}^{l-2}(\Delta t)^2
        \sum_{j=1}^{i-1}\tilde{\phi}_{i-j}
        \leq Cl\hat{\alpha}^{l-1}(\Delta t)^2.
    \]
    Moreover,
\[
    \|(\phi^{(*l)})'\|_2
    \leq\alpha^{l-2}\left(\alpha\|\phi'\|_2+|\phi(0)|\|\phi\|_2\right),
    \qquad
    \|\phi^{(*l)}\|_2\leq\alpha^{l-1}\|\phi\|_2,
    \qquad
    \|\phi^{(*l)}\|_\infty\leq\alpha^{l-1}\|\phi\|_\infty.
\]
    Since $\alpha<\hat{\alpha}$, these estimates imply
    \[
        \left|
        \dfrac{1}{\Delta t}
        \int_{(n-1)\Delta t}^{n\Delta t}
        \int_{(n-i-1)\Delta t}^{(n-i)\Delta t}
        \phi^{*(l+1)}(t-s)dsdt
        -(\tilde{\phi}^{*(l+1)})_i
        \right|
        \leq C(l+1)\hat{\alpha}^{l-1}(\Delta t)^2.
    \]
    Thus, the assertion holds for every $l\geq2$.

    Finally, let $i=0$. For $l=0$, both terms are equal to one. For $l\geq1$, $(\tilde{\phi}^{(*l)})_0=0$, and
    \begin{align*}
        &\dfrac{1}{\Delta t}
        \int_{(n-1)\Delta t}^{n\Delta t}
        \int_{(n-1)\Delta t}^{n\Delta t}
        \phi^{(*l)}(t-s)dsdt\\
        &\quad=
        \dfrac{1}{\Delta t}
        \int_0^{\Delta t}
        (\Delta t-u)\phi^{(*l)}(u)du\\
        &\quad\leq
        \dfrac{\|\phi\|_\infty^l}{(l-1)!\Delta t}
        \int_0^{\Delta t}(\Delta t-u)u^{l-1}du
        =\dfrac{\|\phi\|_\infty^l}{(l+1)!}(\Delta t)^l.
    \end{align*}
    Hence, the left-hand side is $O((\Delta t)^l)$ for $l\geq1$, which completes the proof.
\end{proof}

We next establish the counterpart of Lemma~\ref{lem.decompose-time} for the binned Poisson approximation. Unlike Lemma~\ref{lem.decompose-time}, the pointwise temporal discretization produces a nonzero approximation error even when $l=1$ or $l'=1$.

\begin{lem}\label{lem.binnedpoisson-decompose-time}
    Let $m,n\in\mathbb{Z}$ with $m\leq n$, and let $p,l,l'=0,1,2,\ldots$. Fix any constant $\hat{\alpha}$ satisfying $\alpha<\hat{\alpha}<1$. Then, the following expansions hold as $\Delta t\to0$.

    \begin{enumerate}[(i)]
        \item For $m<n$,
        \begin{align}
            &\int_{(m-1)\Delta t}^{m\Delta t}
            \left(\int_{(n-1)\Delta t}^{n\Delta t}\phi^{(*l)}(t-s)dt\right)
            \left(\int_{(n+p-1)\Delta t}^{(n+p)\Delta t}\phi^{(*l')}(t'-s)dt'\right)ds\nonumber\\
            &=\Delta t(\tilde{\phi}^{(*l)})_{n-m}
            (\tilde{\phi}^{(*l')})_{n+p-m}\nonumber\\
            &\quad+
            \begin{cases}
                \begin{aligned}
                    &(\tilde{\phi}^{(*l)})_{n-m}
                    \left\{
                    \bm{1}_{l'=1}O((\Delta t)^{5/2})
                    +\bm{1}_{l'\geq2}l'\hat{\alpha}^{l'-2}O((\Delta t)^3)
                    \right\}\\
                    &\quad+
                    \left\{
                    \bm{1}_{l=1}O((\Delta t)^{5/2})
                    +\bm{1}_{l\geq2}l\hat{\alpha}^{l-2}O((\Delta t)^3)
                    \right\}\\
                    &\qquad\times
                    \dfrac{1}{\Delta t}
                    \int_{(m-1)\Delta t}^{m\Delta t}
                    \int_{(n+p-1)\Delta t}^{(n+p)\Delta t}
                    \phi^{(*l')}(t'-s')dt'ds'\\
                    &\quad+
                    O((\Delta t)^2)
                    \|\partial_sD_{l,m,n}\|_{L^2((m-1)\Delta t,m\Delta t]}
                    \|\partial_sD_{l',m,n+p}\|_{L^2((m-1)\Delta t,m\Delta t]},
                \end{aligned}
                & ll'>0,\\
                0, & ll'=0,
            \end{cases}
            \label{eq.binnedpoisson-decompose-time-error}
        \end{align}
        where
        \[
            D_{l,m,n}(s)
            =\int_{(n-1)\Delta t}^{n\Delta t}
            \left(
            \phi^{(*l)}(t-s)
            -\dfrac{1}{\Delta t}
            \int_{(m-1)\Delta t}^{m\Delta t}
            \phi^{(*l)}(t-s')ds'
            \right)dt.
        \]

        \item For $m=n$,
        \begin{align}
            &\int_{(n-1)\Delta t}^{n\Delta t}
            \left(\int_{(n-1)\Delta t}^{n\Delta t}\phi^{(*l)}(t-s)dt\right)
            \left(\int_{(n+p-1)\Delta t}^{(n+p)\Delta t}\phi^{(*l')}(t'-s)dt'\right)ds\nonumber\\
            &=\Delta t(\tilde{\phi}^{(*l)})_0
            (\tilde{\phi}^{(*l')})_p\nonumber\\
            &\quad+
            \begin{cases}
                \displaystyle
                \alpha^{l'-1}
                \frac{\|\phi\|_\infty^{l+1}}{(l+1)!}
                O((\Delta t)^{l+2}),
                & ll'>0,\\[2mm]
                0,
                & p>0\text{ and }l'=0,\\
                O((\Delta t)^{5/2}),
                & p>0,\ l=0,\text{ and }l'=1,\\
                l'\hat{\alpha}^{l'-2}O((\Delta t)^3),
                & p>0,\ l=0,\text{ and }l'\geq2,\\
                \bm{1}_{l>0}O((\Delta t)^{l+1})
                +\bm{1}_{l'>0}O((\Delta t)^{l'+1}),
                & p=0\text{ and }ll'=0.
            \end{cases}
            \label{eq.binnedpoisson-decompose-time-diagonal}
        \end{align}
    \end{enumerate}
    The constants implicit in the $O$ terms depend only on $\phi$ and the fixed constant $\hat{\alpha}$ and are independent of $m$, $n$, $p$, $l$, $l'$, and $\Delta t$.
\end{lem}

\begin{proof}
    We first consider the case $m<n$. If $ll'=0$, both the integral and the leading term in Eq.~\eqref{eq.binnedpoisson-decompose-time-error} are zero, so the assertion follows immediately.

    Let $ll'>0$. As in the proof of Lemma~\ref{lem.decompose-time}, define
    \[
        B_{l,m,n}
        :=\dfrac{1}{\Delta t}
        \int_{(m-1)\Delta t}^{m\Delta t}
        \int_{(n-1)\Delta t}^{n\Delta t}
        \phi^{(*l)}(t-s')dtds'.
    \]
    Then,
    \[
        \int_{(m-1)\Delta t}^{m\Delta t}
        D_{l,m,n}(s)ds=0,
    \]
    and the same decomposition as in Lemma~\ref{lem.decompose-time} gives
    \begin{align*}
        &\int_{(m-1)\Delta t}^{m\Delta t}
        \left(\int_{(n-1)\Delta t}^{n\Delta t}
        \phi^{(*l)}(t-s)dt\right)
        \left(\int_{(n+p-1)\Delta t}^{(n+p)\Delta t}
        \phi^{(*l')}(t'-s)dt'\right)ds\\
        &=\Delta t B_{l,m,n}B_{l',m,n+p}
        +\int_{(m-1)\Delta t}^{m\Delta t}
        D_{l,m,n}(s)D_{l',m,n+p}(s)ds.
    \end{align*}
    Lemma~\ref{lem.binnedpoisson-phi-convolution} gives
    \[
        B_{l,m,n}
        =(\tilde{\phi}^{(*l)})_{n-m}
        +\begin{cases}
            O((\Delta t)^{3/2}), & l=1,\\
            l\hat{\alpha}^{l-2}O((\Delta t)^2), & l\geq2.
        \end{cases}
    \]
    We decompose the first term as
    \begin{align*}
        \Delta t B_{l,m,n}B_{l',m,n+p}
        &=\Delta t(\tilde{\phi}^{(*l)})_{n-m}
        (\tilde{\phi}^{(*l')})_{n+p-m}\\
        &\quad+\Delta t(\tilde{\phi}^{(*l)})_{n-m}
        \left\{
        B_{l',m,n+p}
        -(\tilde{\phi}^{(*l')})_{n+p-m}
        \right\}\\
        &\quad+\Delta t
        \left\{
        B_{l,m,n}
        -(\tilde{\phi}^{(*l)})_{n-m}
        \right\}
        B_{l',m,n+p}.
    \end{align*}
    Therefore, the second term on the right-hand side is
    \[
        (\tilde{\phi}^{(*l)})_{n-m}
        \left\{
        \bm{1}_{l'=1}O((\Delta t)^{5/2})
        +\bm{1}_{l'\geq2}l'\hat{\alpha}^{l'-2}
        O((\Delta t)^3)
        \right\},
    \]
    while the third term is
    \[
        \left\{
        \bm{1}_{l=1}O((\Delta t)^{5/2})
        +\bm{1}_{l\geq2}l\hat{\alpha}^{l-2}
        O((\Delta t)^3)
        \right\}B_{l',m,n+p}.
    \]
    Moreover, by the same Poincaré and Cauchy--Schwarz estimates as in the proof of Lemma~\ref{lem.decompose-time},
    \begin{align*}
        &\left|
        \int_{(m-1)\Delta t}^{m\Delta t}
        D_{l,m,n}(s)D_{l',m,n+p}(s)ds
        \right|\\
        &\quad\leq
        O((\Delta t)^2)
        \|\partial_sD_{l,m,n}\|_{L^2((m-1)\Delta t,m\Delta t]}
        \|\partial_sD_{l',m,n+p}\|_{L^2((m-1)\Delta t,m\Delta t]}.
    \end{align*}
    Substituting the definition of $B_{l',m,n+p}$ proves (i).

    We next consider the case $m=n$. Suppose first that $ll'>0$. Since $(\tilde{\phi}^{(*l)})_0=0$, the leading term in Eq.~\eqref{eq.binnedpoisson-decompose-time-diagonal} is zero. The integral is bounded by
    \begin{align*}
        &\Delta t\|\phi^{(*l')}\|_\infty
        \int_{(n-1)\Delta t}^{n\Delta t}
        \int_{(n-1)\Delta t}^{n\Delta t}
        \phi^{(*l)}(t-s)dtds\\
        &\quad\leq
        \alpha^{l'-1}\|\phi\|_\infty\Delta t
        \dfrac{\|\phi\|_\infty^l}{(l+1)!}
        (\Delta t)^{l+1}\\
        &\quad=
        \alpha^{l'-1}
        \dfrac{\|\phi\|_\infty^{l+1}}{(l+1)!}
        (\Delta t)^{l+2}.
    \end{align*}

    If $p>0$ and $l'=0$, both the integral and the leading term are zero. If $p>0$, $l=0$, and $l'>0$, the difference from the leading term is
    \begin{align*}
        &\left|
        \int_{(n-1)\Delta t}^{n\Delta t}
        \int_{(n+p-1)\Delta t}^{(n+p)\Delta t}
        \phi^{(*l')}(t'-s)dt'ds
        -\Delta t(\tilde{\phi}^{(*l')})_p
        \right|.
    \end{align*}
    Applying Lemma~\ref{lem.binnedpoisson-phi-convolution} with $i=p$ gives
    \[
        \begin{cases}
            O((\Delta t)^{5/2}), & l'=1,\\
            l'\hat{\alpha}^{l'-2}O((\Delta t)^3), & l'\geq2.
        \end{cases}
    \]

    Finally, let $p=0$ and $ll'=0$. If $l>0$ and $l'=0$, Lemma~\ref{lem.binnedpoisson-phi-convolution} with $i=0$ gives an error of $O((\Delta t)^{l+1})$. Similarly, if $l=0$ and $l'>0$, the error is $O((\Delta t)^{l'+1})$. When $l=l'=0$, the equality is exact. Thus, the error is
    \[
        \bm{1}_{l>0}O((\Delta t)^{l+1})
        +\bm{1}_{l'>0}O((\Delta t)^{l'+1}),
    \]
    which proves (ii) and completes the proof.
\end{proof}

We next sum the preceding expansion over all past temporal bins. The following result is the counterpart of Lemma~\ref{lem.phi-sum-evaluation} for the binned Poisson approximation and retains the dependence on $l$ and $l'$ required for summability.

\begin{lem}\label{lem.binnedpoisson-phi-sum-evaluation}
    Let $p,l,l'=0,1,2,\ldots$, and fix any constant $\hat{\alpha}$ satisfying $\alpha<\hat{\alpha}<1$. Then, the following expansion holds as $\Delta t\to0$:
    \begin{align}
        &\sum_{m=-\infty}^n
        \int_{(m-1)\Delta t}^{m\Delta t}
        \left(\int_{(n-1)\Delta t}^{n\Delta t}
        \phi^{(*l)}(t-s)dt\right)
        \left(\int_{(n+p-1)\Delta t}^{(n+p)\Delta t}
        \phi^{(*l')}(t'-s)dt'\right)ds\nonumber\\
        &=\Delta t\sum_{m=-\infty}^n
        (\tilde{\phi}^{(*l)})_{n-m}
        (\tilde{\phi}^{(*l')})_{n+p-m}\nonumber\\
        &\quad+
        \begin{cases}
            \begin{aligned}
                \bm{1}_{l=0}\bigl\{
                &\bm{1}_{l'=1}O((\Delta t)^{5/2})\\
                &+\bm{1}_{l'\geq2}l'\hat{\alpha}^{l'-2}
                O((\Delta t)^3)
                \bigr\},
            \end{aligned}
            & p>0\text{ and }ll'=0,\\[2mm]
            \begin{aligned}
                &\bm{1}_{l>0}
                O\left(
                \dfrac{\|\phi\|_\infty^l}{(l+1)!}
                (\Delta t)^{l+1}
                \right)\\
                &\quad+\bm{1}_{l'>0}
                O\left(
                \dfrac{\|\phi\|_\infty^{l'}}{(l'+1)!}
                (\Delta t)^{l'+1}
                \right),
            \end{aligned}
            & p=0\text{ and }ll'=0,\\[2mm]
            (l+l')\hat{\alpha}^{l+l'-2}
            O((\Delta t)^{5/2}),
            & ll'>0\text{ and }\min\{l,l'\}=1,\\
            (l+l')\hat{\alpha}^{l+l'-2}
            O((\Delta t)^3),
            & l,l'\geq2.
        \end{cases}
        \label{eq.binnedpoisson-phi-sum-evaluation}
    \end{align}
    The constants implicit in the $O$ terms depend only on $\phi$ and the fixed constant $\hat{\alpha}$ and are independent of $n$, $p$, $l$, $l'$, and $\Delta t$.
\end{lem}
\begin{proof}
    We first consider the case $ll'=0$. By Lemma~\ref{lem.binnedpoisson-decompose-time}, the expansion is exact for every $m<n$. Therefore, only the term $m=n$ contributes to the error. Applying the second part of Lemma~\ref{lem.binnedpoisson-decompose-time}, we obtain
    \[
        \bm{1}_{l=0}\left\{
        \bm{1}_{l'=1}O((\Delta t)^{5/2})
        +\bm{1}_{l'\geq2}l'\hat{\alpha}^{l'-2}
        O((\Delta t)^3)
        \right\}
    \]
    when $p>0$, and
    \[
        \bm{1}_{l>0}
        O\left(
        \dfrac{\|\phi\|_\infty^l}{(l+1)!}
        (\Delta t)^{l+1}
        \right)
        +
        \bm{1}_{l'>0}
        O\left(
        \dfrac{\|\phi\|_\infty^{l'}}{(l'+1)!}
        (\Delta t)^{l'+1}
        \right)
    \]
    when $p=0$. This proves the first two cases.

    We next suppose that $ll'>0$. Since
    \[
        \sum_{r=0}^{\infty}
        (\tilde{\phi}^{(*l)})_r
        =
        \left(
        \sum_{r=0}^{\infty}\tilde{\phi}_r
        \right)^l
        \leq\hat{\alpha}^l
    \]
    for sufficiently small $\Delta t$, we have
    \[
        \sum_{m=-\infty}^{n-1}
        (\tilde{\phi}^{(*l)})_{n-m}
        \leq\hat{\alpha}^l.
    \]
    Moreover, using the definition of $B_{l',m,n+p}$ in the proof of Lemma~\ref{lem.binnedpoisson-decompose-time},
    \begin{align*}
        \sum_{m=-\infty}^{n-1}B_{l',m,n+p}
        &\leq
        \dfrac{1}{\Delta t}
        \int_{(n+p-1)\Delta t}^{(n+p)\Delta t}
        \int_{-\infty}^{t'}
        \phi^{(*l')}(t'-s')ds'dt'\\
        &=\|\phi^{(*l')}\|_1
        =\alpha^{l'}.
    \end{align*}

    Suppose first that $\min\{l,l'\}=1$. Summing the first two error terms in Eq.~\eqref{eq.binnedpoisson-decompose-time-error} over $m<n$ and using the preceding bounds gives
    \begin{align*}
        &\left(
        \bm{1}_{l'=1}\hat{\alpha}^l
        +\bm{1}_{l=1}\alpha^{l'}
        \right)O((\Delta t)^{5/2})\\
        &\quad+
        \left(
        \bm{1}_{l'\geq2}l'\hat{\alpha}^{l+l'-2}
        +\bm{1}_{l\geq2}l\hat{\alpha}^{l-2}\alpha^{l'}
        \right)O((\Delta t)^3)\\
        &\quad\leq
        (l+l')\hat{\alpha}^{l+l'-2}
        O((\Delta t)^{5/2}).
    \end{align*}

    By the same Cauchy--Schwarz argument as in the proof of Lemma~\ref{lem.phi-sum-evaluation}, the sum of the error terms containing $D_{l,m,n}$ is bounded by
    \[
        O((\Delta t)^4)
        \|(\phi^{(*l)})'\|_2
        \|(\phi^{(*l')})'\|_2.
    \]
    If $l=1$ and $l'\geq2$, this is
    \[
        \alpha^{l'-2}O((\Delta t)^4),
    \]
    and the same evaluation holds with $l$ and $l'$ interchanged. For $(l,l')=(1,1)$, it is $O((\Delta t)^4)$. These terms, together with the error corresponding to $m=n$, are absorbed into
    \[
        (l+l')\hat{\alpha}^{l+l'-2}
        O((\Delta t)^{5/2}).
    \]
    This proves the third case.

    Finally, suppose that $l,l'\geq2$. The first two error terms summed over $m<n$ are bounded by
    \begin{align*}
        &\left(
        l'\hat{\alpha}^{l+l'-2}
        +l\hat{\alpha}^{l-2}\alpha^{l'}
        \right)O((\Delta t)^3)\\
        &\quad\leq
        (l+l')\hat{\alpha}^{l+l'-2}
        O((\Delta t)^3).
    \end{align*}
    The terms containing $D_{l,m,n}$ are bounded by
\[
    \alpha^{l+l'-4}O((\Delta t)^4),
\]
where we used the bounds on $\|(\phi^{(*l)})'\|_2$ and $\|(\phi^{(*l')})'\|_2$ from the proof of Lemma~\ref{lem.phi-sum-evaluation}. When combined with the preceding error terms, both this contribution and the error corresponding to $m=n$ are absorbed into
\[
    (l+l')\hat{\alpha}^{l+l'-2}O((\Delta t)^3),
\]
    because
    \[
        \alpha^{l'-1}
        \dfrac{\|\phi\|_\infty^{l+1}}{(l+1)!}
        (\Delta t)^{l+2}
        \leq
        C\hat{\alpha}^{l+l'-2}(\Delta t)^3
    \]
    for sufficiently small $\Delta t$ and some constant $C>0$ independent of $l$ and $l'$. This proves the fourth case and completes the proof.
\end{proof}

We now combine the temporal and spatial estimates to prove Eq.~\eqref{eq.binnedpoisson-autocovariance-asymptotic}.

\begin{proof}
Choose $|S|_0>0$ sufficiently small that
\[
    \hat{\alpha}\left(1+C|S|_0^{1/q}\right)<1,
\]
where $C$ is sufficiently large for the spatial estimates above. Then
\[
    \sum_{l=1}^{\infty}l\hat{\alpha}^{l-1}
    \left(1+C|S|^{1/q}\right)^{l-1}<\infty
\]
uniformly for $|S|\leq|S|_0$. The same conclusion holds with any fixed polynomial in $l$ in place of $l$. Thus, the temporal and spatial error estimates can be summed over $l$ and $l'$, with constants independent of sufficiently small $\Delta t$ and $|S|$.

We first record the counterparts of the estimates in Lemmas~\ref{lem.Psi-evaluation} and~\ref{lem.Psi-sum-evaluation}. Since $\tilde{\phi}_0=0$, we have $\tilde{\Psi}_0=I$. Moreover,
\[
    (\tilde{\phi}^{(*l)})_r
    \leq\Delta t\|\phi\|_\infty\hat{\alpha}^{l-1},
    \qquad l\geq1,\quad r\geq1.
\]
Using the column-sum bound for $\tilde{F}$, we also obtain
\[
    [\tilde{F}^l]_{ij}
    \leq\|f\|_\infty|S_i|\left(1+C|S|^{1/q}\right)^{l-1},
    \qquad l\geq1.
\]
Consequently, the same arguments as in Lemmas~\ref{lem.Psi-evaluation} and~\ref{lem.Psi-sum-evaluation} give
\[
    [\tilde{\Psi}_r]_{ij}=O(\Delta t|S_i|),\qquad r\geq1,
\]
and
\[
    \sum_{r=0}^{\infty}\sum_{k=1}^d[\tilde{\Psi}_r]_{ik}|S_k|=O(|S_i|).
\]
Separating the term $r=0$ when $p=0$, it follows that
\[
    \sum_{r=0}^{\infty}\sum_{k=1}^d
    [\tilde{\Psi}_r]_{ik}[\tilde{\Psi}_{r+p}]_{jk}|S_k|
    =\delta_{p,0}\delta_{i,j}|S_i|+O(\Delta t|S_i||S_j|).
\]

We next compare the covariance representations of the binned Hawkes process and the binned Poisson approximation. Applying the spatial decomposition in Lemma~\ref{lem.decompose-space} with the cellwise distribution concentrated at $\bm{r}_k$, and then using Lemmas~\ref{lem.binnedpoisson-phi-sum-evaluation} and~\ref{lem.binnedpoisson-Fnij}, we obtain the same decomposition as in the proof of Theorem~\ref{thm.cov} in Appendix~\ref{appendix.2nd-order}.

The spatial approximation errors multiplied by the leading temporal terms contribute
\[
    O((\Delta t)^2|S_i||S_j||S|^{1/q}).
\]
Here, the factors depending on $l$ and $l'$ are summable by the uniform bounds established above. The temporal approximation errors multiplied by the spatial factors contribute
\[
    \begin{cases}
        O((\Delta t)^2|S_i||S_j|),&p=0,\\
        O((\Delta t)^{5/2}|S_i||S_j|),&p>0.
    \end{cases}
\]
For $p=0$, the terms with exactly one of $l$ and $l'$ equal to zero give the order $(\Delta t)^2$ contribution. The remaining temporal errors are absorbed into this order. For $p>0$, the terms involving a first-order temporal convolution give the order $(\Delta t)^{5/2}$ contribution, and the higher-order terms are absorbed into it.

Therefore,
\begin{align*}
    &\operatorname{Cov}(X^{(\mathrm{hw})}_{i,n},X^{(\mathrm{hw})}_{j,n+p})\\
    &=\sum_{r=0}^{\infty}\sum_{k=1}^d
    [\tilde{\Psi}_r]_{ik}[\tilde{\Psi}_{r+p}]_{jk}
    \mathbb{E}[X^{(\mathrm{hw})}_{k,n}]\\
    &\quad+O((\Delta t)^2|S_i||S_j||S|^{1/q})
    +\begin{cases}
        O((\Delta t)^2|S_i||S_j|),&p=0,\\
        O((\Delta t)^{5/2}|S_i||S_j|),&p>0.
    \end{cases}
\end{align*}

Finally, the stationary mean estimate gives
\[
    \mathbb{E}[X^{(\mathrm{po})}_{k,n}]
    -\mathbb{E}[X^{(\mathrm{hw})}_{k,n}]
    =O(\Delta t|S_k||S|^{1/q})+O((\Delta t)^2|S_k|).
\]
Hence, replacing the stationary mean in the preceding covariance expression contributes
\begin{align*}
    &\delta_{p,0}\delta_{i,j}
    \left\{O(\Delta t|S_i||S|^{1/q})+O((\Delta t)^2|S_i|)\right\}\\
    &\quad+O((\Delta t)^2|S_i||S_j||S|^{1/q})
    +O((\Delta t)^3|S_i||S_j|).
\end{align*}
After this replacement, the leading term is precisely the covariance of the binned Poisson approximation. Combining the error terms, we conclude that
\begin{align*}
    &\operatorname{Cov}(X^{(\mathrm{po})}_{i,n},X^{(\mathrm{po})}_{j,n+p})
    -\operatorname{Cov}(X^{(\mathrm{hw})}_{i,n},X^{(\mathrm{hw})}_{j,n+p})\\
    &=
    \begin{cases}
        O((\Delta t)^2|S_i|^2)+O(\Delta t|S_i||S|^{1/q})
        +O((\Delta t)^2|S_i|),&p=0,\ i=j,\\
        O((\Delta t)^2|S_i||S_j|),&p=0,\ i\neq j,\\
        O((\Delta t)^{5/2}|S_i||S_j|)
        +O((\Delta t)^2|S_i||S_j||S|^{1/q}),&p>0.
    \end{cases}
\end{align*}
This proves Eq.~\eqref{eq.binnedpoisson-autocovariance-asymptotic}.
\end{proof}
\end{appendices}

\bibliographystyle{unsrt}
\bibliography{ref}

\end{document}